%% file: main_Box_Filtration.tex
\documentclass[11pt]{article}
\usepackage{amsmath,amsfonts,amssymb,amsthm,array}
\usepackage{appendix}
\usepackage{commath}
\usepackage{comment}
\usepackage{bm,color,setspace,graphicx,times}
\usepackage{wrapfig,subcaption}
\usepackage{algorithm}
\usepackage[noend]{algpseudocode}
\usepackage{mathtools}
\usepackage[inline]{enumitem}
\usepackage{relsize}
\usepackage[margin=1in]{geometry}
\usepackage[mathscr]{eucal}
\graphicspath{{Figs/}}

\theoremstyle{plain}
\newtheorem{thm}{Theorem}[section]
\newtheorem{exm}[thm]{Example}
\newtheorem{lem}[thm]{Lemma}
\newtheorem{prop}[thm]{Proposition}
\newtheorem{cor}[thm]{Corollary}
\newtheorem{defn}[thm]{Definition}

\newtheorem{claim}[thm]{Claim}
\newtheorem{asmn}[thm]{Assumption}
\newtheorem{proper}[thm]{Property}
\newtheorem{con}[thm]{Conjecture}
\newtheorem{rem}[thm]{Remark}

\newcommand{\bbN}{\mathbb{N}}

\newcommand{\bbR}{\mathbb{R}}
\newcommand{\bbS}{\mathbb{S}}
\newcommand{\bbT}{\mathbb{T}}
\newcommand{\bbU}{\mathbb{U}}
\newcommand{\bbV}{\mathbb{V}}

\newcommand{\bbZ}{\mathbb{Z}}

\newcommand{\clC}{\mathcal{C}}

\newcommand{\clH}{\mathcal{H}}
\newcommand{\clI}{\mathcal{I}}
\newcommand{\clJ}{\mathcal{J}}
\newcommand{\clK}{\mathcal{K}}
\newcommand{\clL}{\mathcal{L}}

\newcommand{\clP}{\mathcal{P}}
\newcommand{\clQ}{\mathcal{Q}}

\newcommand{\clS}{\mathcal{S}}

\newcommand{\clU}{\mathcal{U}}

\newcommand{\srS}{\mathscr{S}}

\newcommand{\tN}{\tilde{N}}
\newcommand{\tV}{\tilde{V}}
\newcommand{\bV}{\bar{V}}
\newcommand{\bX}{\bar{X}}
\newcommand{\bY}{\bar{Y}}
\newcommand{\tM}{\tilde{M}}

\newcommand{\tl}{\tilde{l}}
\newcommand{\tc}{\tilde{c}}
\newcommand{\tu}{\tilde{u}}

\newcommand{\tpi}{\tilde{\pi}}
\newcommand{\talpha}{\tilde{\alpha}}
\newcommand{\halpha}{\hat{\alpha}}
\newcommand{\hpi}{\hat{\pi}}
\newcommand{\balpha}{\bar{\alpha}}
\newcommand{\bpi}{\bar{\pi}}
\newcommand{\bgl}{\bar{\gl}} 

\newcommand{\dew}{w_{\delta}}

\newcommand{\Vy}{V(y)}

\newcommand{\bM}{\bar{M}}

\newcommand{\Mde}{M_{\delta}}

\newcommand{\sol}{\operatorname{Sol}}
\newcommand{\psol}{\overline{\operatorname{Sol}}}
\newcommand{\pC}{\overline{C}}

\DeclarePairedDelimiter{\ceil}{\lceil}{\rceil}
\DeclarePairedDelimiter{\floor}{\lfloor}{\rfloor}

\DeclarePairedDelimiter{\iceil}{\upharpoonright}{\upharpoonleft}

\usepackage{hyperref}

\usepackage[nameinlink,noabbrev,capitalize]{cleveref} 
\crefalias{subequation}{equation}
\crefalias{thm}{theorem}

\let\oldlemma\lem
\renewcommand{\lem}{%
	\crefalias{thm}{lem}
	\oldlemma}
\Crefname{lem}{Lemma}{Lemmas}

\let\olddefn\defn
\renewcommand{\defn}{%
	\crefalias{thm}{defn}
	\olddefn}
\Crefname{defn}{Definition}{Definitions}

\let\oldrem\rem
\renewcommand{\rem}{%
	\crefalias{thm}{rem}
	\oldrem}
\Crefname{rem}{Remark}{Remarks}

\let\oldcor\cor
\renewcommand{\cor}{%
	\crefalias{thm}{cor}
	\oldcor}
\Crefname{cor}{Corollary}{Corollaries}

\let\oldclaim\claim
\renewcommand{\claim}{%
	\crefalias{thm}{claim}
	\oldclaim}
\Crefname{claim}{Claim}{Claims}

\let\oldprop\prop
\renewcommand{\prop}{%
	\crefalias{thm}{prop}
	\oldprop}
\Crefname{prop}{Proposition}{Propositions}

\let\oldcon\con
\renewcommand{\con}{%
	\crefalias{thm}{con}
	\oldcon}
\Crefname{con}{Conjecture}{Conjectures}

\let\oldasmn\asmn
\renewcommand{\asmn}{%
	\crefalias{thm}{asmn}
	\oldasmn}
\Crefname{asmn}{Assumption}{Assumptions}

\let\oldexm\exm
\renewcommand{\exm}{%
	\crefalias{thm}{exm}
	\oldexm}
\Crefname{exm}{Example}{Examples}

\let\oldproper\proper
\renewcommand{\proper}{%
	\crefalias{thm}{proper}
	\oldproper}
\Crefname{proper}{Property}{Properties}

\newcommand{\add}[1]{#1}
\usepackage[normalem]{ulem}
\newcommand{\delete}[1]{}
\newcommand{\mdelete}[1]{}
\definecolor{darkgrn}{rgb}{0, 0.8, 0}

\newcommand{\R}{ {\mathbb R} }
\newcommand{\Z}{ {\mathbb Z} }

\newcommand{\hV}{\hat{V}}

\newcommand{\gl}{\lambda} 

\DeclareMathOperator*{\Nrv}{Nrv}
\DeclareMathOperator*{\BoxF}{Box}
\DeclareMathOperator*{\eBoxF}{\epsilon-Box}
\DeclareMathOperator*{\pBoxF}{\overline{Box}}

\newcommand{\frcpt}{\operatorname{frac}}

\newcommand{\dis}{\operatorname{dis}}
\newcommand{\dGH}{\operatorname{d_{GH}}}
\newcommand{\dgm}{\operatorname{dgm}}

\newcommand{\lpc}{\mathfrak{L}} 

\usepackage{authblk}

\newcommand{\dmn}{n} 

\title{\Huge Box Filtration}
\author[1]{Enrique Alvarado\thanks{enrique3@iastate.edu}}
\author[2]{Prashant Gupta\thanks{{\bfseries Corresponding author;}  prashant.gupta@cuanschutz.edu}}
\author[3]{Bala Krishnamoorthy\thanks{kbala@wsu.edu}}
\affil[1]{Department of Mathematics, Iowa State University, Ames, US}
\affil[2]{Department of Biomedical Informatics, University of Colorado Anschutz Medical Campus, Aurora, CO, USA}
\affil[3]{Department of Mathematics and Statistics, Washington State University}

\begin{document}

 \maketitle

 \input{abstract}

 \clearpage

 \input{intro}

 \input{construct}

 \input{stable}

 \input{examples}

 \input{mapper}

 \input{disc}

 \clearpage

\input{main_Box_Filtration.bbltex}
 \input{app}


\end{document}

%% file: abstract.tex
\begin{abstract}

  
  We define a new framework that unifies the filtration and the mapper approaches from topological data analysis, and present efficient algorithms to compute it.
  Termed the \emph{box filtration} of a point cloud data (PCD), we grow boxes (hyperrectangles) that are not necessarily centered at each point (in place of balls centered at each point as done by most current filtrations).
  We grow the boxes non-uniformly and asymmetrically in different dimensions based on the distribution of points.
  We present two approaches to handle the boxes: a point cover where each point is assigned its own box at start, and a pixel cover that works with a pixelization of the space of the PCD.
  Any box cover in either setting automatically gives a mapper of the PCD.
  We show that the persistence diagrams generated by the box filtration using both point and pixel covers satisfy the classical stability based on the Gromov-Hausdorff distance.
  Using boxes, rather than Euclidean balls, also implies that the box filtration is identical for pairwise or higher order intersections whereas the Vietoris-Rips (VR) and \v{C}ech filtration are not the same.
  
  Growth in each dimension is computed by solving a linear program that optimizes a cost functional balancing the cost of expansion and benefit of including more points in the box.
  The box filtration algorithm runs in $O(m|\mathcal{U}(0)|\log(m \dmn \pi) \lpc(q))$ time, where $m$ is number of steps of increments considered for growing the box, $|\mathcal{U}(0)|$ is the number of boxes in the initial cover (at most the number of points), $\pi$ is the step length by which each box dimension is incremented, each linear program is solved in $O(\lpc(q))$ time, $\dmn$ is the dimension of the PCD, and $q = n \times |X|$.
  We also present a faster algorithm that runs in $O(m|\mathcal{U}(0)|k \lpc(q))$ where $k$ is the number of steps allowed to find the optimal box.
  We demonstrate through  multiple examples that the box filtration can produce more accurate results to summarize the topology of the PCD than VR and distance-to-measure (DTM) filtrations.
  Software for our implementation is available \href{https://github.com/pragup/Box-Filteration/}{here}.
\end{abstract}

%% file: intro.tex
\section{Introduction}

  Persistent homology has matured into a widely used and powerful tool in topological data analysis (TDA) \cite{EdMo2013}.
  A typical TDA pipeline starts with a point cloud data (PCD) $X$ in $\bbR^\dmn$ and uses the default Euclidean metric.
  The fundamental step in such TDA frameworks is the construction of a \emph{filtration}, i.e., a nested sequence of simplicial complexes, built on the PCD.
  Most commonly, the filtration is constructed by growing Euclidean balls centered at each point in $X$, and is termed the Vietoris-Rips (VR) filtration.
  The resulting VR persistence diagram (PD) of $X$ summarizes its topology across multiple scales, and could have implications for the application generating the data \cite{Ca2009,CoSh2009}.
  Going one step further, we can compare two different PCDs $X$ and $Y$  by computing the bottleneck distance between their PDs \cite[Chap.~5]{ChdeSGlOu2016}.
  This pipeline also satisfies a standard notion of stability \cite{ChdeSOu2014}---the bottleneck distance between their PDs is bounded by twice the Gromov-Hausdorff distance between $X$ and $Y$.
  But outlier points in $X$ can cause a large change in its PD \cite{BlGaMaPa2014}.
  
  \smallskip
  Various approaches have been proposed to tackle the problem of outliers in the context of filtrations.
  They include distance to measure (DTM) class of filtrations \cite{AnChGlIkInTiUm2019, BuChOuSh2016, ChCoMe2011, GuMeMo2011}, approaches using density functions \cite{BoMuTa2017, ChGuOuSk2011, ChFrGuOuSk2013}, and kernel density estimates \cite{PhWaZh2013}.
  Intuitively, DTM filtrations grow the balls as guided by where \emph{measure} is greater, i.e., where there are more points concentrated in the PCD, and hence ignores isolated outlier points.
  More recently, 2-parameter approaches called bifiltrations \cite{CoKeLeOs2021} have been proposed that use both distance and density thresholds as parameters.
  Another bifiltration approach termed localized union of balls \cite{KeSo2023} considers a localized space where the first parameter controls the radius of the ball and the second parameter controls locality of the data.
  But most bifiltration approaches handle the growth in both parameters symmetrically by default.

  \smallskip
  The key factor common to all these approaches is that they grow \emph{balls} centered at points based on various parameters to generate the filtration.
  A ball centered at a point is the natural choice for a symmetric convex body that captures its neighborhood.
  But balls grow symmetrically and uniformly in all directions, and hence the corresponding filtrations could induce a \emph{symmetry bias}.
  This could be undesirable especially when the points may be distributed non-uniformly across a subset of dimensions.
  It is also difficult to control the growth of balls based on the distribution of other points in the PCD.
  Another approach could be to grow ellipses centered at points instead of balls \cite{KaLe2020}, but ellipses are symmetric along each direction with respect to the points just as balls are. 

  \smallskip
  Another TDA approach for characterizing the structure of PCD is the \emph{mapper} \cite{SiMeCa2007}, which starts with a cover of the ambient space containing the PCD.
  Defined as the nerve, i.e., dual complex, of a refined pullback of this cover, the mapper has found increasing use in diverse applications \cite{Lumetal2013}.
  The framework of persistent homology has been employed to prove theoretical stability of mapper constructions \cite{CaOu2017,DeMeWa2016}.
  At the same time, implementing such stable constructions, e.g., the multiscale mapper \cite{DeMeWa2016}, still remains a challenge.
  Most users work with a single mapper (as opposed to a mapper filtration) constructed on an appropriately chosen cover of the ambient space using overlapping hypercubes.
  More generally, the frameworks of mapper and filtrations for persistent homology have been used mostly independent of each other.
  An exception may be the ball mapper construction \cite{Dl2019}, which built a filtration by growing balls centered at a subset of points in the PCD.
  But the ball mapper also suffers from the isotropic nature of the balls.

  \medskip
  Motivated by the hypercube covers used as a default choice in mapper constructions, we posit that building filtrations by growing hyperrectangles, i.e., boxes, non-uniformly in different directions based on the distribution of points in $X$ could better capture its topological features.
  Since boxes are convex just as balls are, their collection satisfies the nerve lemma \cite{EdHa2009}, which states that the simplicial complex defined as the nerve of the boxes has the same homotopy type as the collection.
  Furthermore, boxes can be naturally grown at different rates in different dimensions.

  \subsection{Our Contributions}

  We define a new framework that unifies the filtration and the mapper approaches, and present efficient algorithms to compute it.
  Termed the \emph{box filtration} of a given PCD $X \in \bbR^\dmn$, it is built by growing boxes, i.e., hyperrectangles, as the convex sets covering $X$.
  The boxes are grown non-uniformly and non-symmetrically in different dimensions.
  We present two approaches to handle the boxes: a \emph{point cover} (\cref{sssec:pointcover}) where each point is assigned its own box at start, and a \emph{pixel cover} (\cref{sssec:pixelcover}) that works with a pixelization of the space of the PCD.
  Any box cover in either setting automatically gives a mapper of the PCD.
  We expand each pivot box in each dimension individually.
  This expansion is effected by the minimization of an objective function that balances the cost for growth of the box with the benefit of covering points currently not in the box.
  We build a filtration by repeating this growth step in an \emph{expansion algorithm} which ensures that the new boxes contain the boxes from the previous step.
  Using boxes rather than balls as cover elements provides the nice property that all higher order intersections of the boxes are guaranteed as soon as every pair of them intersect (\cref{lem:boxcechVR}), whereas the VR and \v{C}ech filtrations are not the same.  
  
  We prove that, under mild assumptions, the box filtrations created using both point and pixel covers satisfy the classical stability result based on the Gromov-Hausdorff distance (\cref{thm:pointstabilitytheorem,thm:pointpixelstabilitytheorem}).
  Growth in each dimension is computed by solving a linear program that optimizes a cost functional balancing the cost of expansion and benefit of including more points in the box.
  The box filtration algorithm runs in $O(m|\mathcal{U}(0)|\log(m \dmn \pi)  \lpc(q))$ time, where $m$ is number of steps of increments considered for growing the box, $|\mathcal{U}(0)|$ is the number of boxes in the initial cover (at most the number of points), $\pi$ is the step length by which each box dimension is incremented, each linear program is solved in $O(\lpc(q))$ time, $\dmn$ is the dimension of the PCD, and $q = n \times |X|$.
  We also present a faster algorithm that runs in $O(m|\mathcal{U}(0)|k\lpc(q))$ where $k$ is the number of steps allowed to find the optimal box.

  We demonstrate through multiple examples that the box filtration can produce results that are more resilient to noise and with less symmetry bias than VR and distance-to-measure (DTM) filtrations.
  \cref{fig:spotlight} highlights one of the instances (see \cref{sec:examples} for details).
  The point cloud consists of 100 points sampled from an ellipse along with 50 random points in and around it.
  Box filtration clearly identifies the ellipse for many of its parameter values ($\alpha=0.2, 0.5$--$0.9$) while DTM manages to do so for only one value of its parameter ($m=0.2$), and even in that case, the distinction of the ellipse feature from the noise is not as clear as captured by the box filtration.
  Not surprisingly, VR fails to identify the ellipse.

  \begin{figure}[hb!] 
    \centering
	\begin{subfigure}[t]{1.5in}
          \vspace*{-1.55in}
	  \centering
	  \includegraphics[width=1.5in]{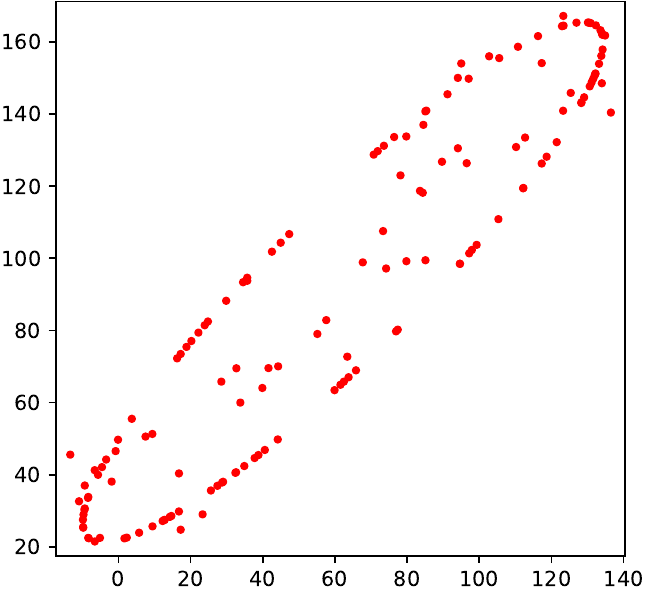}
	\caption*{\label{fig:spotlight2}}
	\end{subfigure}
    \begin{subfigure}[t]{1.5in}
      \centering
      \includegraphics[width=1.5in]{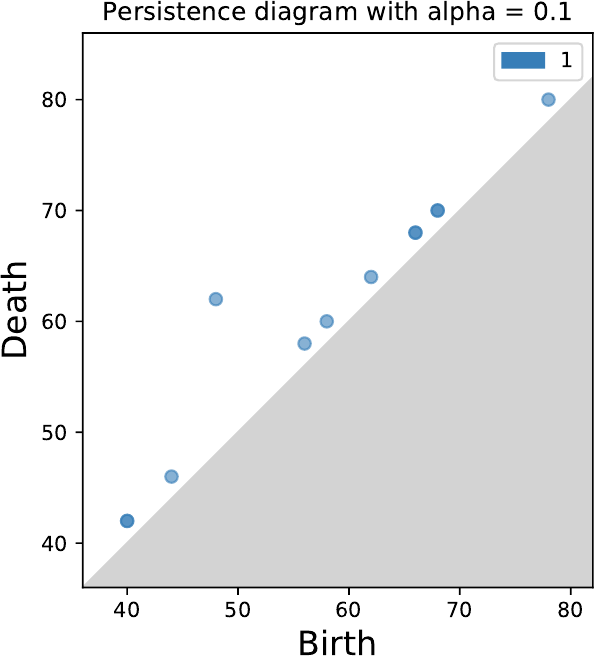}
      \caption*{\label{fig:bfpersistencespotlight2}}
    \end{subfigure}
    \quad
    \begin{subfigure}[t]{1.5in}
    	\centering
    	\includegraphics[width=1.5in]{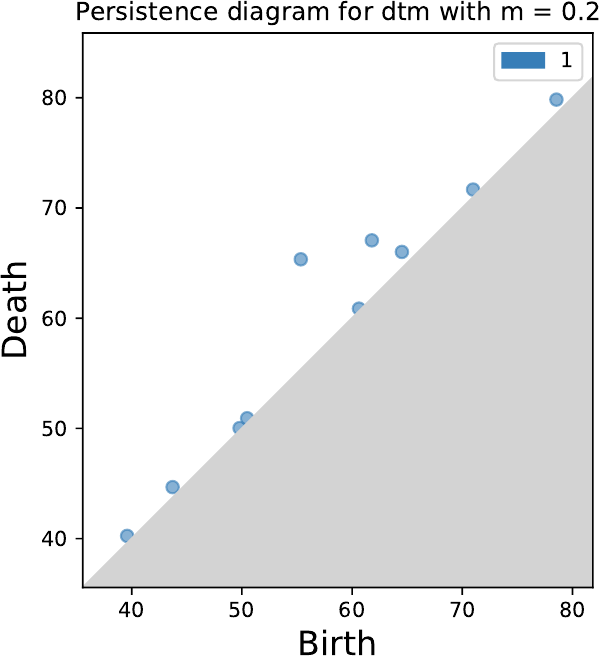}
    	\caption*{\label{fig:dtmpersistencespotlight2}}
    \end{subfigure}
	\quad
	\begin{subfigure}[t]{1.5in}
		\centering
		\includegraphics[width=1.6in]{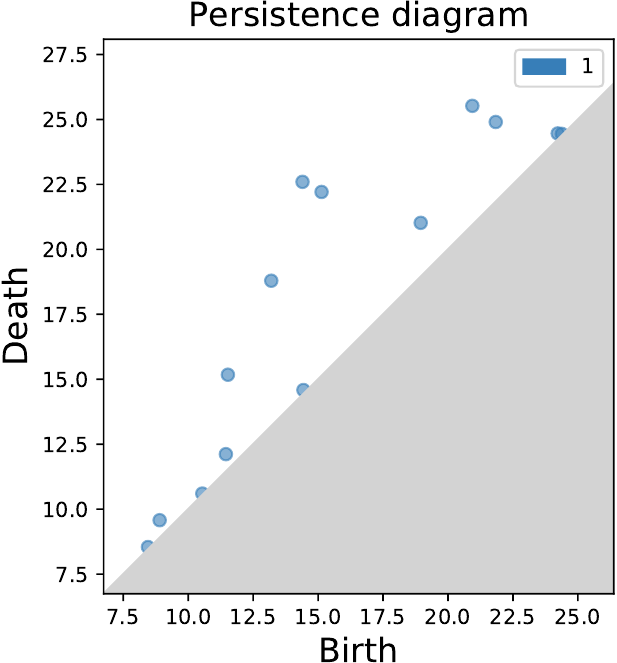}
		\caption*{\label{fig:vrpersistencespotlight2}}
	\end{subfigure}

        \caption{\label{fig:spotlight} Point cloud (Left) consisting of 100 points sampled on an ellipse along with 50 random points in and around it.
          Best persistence diagrams using the box filtration (second figure) and DTM (third figure) along with that using VR filtration (right/fourth figure) are shown.
          Box filtration identifies the ellipse clearly over several values of its parameter $\alpha$ (diagram for $\alpha=0.1$ is shown here), while results are less clear with DTM for most values of its parameter $m$ (best DTM diagram for $m=0.2$ is shown here).
          VR filtration fails to identify the ellipse.
    }
  \end{figure}	

  Since any box cover of $X$ gives a mapper, the box filtration can also function as a mapper framework with stability guarantees.
  But with applicability to large PCDs in mind, we also present an efficient box mapper algorithm (\cref{sec:mapper}) that first applies $k$-means clustering to $X$ and then runs a single box filtration step for one $\pi$ value to generate the mapper.

  We illustrate the box filtration framework on a simple PCD consisting of three points in 2D in \cref{fig:introExample}.
  The box filtration is shown for a few choices of parameters $\alpha$ and $\pi$ (see \cref{sec:construction} for details).
  \vfill
  \begin{figure}[hbp!]
    \hspace*{-0.25in}
      \begin{tabular}{m{4.0cm} m{4.0cm} m{4.0cm} m{4.0cm}}
	\includegraphics[width=1.6in]{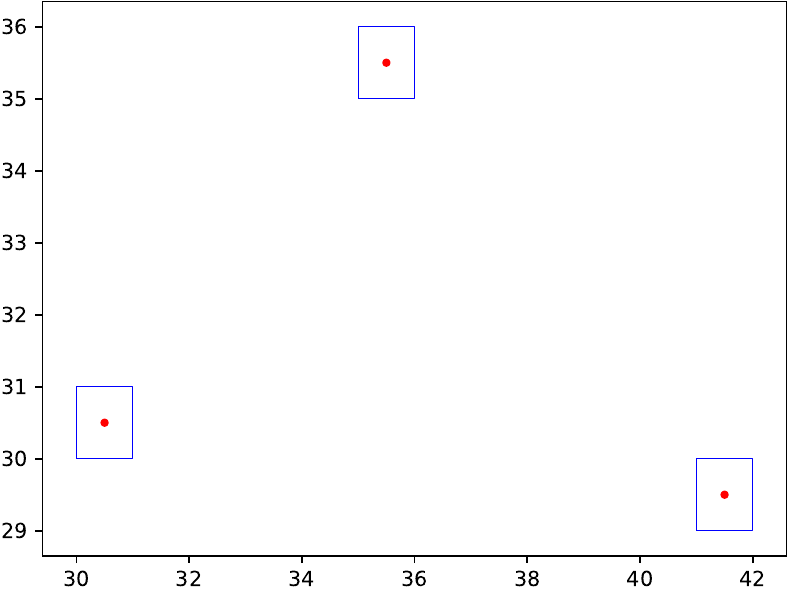}&
	\includegraphics[width=1.6in]{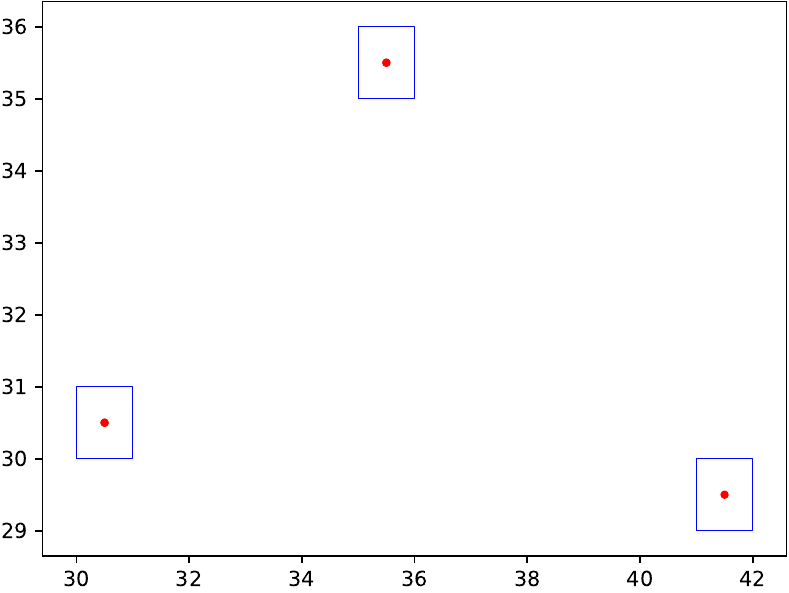}&
	\includegraphics[width=1.6in]{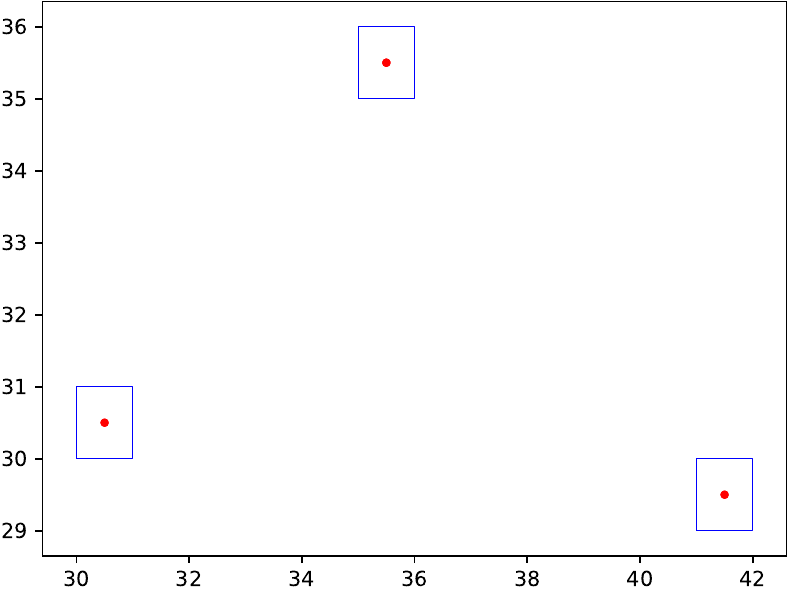}&
	\includegraphics[width=1.6in]{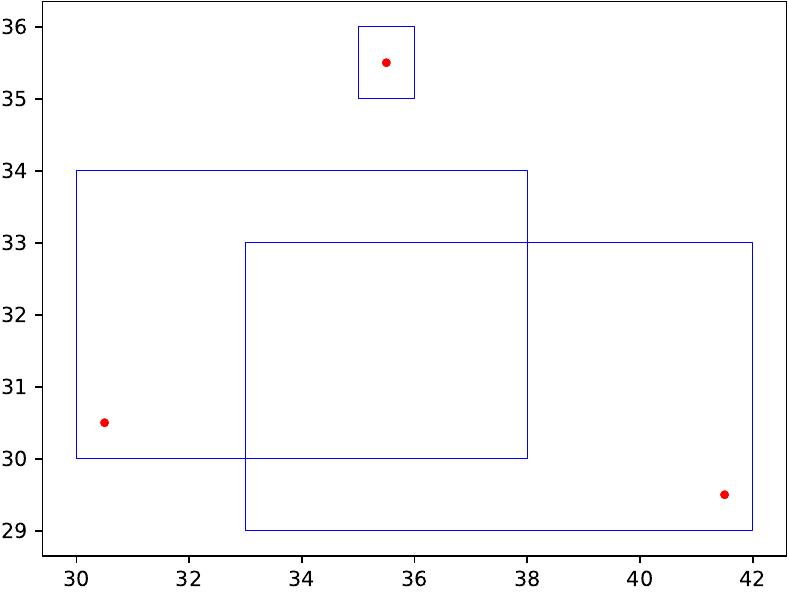}\\
	\includegraphics[width=1.6in]{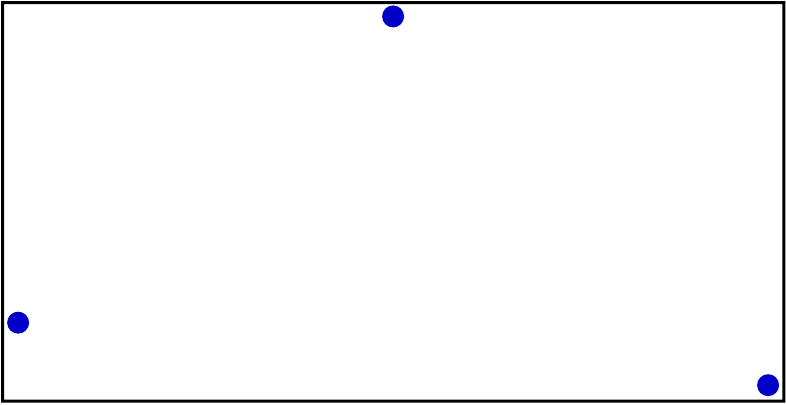}&
	\includegraphics[width=1.6in]{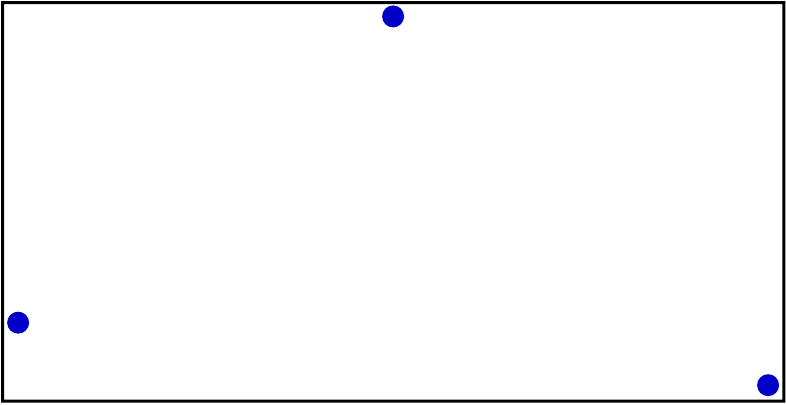}&
	\includegraphics[width=1.6in]{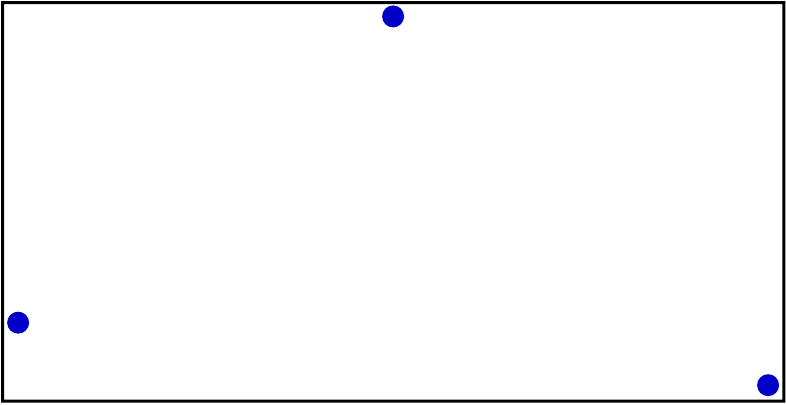}&
	\includegraphics[width=1.6in]{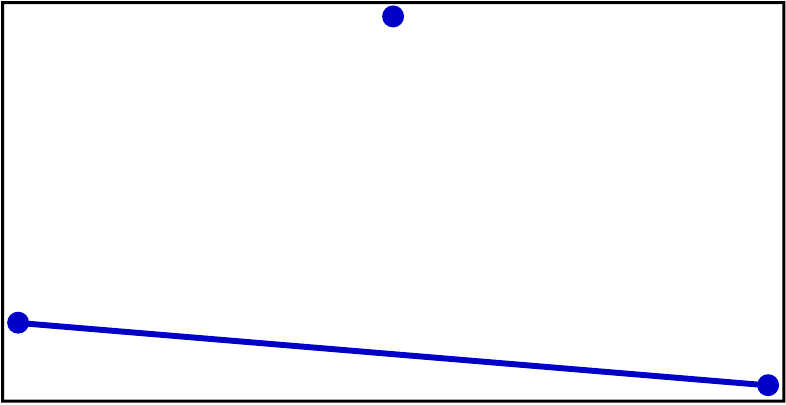}\\
        \\
	\includegraphics[width=1.6in]{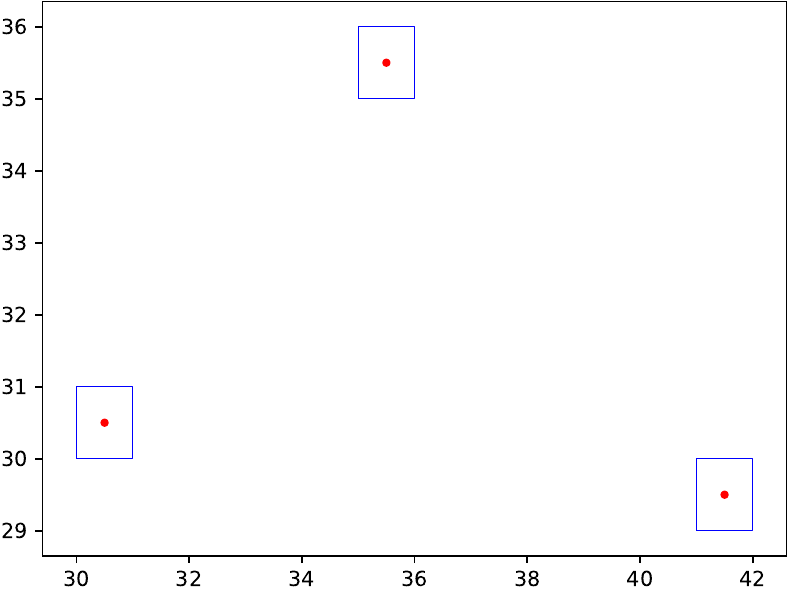}&
	\includegraphics[width=1.6in]{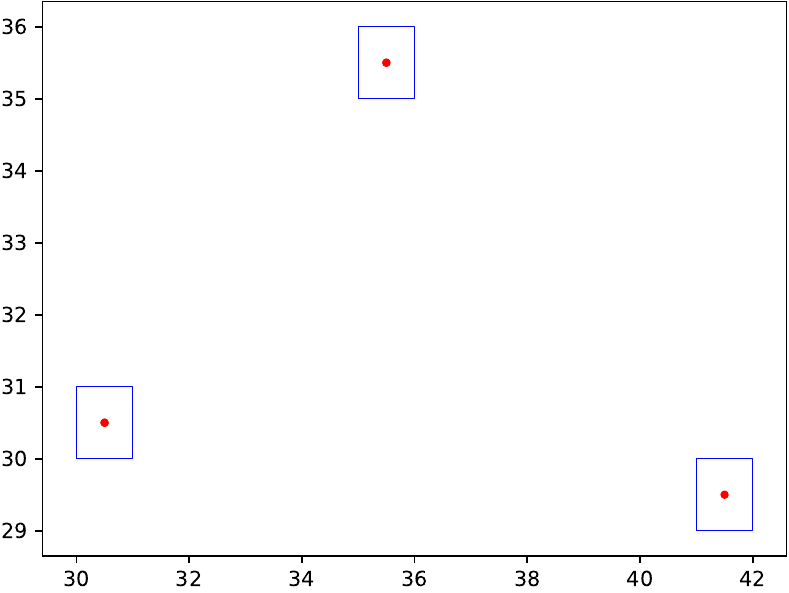}&
	\includegraphics[width=1.6in]{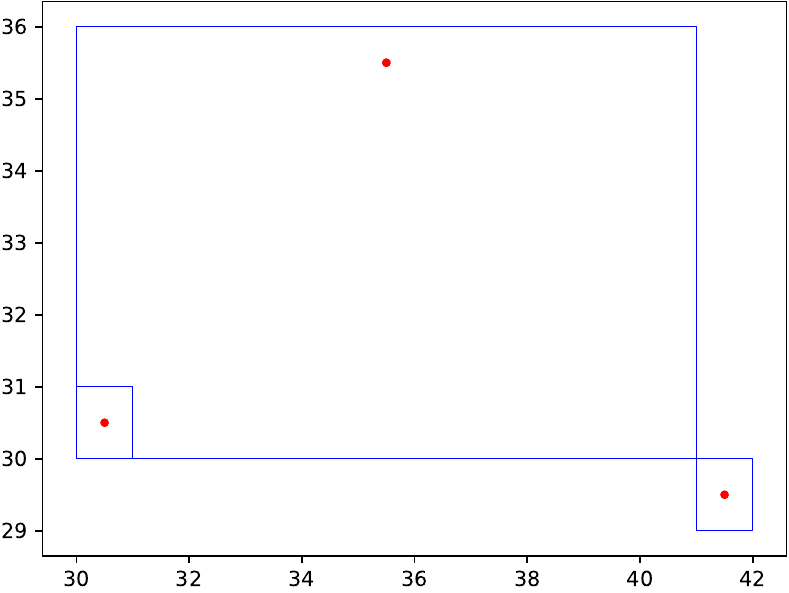}&
	\includegraphics[width=1.6in]{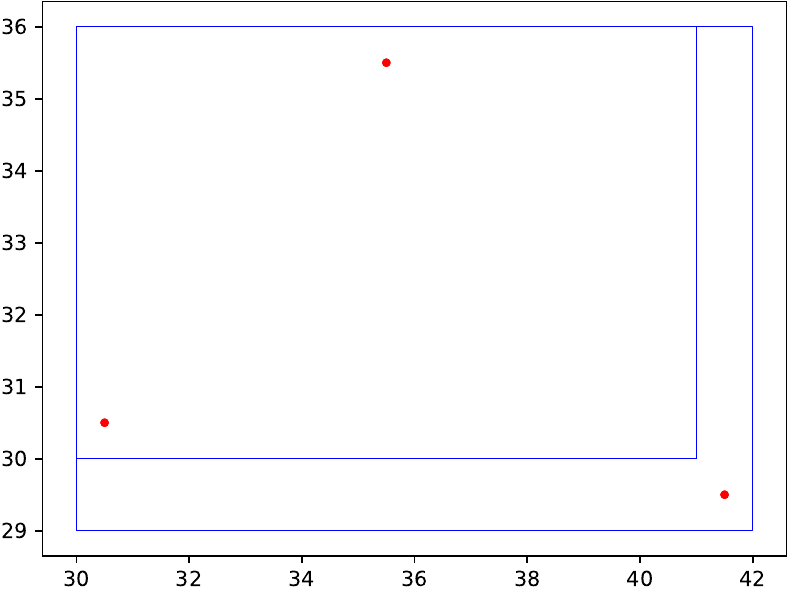}\\
	\includegraphics[width=1.6in]{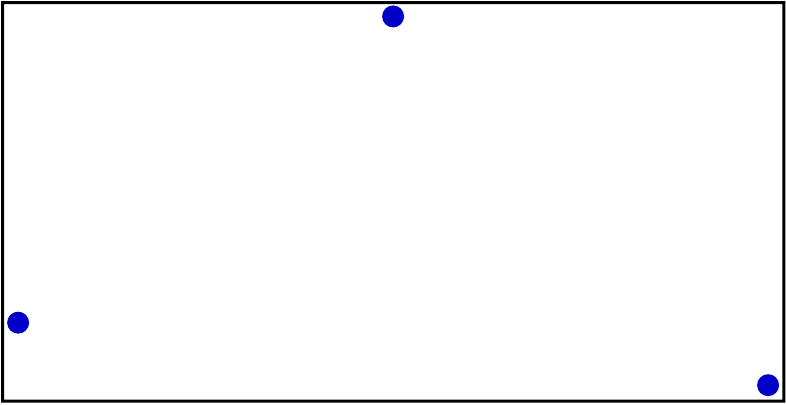}&
	\includegraphics[width=1.6in]{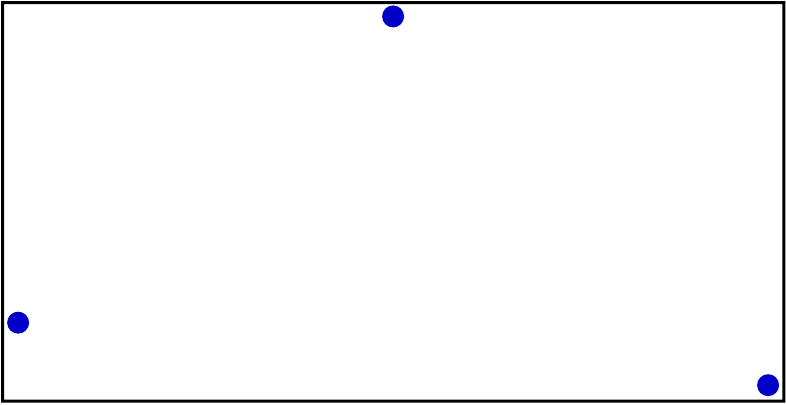}&
	\includegraphics[width=1.6in]{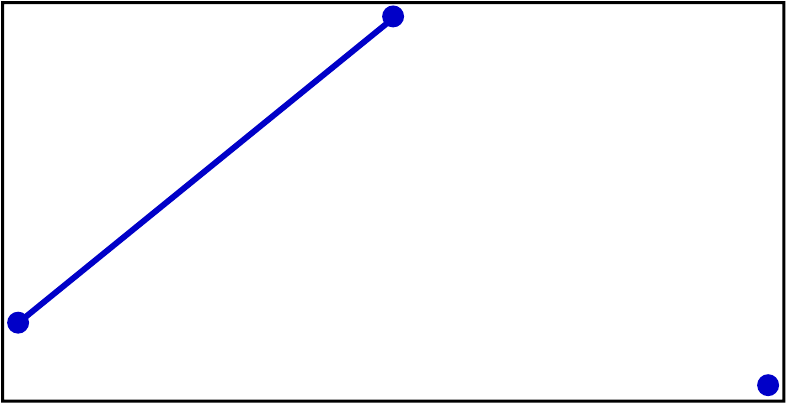}&
	\includegraphics[width=1.6in]{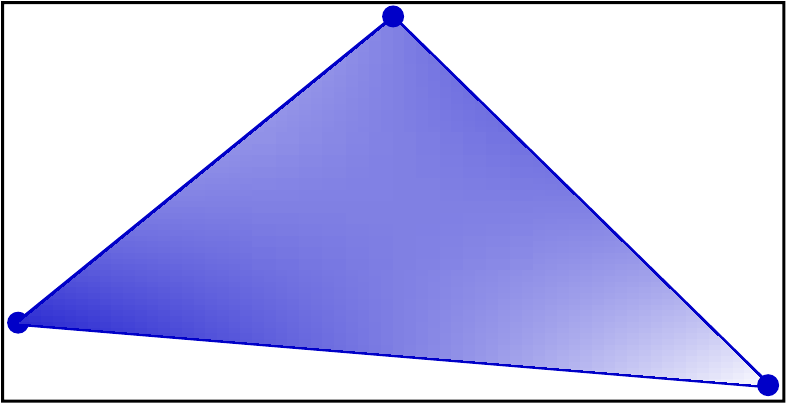}\\
        \\
	\includegraphics[width=1.6in]{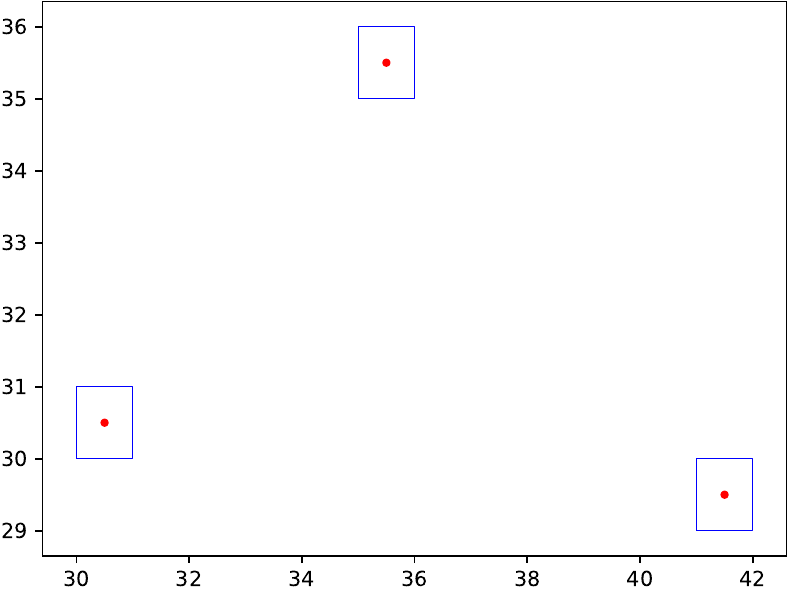}&
	\includegraphics[width=1.6in]{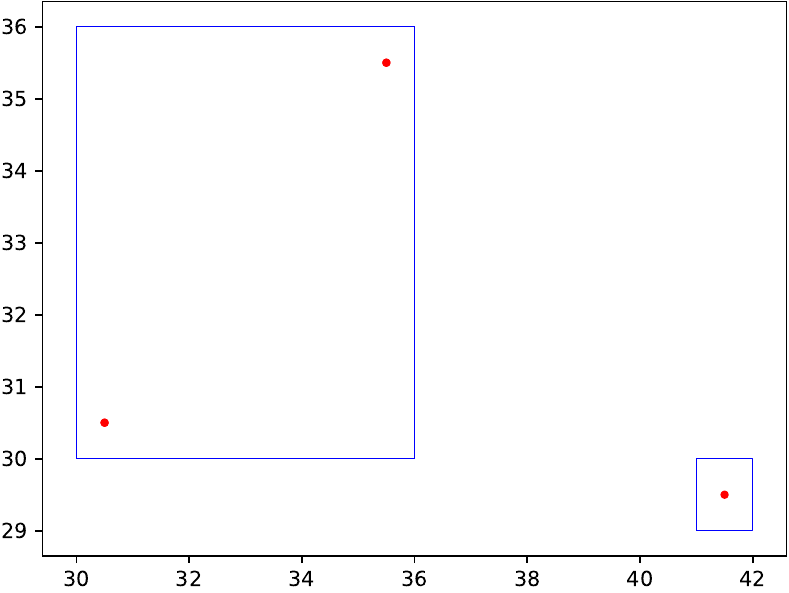}&
	\includegraphics[width=1.6in]{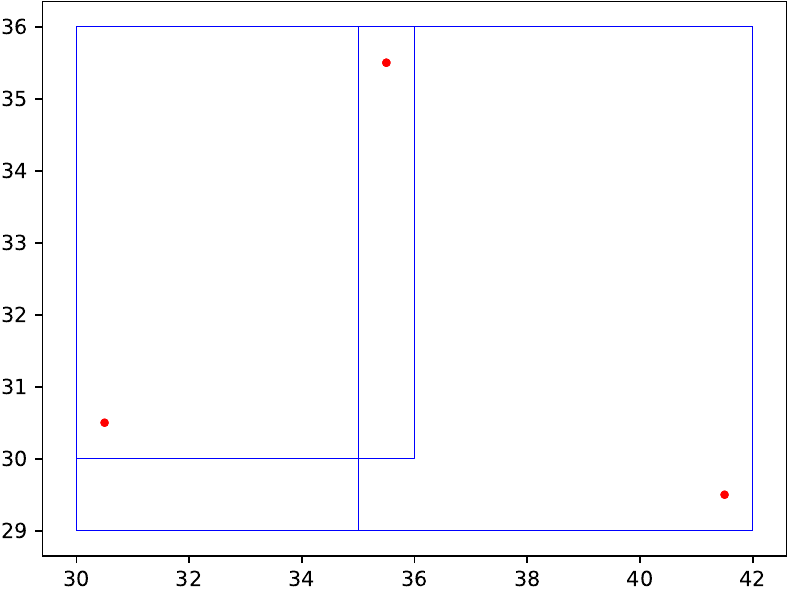}&
	\includegraphics[width=1.6in]{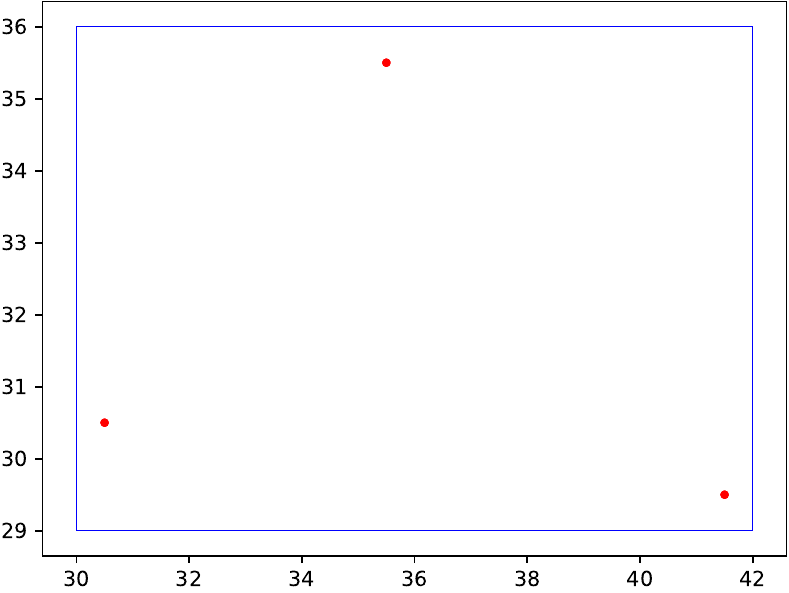}\\
	\includegraphics[width=1.6in]{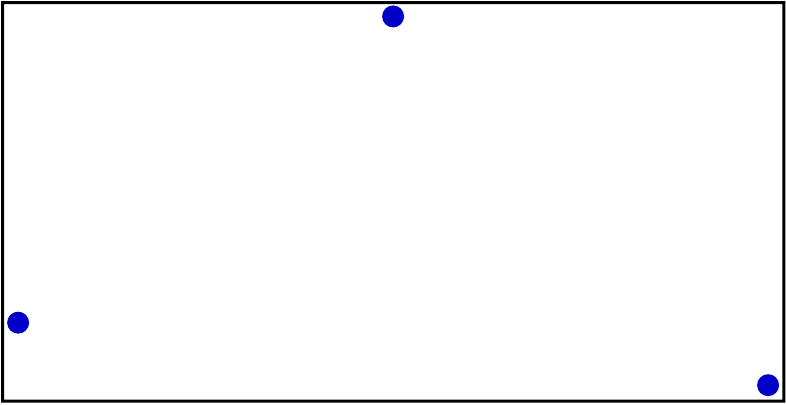}&
	\includegraphics[width=1.6in]{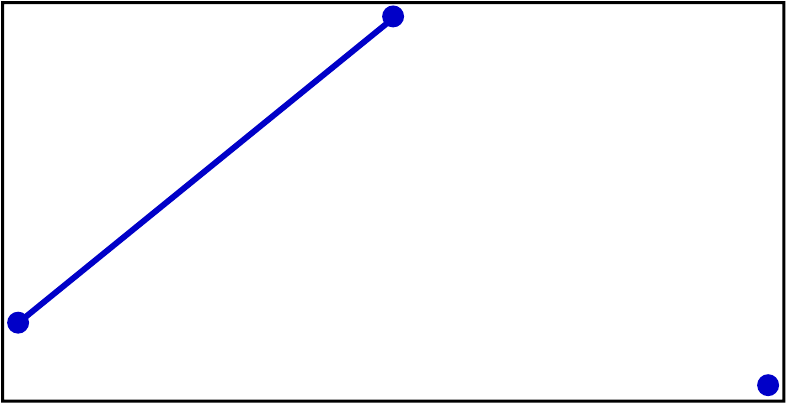}&
	\includegraphics[width=1.6in]{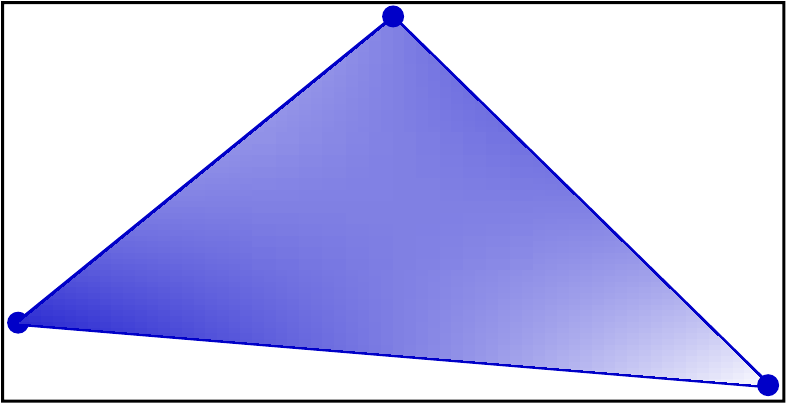}&
	\includegraphics[width=1.6in]{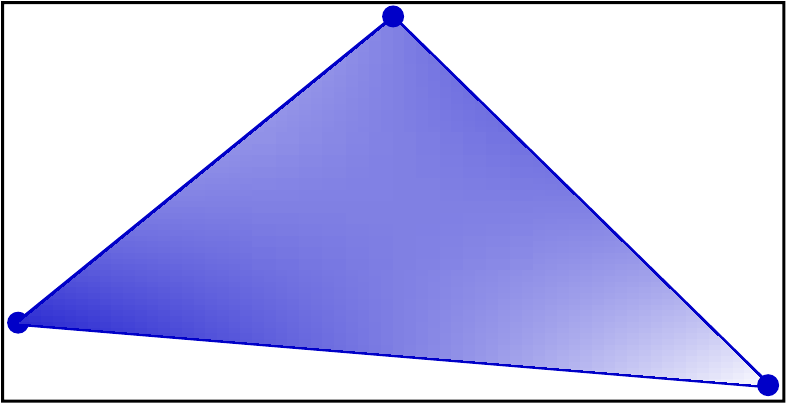}\\
      \end{tabular}
      \caption{\label{fig:introExample} Point cloud ($X$) is shown in red.
        Top row shows covers of $X$ for $\alpha = 0.5$ with $\pi = 5, 6, 11, 12$ from left to right.
        Second row shows their nerves.
        Similarly, middle two rows show covers for $\alpha = 0.6$ and their nerves,
        while bottom two rows show them for $\alpha = 0.7$, for the same set of $\pi$ values.
      }
  \end{figure}

  \noindent Finally, we list notation used in the rest of the paper along with their meanings in \cref{tab:notation}.
  
  \begin{table}[htp!]
    \caption{\label{tab:notation} Notation used, and their explanations.}
    \begin{tabular}{ll}
      \hline
      Notation & Definition/Explanation \\    \hline
      $\sigma$     & unit pixel\\
      $m_{\sigma}, w_{\sigma}, \theta(\sigma)$     & centroid, weight, and number of points in pixel $\sigma$\\
      $w_{x}$ & weight of point $x$  \\
      $\alpha, \pi$     & $\alpha \in [0, 1], \pi \in \bbR_{+}$ are parameters for linear optimization\\
      $\clU(0)$     & initial collection of boxes that covers $X$\\
      $V$	& $V = [l_{1}, u_{1}] \times \dots \times [l_{n}, u_{n}]$ is a box in $\clU(0)$\\
      $V[j\pi]$ & $j^{th}$ expansion of $V$\\
      $B(V, \pi)$		& $\pi$-neighborhood box of $V$: $(l_1 - \pi, u_1 + \pi) \times \dots \times (l_n - \pi, u_n + \pi)$ \\ 
      $C_{\alpha}(V, N)$& cost of $V$ in the neighborhood $N \supseteq V$\\
      $\sol_{\alpha}(V, N)$ & set of optimal solutions for input box $V$ and neighborhood $N\supseteq V$\\ 
      $\clU(j\pi)$     & $j^{th}$ cover $~\forall V \in \clU(0)$\\
      $\Theta(V)$      & Set of pixels with $m_{\sigma} \in V, \theta(\sigma)\neq 0$ \\
      $\psi_{1}(V), \psi_{2}(V), \psi_{3}(V)$ & rounded boxes for given box $V$ (see \cref{eq:psi1,eq:psi2,eq:psi3})\\
      $K(\clU)$		& filtration corresponding to cover $\clU$\\
      \hline
    \end{tabular}
  \end{table}

%% file: construct.tex
\section{Construction}\label{sec:construction}

We start with a formal definition of the basic building block of our construction, namely, the \emph{box}.

\begin{defn}\label{def:box}
  \emph{\bfseries (Box)} A box in $\bbR^\dmn$ is defined as the $\dmn$-fold Cartesian product $[l_1, u_1] \times \dots \times [l_\dmn, u_\dmn]$ where $l_i \leq u_i ~\forall i \in \{1, \dots, \dmn\}$.
\end{defn}

\subsection{Box Cover}\label{ssec:cover}

We first define and study the notion of a \emph{point cover}, with the initial cover consisting of a collection of boxes that has one box centered at the each point in the point cloud and contains no other point. 
This initial setting using hypercubes, i.e., $\ell_\infty$ balls, is analogous to that of VR filtration using Euclidean balls,
but we grow these boxes non-symmetrically to build the box filtration (\cref{ssec:filtration}).
We prove most results---stability, in particular---first for the box filtration constructed on a point cover (\cref{sec:costabnoise}).

At the same time, working with an individual box for each point could be computationally expensive, especially for large PCDs.
To address this challenge, we define the notion of a \emph{pixel cover}, which works with a pixelization of the ambient space containing the PCD.
The initial boxes are now chosen as the pixels that contain point(s) from the PCD,
and the boxes are grown by pixel units based on distances (and other related values) computed at the pixel level.
Working with a coarser pixelization can result in increased computational savings.
Furthermore, we show that the results proven for point covers---stability in particular---also hold for pixel covers.

\subsubsection{Point Cover} \label{sssec:pointcover}

Given the finite PCD $X \in \bbR^\dmn$, we specify the initial cover $\clU(0)$ of its \emph{point cover} to consist of a collection of hypercubes ($\ell_\infty$-balls) or boxes such that each point is located in a single box.
Note that a box could contain more than one point.
Every box $V = [l_1, u_1] \times \dots \times [l_\dmn, u_\dmn]$ in the initial cover $\clU(0)$ is a \emph{pivot box}.
Note that a pivot box could be of lower dimension than $\dmn$ itself, i.e., we could have $l_i = u_i$ for a subset (or all) of $i \in \clI = \{1, \dots, \dmn\}$. 

Our goal is to expand each pivot box in $\clU(0)$ using linear optimization.
We specify the framework of expansion using two parameters $\pi \in \bbR_{+}$ and $\alpha \in [0, 1]$.
We find a sequence of optimal expansions of $V$ in the given $\pi$-neighborhood $N = B(V, \pi) := (l_1 - \pi, u_1 + \pi) \times \dots \times (l_n - \pi, u_n + \pi)$.
\add{To simplify notation, we use $x\in N$ to denote $x \in N \cap X$ throughout the paper.} 
We use parameter $\alpha$ to control the relative weight of two opposing terms in the objective function of the optimization model for expanding $V$ (\cref{eq:ConLPobj}).
Let $\sol_{\alpha}(V, N)$ represent the set of optimal solutions for the input box $V$ and neighborhood $N$.
When evident from context, we denote the solution set simply as $\sol$.

Let $\tV = [\tl_1, \tu_1] \times \dots \times [\tl_n, \tu_n] \supseteq V$ be an an expanded box in the neighborhood $N$.
We let $C_{\alpha}(\tV, N)$ denote the objective (cost) function of the linear program
(when the context is evident, we denote the objective function simply as $C$).
Since our goal is to ultimately cover all points, the points in $N$ that are \emph{not} covered by $\tV$ incur a cost in $C_\alpha(\tV, N)$.
This cost of non-coverage of a point is higher when the point is farther away from $\tV$ and is zero when the point is covered by, i.e., is inside of, $\tV$.
We capture this cost of non-coverage using the \emph{weight} $w_{x} \in \mathbb{R}$ for a point $x \in N$, where
\begin{equation}\label{eq:pointweight}
  w_{x} \leq \min \{\{x_{i} - \tl_i \mid i \in\mathcal{I}\} \cup \{\tu_i - x_{i} \mid i \in \mathcal{I}\}\cup \{0\}\}.
\end{equation}
Note that when $x \not\in \tV$ we have $w_x < 0$.
Furthermore, we get that $x \in \tV$ if and only if
\[\min \{\{x_{i} - \tl_i \mid i \in\mathcal{I}\} \cup \{\tu_i - x_{i} \mid i \in \mathcal{I}\}\} \geq 0,~
\text{ which implies }~w_x = 0.\]

At the same time, we would just grow the box in each dimension to the maximum extent possible if we were minimizing only this cost.
Hence, we include an opposing cost based on the size of $\tV$ as measured by the sum of the lengths of its edges.
We now specify the objective function $C_\alpha(\tV, N)$ along with the constraints of the linear program (LP):


\begin{subequations} \label{eq:ConLP} 
\begin{align}
\min_{\forall \tV\supseteq V} \hspace*{0.1in} C_{\alpha}(\tV, N) = & -\alpha \sum\limits_{x \in N} \, w_{x} ~+~ (1 - \alpha)\sum\limits_{i \in \mathcal{I}}(\tu_i - \tl_i)\hspace*{0.1in}  \label{eq:ConLPobj}\\
\text{s.t.}  \hspace*{0.6in} \nonumber \\[-15pt]
\label{eq:ConLPconstraints}
\begin{split}
& \tu_i ~\geq u_i~ \forall i \in \clI, \\
& \tl_i \,~\leq  \,l_i~~ \forall i \in \clI, \\
& w_{x} \leq~x_{i} - \tl_i ~~ \forall i \in \clI, x \in N, \\
& w_{x} \leq~\tu_i - x_{i} ~~ \forall i \in \clI, x \in N, \\
& w_{x} \leq~0~~ \forall x \in N. 
\end{split}
\end{align}
\end{subequations}

\begin{exm}\label{exm:notunique}
  The optimal solution to the linear program in \cref{eq:ConLP} may not be unique.
\end{exm}
Consider the one-dimensional point cloud $X = \{a, b\}$ as shown in \cref{fig:nonuniquenessfigure1crop}.
Let the initial point cover $\clU(0)$ consist of the single pivot box $V = [a, a]$.
If $\tV \supseteq V$ is some expansion of $V$, then $\tV$ can be optimal only if it expands in the $x$-direction toward $b$.
Let $\tV = [a, x]$ with $x \leq b$, and $N = B(V, \pi = b-a+\delta)$ for small $\delta > 0$.
Then
\[ C_\alpha(\tV, N) = \alpha (b - a - x) + (1 - \alpha) (x - a), ~\text{ which gives that }\]
\[ \dfrac{\partial C_\alpha(\tV, N)}{\partial x} = 1 - 2\alpha.\]
If $\alpha = 0.5$, then $\dfrac{\partial C_\alpha(\tV, N)}{\partial x} = 0~~\forall x \in N$.
Hence the optimal solution for the LP is not unique.     

\begin{figure}[htp!] 
	\centering
	\includegraphics[scale=0.45]{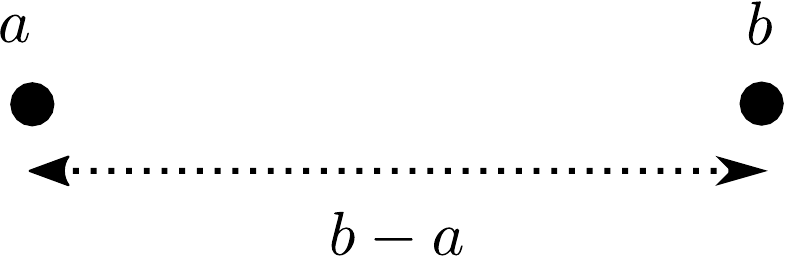}
	\caption{\label{fig:nonuniquenessfigure1crop}
        Example illustrating non-uniqueness of optimal solutions to the box expansion linear program.}
\end{figure}

\begin{defn}\label{def:boxesunion}
  \emph{\bfseries (Union of boxes)} 
  Let $V^{1} = \Pi_{i \in \clI} [l_{i}^{1}, u_{i}^{1}], V^{2} = \Pi_{i \in \clI} [l_{i}^{2}, u_{i}^{2}]$ be two boxes. 
    Then their union is the box $V^{1} \cup V^{2} = \Pi_{i \in \clI} [\hat{l}_i, \hat{u}_i]$ where $\hat{l}_{i}= \min\{l_{i}^{1}, l_{i}^{2}\}$ and $\hat{u}_{i}= \max\{u^{1}_{i}, u^{2}_{i}\}$ for each $i \in \clI$.
\end{defn}

Note that either box in \cref{def:boxesunion} could have zero width along some dimensions, i.e., $l_i=u_i$ for some (or all) $i \in \clI$.
The default setting in which we use the union operation is when we grow a given box in a subset of directions.
By default, we consider growing a box in the standard sequence of dimensions $1, \dots, \dmn$.
We present a useful property of how the cost (objective) function of the box expansion LP in \cref{eq:ConLPobj} changes when a box is expanded in one extra dimension in the standard sequence.

Let $\clI_{i}=\{1, \dots, i\}$ and $V \subseteq \tV$.
Then $S(V, \tV)$ is the ordered sequence whose $i$-th entry is the union of $V$ and the projection of $\tV$ onto the set of directions $\clI_{i}$.
Let $\tc_i$ denote the change in the cost function resulting from the expansion of $\tV$ in the one additional $i$-th direction with respect to the previous, i.e., $(i-1)$-st, entry in the sequence as captured by the projection of $\tV$ onto $V$ in the $i$-th dimension.

\begin{prop}\label{prop:unioncosteq}
  Let $V^{l} \supseteq V,\, V^{k} \supseteq V$, and $\hV = V^{l} \cup V^{k}$ be expansions of a box $V$ such that $V = V^k \cap V^l$ for some neighborhood $N$.
  Let $S(V, V^{l}), \, S(V, V^{k})$, and $S(V, \hV)$ be the sequences with $c^{l}_i, \, c^{k}_i$, and $\hat{c}_i$ being the corresponding changes in the cost function at the $i$-th step of the sequences, respectively.
  Then we have 
  \begin{equation} \label{eq:ciklub}
    \sum\limits_{i \in \clI} \hat{c}_{i} \leq \sum\limits_{i \in \clI} (c_{i}^{k} + c_{i}^{l}).
  \end{equation}	
\end{prop}
\begin{proof}
  Let $P^k$ and $P^l$ be the set of directions along which we expand from $V$ to $ V^k$ and from $V$ to $ V^l$, respectively.
  Then $P^k \cap P^l = \emptyset$ since $V = V^{k} \cap V^{l}$.
  If the $\ell_{\infty}$ distance of some $x \in X$ from $V$ is determined by the $i$-th direction in $P^k$, then there is a reduction in the $\ell_{\infty}$ distance of $x$ as we expand from $V$ to $ V^{k}$.
  Now, the $\ell_{\infty}$ distance of $x$ can be further reduced if we expand from $V^k$ to $V^l$ if it is determined by one of the directions in $P^l$.
  Hence we get that
  \[ -\alpha (w_{x} - \hat{w}_{x}) \, \leq -\alpha (w_{x} - w^k_{x}) + -\alpha (w_{x} - w^l_{x}).\] 
  Considering all $x \in N$, we get
  \[ -\alpha \sum\limits_{x \in N} (w_{x} - \hat{w}_{x}) \leq -\alpha \sum\limits_{x \in N} (w_{x} - w^k_{x}) -\alpha \sum\limits_{x \in N} (w_{x} - w^l_{x}).\]
  For the sequence $S(V, \hat{V})$, we assume without loss of generality that we expand from $V$ to $V^{k}$ to $\hV$.
  We further observe that  $|\hV| - |V| = |V^{k}| - |V| + |V^{l}| - |V|$.
  The inequality in \cref{eq:ciklub} immediately follows.
\end{proof} 
\add{In the following theorem, we show that the set of solutions $\sol_{\alpha}(V, N)$ forms a lattice.}
\begin{thm}\label{thm:unionintersectionpiforx}
  Let $\sol_{\alpha}(V, N)$ be the set of optimal solutions for the input box $V$ for a given neighborhood $N = B(V, \pi)$ with parameters $\alpha \in [0, 1], \pi \in \bbR_{+}$.
  Then the following results hold. 
  \begin{enumerate}
    \item If $V^{l}, V^{k} \in \sol_{\alpha}(V, N)$, then $V^{l} \cap V^{k} \in \sol_{\alpha}(V, N)$.
    \item If $V^{l}, V^{k} \in \sol_{\alpha}(V, N)$, then $V^{l} \cup V^{k} \in \sol_{\alpha}(V, N)$.
  \end{enumerate}
\end{thm}

\begin{proof}
  Suppose $V^{l}, V^{k} \in \sol_{\alpha}(V, N) $ 
  and $\tV = V^{l} \cap V^{k} \not\in \sol_{\alpha}(V, N)$.
  We assume $V^k \neq V^l$, else the result is trivial.
  Consider the two sequences $S(\tV, V^{k})$ and $S(\tV, V^{l})$ with cost changes $c_{i}^{k}$ and $c_{i}^{l}$ at each step of the sequences.
  The total change in the cost for the sequences are then $\sum\nolimits_{i \in \clI} c_{i}^{k} < 0$ and $\sum\nolimits_{i \in \clI} c_{i}^{l} < 0$ since $\tV$ is not an optimal solution, but both $V^k$ and $V^l$ are (and hence, both these sums are equal).
  Let $\hV = V^{l} \cup V^{k}$ (see \cref{def:boxesunion}).
  Then for its sequence $S(\tV, \hat{V})$, the total change in cost is 
  \begin{equation}\label{eq:unionintersectionpiforx1}
    \sum\limits_{i \in \clI_{n}} \hat{c}_{i} \leq \sum\limits_{i \in \clI_{n}} (c_{i}^{k} + c_{i}^{l}) 
  \end{equation}
  by \cref{prop:unioncosteq}, and we already noted that $\sum\nolimits_{i \in \clI}  c_{i}^{k} < 0$ and $\sum\nolimits_{i \in \clI} c_{i}^{l} < 0$.
  But this result contradicts the optimality of $V^{l}$ and $V^{k}$, since the total change in cost for $\hV$ is lower than those for either of these two boxes
  (in fact, \cref{eq:unionintersectionpiforx1} can hold only if $\sum\nolimits_{i \in \clI} c_{i}^{k} = \sum\nolimits_{i \in \clI} c_{i}^{l} = 0$).
  Hence by contradiction, $\tV$ is also optimal solution.
  At the same time, by \cref{eq:unionintersectionpiforx1} we get that $\sum\nolimits_{i \in \clI} \hat{c}_{i} \leq 0$ but we have shown $\tV$ is an optimal solution.
  This implies $\hV$ is also an optimal solution, since \cref{eq:unionintersectionpiforx1} will not hold otherwise. 
\end{proof}

\paragraph{Example for \cref{thm:unionintersectionpiforx}:}
We know that the optimal solution is not unique in \cref{exm:notunique}.
Suppose $[a, c]$ and $[a, d]$ are optimal solutions such that $c \leq d$ in \cref{exm:notunique}.
Then we know that both $[a, c] \cap [a, d] = [a, c]$ and $[a, c] \cup [a, d] = [a, d]$ are also optimal.
  
\begin{thm}\label{thm:unionintersectionlargepiforx}
  Let $\sol(V, N)$ and $\sol(V, \tN)$ be the set of optimal solutions for a box $V$ with $ \alpha \in [0, 1], \pi, \tpi \in \bbR_{+}$, and neighborhoods $N = B(V, \pi)$ and $\tN = B(V, \tpi)$.
  The following results hold when $\tpi \geq \pi$:
  \begin{enumerate}
    \item  $C(V^{l}, \tN) \leq C(V', \tN)$, where $V' \subseteq  V^{l}$ and $V^{l} \in \sol(V, N)$. \label{thm:unionintersectionlargepiforx:1}
    \item $C(\bigcap\nolimits_{V^{l} \in \sol(V, N)} V^{l}, \tN) < C(V', \tN)$, where $V' \subset  \bigcap\nolimits_{V^{l} \in \sol(V, N)} V^{l}$ and $V \subseteq V'$. \label{thm:unionintersectionlargepiforx:2}
    \item $\forall V^{l} \in \sol(V, N), ~~ \exists V^{k} \in \sol(V, \tN) \text{ such that } V^{l} \subseteq V^{k}$. \label{thm:unionintersectionlargepiforx:3}
    \item $\forall V^{k} \in \sol(V, \tN), ~~ \exists V^{l} \in \sol(V, N) \text{ such that } V^{l} \subseteq V^{k}$. \label{thm:unionintersectionlargepiforx:4} 
  \end{enumerate}
\end{thm}

We interpret the statements in \cref{thm:unionintersectionlargepiforx} before presenting their proofs.
Statement \ref{thm:unionintersectionlargepiforx:1} implies that the cost for a subset of an optimal solution cannot be smaller than that of the optimal solution, even when the neighborhood is expanded when considering the subset.
Statement \ref{thm:unionintersectionlargepiforx:2} says that the cost for a candidate solution that is strictly smaller than the intersection of all optimal solutions for a given input, which itself is also optimal by \cref{thm:unionintersectionpiforx}, is strictly larger (again, even when the neighborhood is expanded).
Statements \ref{thm:unionintersectionlargepiforx:3} and \ref{thm:unionintersectionlargepiforx:4} imply containments of optimal solutions when the neighborhood is changed, i.e., when the neighborhood is enlarged, we get an optimal solution that contains the original optimal solution, and vice versa.

\begin{proof}
  \hspace*{0.0in}\\
  \vspace*{-0.2in}
  \begin{enumerate}
    \item Suppose $V \subseteq V' \subseteq V^{l}$.
      For the sequence $S(V',V^{l})$, the change in cost is $\sum\nolimits_{i \in \clI} c_{i}^{l} \leq 0$ for neighborhood $N$ since $ V^{l} \in \sol(V, N)$.
      Let $A = B(V, \tpi) \setminus B(V, \pi)$ and let $e^l = -\alpha \sum\nolimits_{x \in A} w_{x}$ and $e' = -\alpha \sum\nolimits_{x \in A} w'_{x}$ be the costs of points in $A$ for the boxes $V^{l}$ and $V'$, respectively.
      Then $e^l - e' \leq  0$ since $V' \subseteq V^{l}$.
      Furthermore, the change in the cost of the sequence $S(V', V^{l})$  is $\sum\nolimits_{i \in \clI}c_{i}^{l} + e^l - e' \leq 0$ for neighborhood $\tN$.
      Hence $C(V^{l}, \tN) \leq C(V', \tN)~~\forall V^{l} \in \sol(V, N)$.
		
    \item We know that $\bigcap\nolimits_{V^{l} \in \sol(V, N)} V^{l} \in \sol(V, N)$ from \cref{thm:unionintersectionpiforx}.
      Since $\bigcap\nolimits_{V^{l} \in \sol(V, N)} V^{l}$ is the smallest box in $\sol(V, N)$, there does not exist $V' \subset \bigcap\nolimits_{V^{l} \in \sol(V, N)} V^{l}$ such that $C(\bigcap\nolimits_{V^{l} \in \sol(V, N)} V^{l}, \tN) = C(V', \tN)$.
      The result now follows from the previous result in Statement \ref{thm:unionintersectionlargepiforx:1}.
      
    \item Given a $V^{l} \in \sol(V, N)$ and any $V^{k} \in \sol(V, \tN)$, let $ \tV = V^{l} \cap V^{k}, \hV = V^{l} \cup V^{k}$.
      The total change in the cost of the sequence $S(\tV, V^{k})$ is $\sum\nolimits_{i \in \clI }c_{i}^{k} \leq 0$ for neighborhood $B(V, \tN)$ since $V^{k} \in \sol(V, \tN)$.
      Similarly, the total change in the cost of the sequence $S(\tV, V^{l})$ is $\sum\nolimits_{i \in \clI}c_{i}^{l} \leq 0$ for neighborhood $B(V, \tN)$ since $V \subseteq \tV \subseteq V^{l}$, and $V^{l} \in \sol(V, N)$ (see Statement \ref{thm:unionintersectionlargepiforx:1}).
      For $S(\tV, \hV)$, total change in cost of the sequence is
      \[ \sum\limits_{i \in \clI} \hat{c}_{i} \leq \sum\limits_{i \in \clI } c_{i}^{k}+ c_{i}^{l}\]
      following \cref{prop:unioncosteq} and $\sum\limits_{i \in \clI } c_{i}^{k} \leq 0$ as well as $\sum\limits_{i \in \clI } c_{i}^{l} \leq 0$ for neighborhood $B(V, \tN)$. 
      Given $V^{k}$ is an optimal solution, the above equation can hold only if $\sum\nolimits_{i \in \clI } \hat{c}_{i}= \sum\nolimits_{i \in \clI } c_{i}^{k}$, $\sum\nolimits_{i \in \clI } c_{i}^{l}=0$. 
      This implies $\hV$ is optimal solution over $\tN$, given $V^{k}$ is optimal solution in $\sol(V, \tN)$, and $V^l \subseteq \hV$, proving the Statement.
		
    \item Following the proof above for Statement \ref{thm:unionintersectionlargepiforx:3}, we get that $\sum\nolimits_{i \in \clI } c_{i}^{l}=0$.
      This implies that $C(\tV, \tN) = C(V^{l}, \tN)$.
      Also, $C(\tV, N) \geq C(V^{l}, N)$ since $V^{l} \in \sol(V, N)$.
      The $\ell_{\infty}$ distance of $x \in \tN \setminus N$  for $\tV$ is always greater than or equal to the same distance for $V^{l}$ since $\tV \subseteq V^{l}$.
      This implies $C(\tV, N) = C(V^{l}, N)$, and hence $\tV \in \sol(V, N)$.
  \end{enumerate}
  \vspace*{-0.3in}
\end{proof}

\paragraph{Example for \cref{thm:unionintersectionlargepiforx}}
Consider the 1D point cloud $X = \{a, b, c\}$ such that $a - c < b -a$ as shown in \cref{fig:theorem26examplefigurecrop}.
Let input box $V = [a, a] \in \clU(0)$. If $\tV \supseteq V$ is some expansion of $V$, then $\tV$ can only be optimal if it expands towards $b$ for neighborhood $N =B(V, \pi = a-c+\delta)$ for small $\delta > 0$.
Let $\tV = [x, a]$, $x \geq c$ , then for $\alpha = 0.5$, $[x, a]$ is optimal $\forall x \in [c, a]$ (shown in example \ref{exm:notunique}).
This means $C([c, a], \tN) = C([x, a], \tN) = C([c, a]) + \alpha(b - a)~\forall~x \in [a, b]$ (satisfy items $1, 2$).
 
Now, consider the same problem with neighborhood $B(V, \tpi = c-a+\delta)$.
If $\tV \supseteq V$ is some expansion of $V$, then $\tV$ can only be optimal if it expands towards $b, c$ for neighborhood $\tN = B(V, \tpi = b-a+\delta)$.
Let $\tV = [x^{-}, x^{+}], x^{-}\in [c, a], x^{+} \in [a, b]$, then

\[C(\tV, \tN) = \alpha(b - a - x^{+} + x^{-} - c) + (1 - \alpha)(x^{+} - x^{-})\]
\[\dfrac{\partial C }{\partial x^{+}} = 1 - 2\alpha\] 
\[\dfrac{\partial C }{\partial x^{-}} = -1 + 2\alpha\] 

Now for $\alpha = 0.5$, $[x^{-}, x^{+}], \forall x^{-}\in [c, a], \forall x^{+} \in [a, b]$ are optimal (satisfy items 3, 4).
\begin{figure}[htp!] 
  \centering
  \includegraphics[scale=0.45]{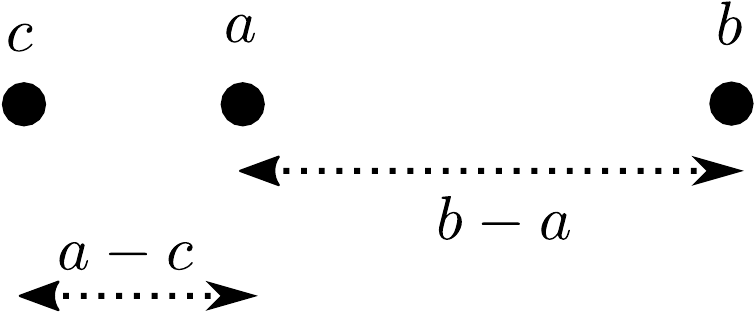}
  \caption{\label{fig:theorem26examplefigurecrop}}
\end{figure}
 
\begin{lem}\label{lem:alphachangeforx}
  Let $\sol_{\alpha}(V, N), \sol_{\talpha}(V, N)$  be sets of optimal solutions for given input box $V$, $\alpha, \talpha \in [0, 1],  \pi \in \bbR_{+}$ and given neighborhood $N = B(V, \pi)$.
  \add{Let $M$ be the largest optimal solution in $\sol_{\alpha}(V, N)$.} 
  Then $M \subseteq V^{k} \in \sol_{\talpha}(V, N)$, given $\talpha > \alpha$.
\end{lem}
\begin{proof}
  Let $V^{l} \in \sol_{\alpha}(V)$, $V^{k} \in \sol_{\talpha}(V)$.
  For $S(\tV, \hV)$, where $\tV \subseteq \hV$, the change in cost at each step of the sequence for $\talpha$ is $\leq$ change in cost at each step for $\alpha$ based on the optimization problem since $\talpha > \alpha$.
  Let $\tV = V^{l} \cap V^{k}$, $\hV = V^{l} \cup V^{k}$.
  Let $c^{k}_i, c^{l}_i, \hat{c}_i$ be the changes in cost at each step of the sequences $S(\tV, V^{k}), S(\tV, V^{l})$, and $S(\tV, \hV)$ for $\talpha$.
  Then
  \[\sum\limits_{i \in \clI } c_{i}^{k}\leq 0,~
  \sum\limits_{i \in \clI } c_{i}^{l} \leq 0, \text{ and }
  \sum\limits_{i \in \clI } \hat{c}_{i} \leq \sum\limits_{i \in \clI } c_{i}^{k}+ c_{i}^{l}\]
  based on \cref{prop:unioncosteq}.
  Since $V^{k} \in \sol_{\talpha}(V)$, the only possible case is  $\sum\limits_{i \in \clI } \hat{c}_{i} = \sum\limits_{i \in \clI } c_{i}^{k}$ and $ \sum\limits_{i \in \clI } c_{i}^{l}=0$, otherwise $V^{k}$ is not an optimal solution in $\sol_{\talpha}(V)$.
  This implies $\hV \in \sol_{\talpha}(V)$ and $C_{\talpha}(\tV) = C_{\talpha}(V^{l})$.
  Suppose $V^{l} \neq \tV$, then 
  \[ C_{\alpha}(V^{l}) = C_{\alpha}(\tV) + (1- \alpha)k_2 - \alpha k_1 \]
  where $k_1, k_2 >0$ and $(1- \alpha)k_2 - \alpha k_1 \leq 0$ since $V^{l} \in \sol_{\alpha}(V)$.
  This is true based on the formulation of the optimization problem.
  $k_1, k_2 = 0$ if and only if $\tV = V^{l}$.
  Let $\talpha = \alpha + \epsilon, \epsilon > 0$.
  Then,
  \[ C_{\talpha}(V^{l}) = C_{\talpha}(\tV) + (1- \alpha)k_2 - \alpha k_1 -\epsilon(k_1+ k_2). \]
  Now, $C_{\talpha}(V^{l}) = C_{\talpha}(\tV)$ if and only if $V^{l} = \tV$ since $k_1 = k_2 =0$ when $V^{l} = \tV$.
  It has to be true for any $V^{l} \in \sol_{\alpha}(V)$.
  This implies $M \subseteq V^{k}$.
\end{proof}

\begin{lem}\label{lem:pichangeforx} 
  Suppose $V^{k} \in \sol_{\alpha}(M, \Delta{N})$.
  Then $V^{k} \in \sol_{\alpha}(V, \tN)$, where $M \in \sol_{\alpha}(V, N)$ is a largest optimal solution, $N = B(V, \pi),\tN = B(V, \tpi)$, and \add{$\Delta{N} = \tN \setminus M$}.
\end{lem}
\begin{proof}
  \add{Consider a box $V^l$ such that $M \subseteq V^l \subseteq \tN$.
    The cost of $V^l$ is the same in both $\tN$ and $\Delta N$ since any point $x \in \tN \setminus V^l$ satisfies $x \in \Delta N$ and any point $x  \in \Delta N \setminus V^l$ satisfies $x \in \tN$.
    Let $\tM$ be the largest optimal solution of $\sol_{\alpha}(V, \tN)$.
    Then $M\subseteq\tM$ by Statement \ref{thm:unionintersectionlargepiforx:3} of \cref{thm:unionintersectionlargepiforx}.
    Since $M\subseteq\tM$ and the cost of any box $V^l$ is the same for both $\tN$ and $\Delta N$ neighborhood, $\sol_{\alpha}(M, \Delta{N}) \subseteq \sol_{\alpha}(V, \tN)$.
    Hence we have proved that $V^{k} \in \sol_{\alpha}(V, \tN)$.}
\end{proof}

\begin{lem}\label{lem:pointcountalpharelaappendixmod}
  Let $M$ be a largest optimal solution in $\sol_{\alpha}(V, N)$ such that $M \neq V$.
  With $\gamma = (1/\alpha) -  1$, we get that
  \[\dfrac{\theta(N\setminus M) + \theta(\partial M)}{p} \geq \gamma \geq \dfrac{\theta(N\setminus M)}{q}\]
  where $p, q \in \{1, \dots, 2n\}$ are the numbers of facets of $M$ that do not intersect $V$ and $N$, respectively.
\end{lem}

\begin{proof}
  For brevity, we denote $\theta (M), \theta (\partial M)$, and $\theta(N\setminus M)$ as $\theta, \partial \theta$, and $\theta^{c}$, respectively, in this proof.
  For a small $\delta > 0$, let $\Mde$ be the inward $\delta$-offset of the $p$ facets of $M$ that do not intersect $V$ such that  $V \subseteq \Mde$ and  if $x \in \partial M$ then $x \notin \Mde$, but if $x \in M \setminus \partial M$ then $x \in \Mde$.
  Let $w(x)$ and $\dew(x)$ be the weights of $x \in B(V,\pi)$ with respect to $M$ and $\Mde$, respectively.  
  Since $M$ is an optimal solution, we get that 
  \begin{equation}
    \begin{aligned}
      C_{\alpha}(M, N) &\leq C_{\alpha}(\Mde, N),~\text{ i.e., }\\
      -\alpha \sum\limits_{x \in N} w(x) + (1 - \alpha) |M| & \leq -\alpha \sum\limits_{x \in N} \dew(x) + (1 - \alpha) |\Mde|.
    \end{aligned}
  \end{equation}
  Since $\Mde$ is the inward $\delta$-offset of $M$ we get that $-\dew(x) + w(x) = \delta$ for $x \in \partial M$ or $x \in N\setminus M$, and $\dew(x) = w(x) = 0$ otherwise.
  Then we get that
  \begin{equation}
    \begin{aligned}
      (1 - \alpha)p\delta & \leq \alpha(\theta^{c} + \partial \theta)\delta,~\text{ which gives that } \\
      \dfrac{p}{\theta^c + \partial \theta + p} &\leq \alpha,~\text{ i.e.,}\\
      \dfrac{\theta^c + \partial \theta}{p} &\geq \gamma.
    \end{aligned}
  \end{equation}
	
  Now let $\Mde$ be the outward $\delta$-offset of the $q$ facets of $M$ that do not intersect $N$ such that if $x \in N\setminus M $ then $x \notin \Mde$.
  Hence we get $\dew(x) - w(x) = \delta$ for $x \in N \setminus M$, and $\dew(x) =  w(x) = 0$ otherwise.
  Hence we get that
  \begin{equation}
    \begin{aligned}
      C_{\alpha}(M, N) &\leq C_{\alpha}(\Mde, N), \text{ i.e.,}\\
      -\alpha \sum\limits_{x \in N} w(x) + (1 - \alpha) |M| &\leq -\alpha \sum\limits_{x \in N} \dew(x) + (1 - \alpha) |\Mde|, \text{ and }\\
      \alpha(\theta^{c})\delta &\leq (1 - \alpha) q\delta, \text{ which gives } \\
      \theta^{c} &\leq \left(\frac{1}{\alpha} - 1\right) q, \text{ i.e., } \\
      \dfrac{\theta^{c}}{q} &\leq \gamma.
    \end{aligned}
  \end{equation}
\end{proof}

Assuming there are $m$ points in the neighborhood, i.e., $|N|=m$, and setting $q=nm$, the running time of the LP in \cref{eq:ConLP} is $O(q^{3}\log{q})$ \cite{Schrijver1986}.
This bound could be prohibitively large in practice.
To reduce number of variables in the constructed LP, we consider working with discretizations of the space that are coarser than the point cover.

\input{pixelcover}

\input{pntpxlcovers}

\input{boxfiltration}

%% file: pixelcover.tex
\subsubsection{Pixel Cover}\label{sssec:pixelcover}
We work with any discretization of the $\dmn$-dimensional space containing $X \in \R^n$ in which each \emph{pixel} is a unit cube with integer vertices.
We assume also that $\pi \in \bbZ$ for defining the neighborhoods.
We note that \cref{prop:unioncosteq} holds in the pixel space as well.
Hence all results proven for point covers in \cref{sssec:pointcover} hold true for pixel covers as well.

We define the \emph{integer ceiling} of a number $x \in \R$ as $\iceil{x}$ as the smallest integer strictly larger than $x$.
In other words, $\iceil{x}=\ceil{x}$ when $x\not\in \Z$ and $\iceil{x}=x+1$ when $x \in \Z$.
We start with each point $x \in X$ belonging to a unique pixel $\sigma$.
For $x=(x_{1}, \dots, x_{n}) \in X$ we let $\sigma = [\floor{x_1}, \iceil{x_1}] \times \dots \times [\floor{x_n}, \iceil{x_n}]$.
We also define $m_{\sigma} = (m_{\sigma}^1, \dots, m_{\sigma}^n)$ as the centroid of pixel $\sigma$ and $\theta({\sigma})$ as the number of points of $X$ in $\sigma$.
Given a box $\tV$, we denote by $\Theta(\tV)$ the set of pixels $\sigma$ such that $m_{\sigma} \in \tV$ and $\theta(\sigma)\neq 0$, i.e., the set of nonempty pixels whose centroids are in the box.

For a given input box $V$, let $\tV = [\tl_1, \tu_1] \times \dots \times [\tl_n, \tu_n] \supseteq V$ be some box in the neighborhood $N = B(V, \pi)$.
The total width of box $\tV$ denoted by $|\tV|$ is $\sum\limits_{i \in \clI} \tilde{u}_i - \tilde{l}_i$. 
Considering all $\sigma \in \Theta(N)$, let $\overline{C}_{\alpha}(\tV, V)$ be the objective/cost function of the linear optimization of $\tV$ for a given input box $V$.

Let $w_{\sigma}$ be a variable corresponding to pixel $\sigma \in \Theta(N)$ for a given expansion $\tV$,
\begin{equation}\label{eq:pixelweight}
w_{\sigma} \leq \min \{\{m_{\sigma}^{i} - \tl_i \mid i \in\mathcal{I}\} \cup \{\tu_i - m_{\sigma}^{i} \mid i \in \mathcal{I}\} \cup \{0.5\} \}.
\end{equation}
Furthermore, note that $m_{\sigma} \in \tV$ if and only if
\[\min \{\{m_{\sigma}^{i} - \tl_i \mid i \in\mathcal{I}\} \cup \{\tu_i - m_{\sigma}^{i} \mid i \in \mathcal{I}\} \} \geq 0.\]

Let $V = [l_1, u_1]\times \dots \times [l_n, u_n]$ for simplicity of the notation in the formulation.

\begin{equation} \label{eq:LPobj}
\begin{aligned}
  \min_{\forall \tV \supseteq V} \hspace*{0.2in} & \pC_{\alpha}(\tV, N) = -\alpha \sum\limits_{\sigma \in \Theta(N)} \, w_{\sigma} \theta({\sigma}) ~+~ (1 - \alpha)\sum\limits_{i \in \mathcal{I}}(\tu_i - \tl_i)\hspace*{0.1in}  
\end{aligned}
\end{equation}
\vspace*{-0.1in}
\begin{equation}\label{eq:LPconstraint1}
\begin{aligned}
\mbox{s.t.} 
\hspace*{0.1in}\\ & \tu_i \geq u_i~ \forall i \in \{1,.., n\}, \hspace*{0.1in}\\ & \tl_i \leq  l_i~ \forall i \in \{1,.., n\}
\end{aligned}
\end{equation}

\begin{equation}\label{eq:LPconstraint2}
\begin{aligned}
& w_{\sigma}~\leq~m_{\sigma}^{i} - \tl_i ~ \forall i \in \{1,..,n\}, \sigma \in \Theta(N)\\
& w_{\sigma}~\leq~\tu_i - m_{\sigma}^{i}~ \forall i \in \{1,..,n\}, \sigma \in \Theta(N)\\
& w_{\sigma}~\leq~0.5~ \forall i \in \{1,..,n\}, \sigma \in \Theta(N)\\
\end{aligned}
\end{equation}

We note the correspondence of the LP objective function in \cref{eq:LPobj} to that of the LP in the point cover setting in \cref{eq:ConLPobj}.
As a distinguishing notation, we call the cost functional in the pixel cover setting as $\pC$ (it is $C$ under the point cover) and similarly, the optimal solution sets are denoted as $\psol$.

In general, an optimal solution $[l_1^{*}, u_1^{*}] \times \dots \times [l_n^{*}, u_n^{*}] \in \psol_{\alpha}(V, N)$ can have $l_{i}^{*}, u_{i}^{*} \in \bbR$.
We show that, more specifically, for every optimal solution, there exists an optimal solution with $l_{i}^{*}, u_{i}^{*} \in \bbZ, \forall i \in \clI$---see \cref{thm:lowuplimint}.
Hence we can focus our discussion on optimal solutions with $l_{i}^{*}, u_{i}^{*} \in \bbZ, i \in \clI$.
At the same time, we cannot identify such integer optimal solutions by simple rounding.
In fact, we consider three different rounding functions specific to the pixel cover as we define below.
Recall that the fractional part of $x \in \R$ is defined as $\frcpt(x)=x-\floor{x} \in [0,1)$.
 
\begin{defn} \label{def:rdgs}
Given a box $V = [l_1,u_1] \times \dots \times [l_n,u_n]$ we define three rounding functions of the box using functions $\psi_{1}, \psi_{2}$, and $\psi_{3}$ defined as follows.
\begin{equation}\label{eq:psi1}
\psi_{1}(l_{i}) =\ceil{l_{i}},~~
\psi_{1}(u_{i}) = \floor{u_{i}}
\end{equation}
\begin{equation}\label{eq:psi2}
\psi_{2}(l_{i}) = 
\begin{cases}
\ceil{l_{i}}  & \text{if}~ \frcpt(l_{i}) \in (0.5, 1)\\
\floor{l_{i}} + 0.5 & \text{if}~ \frcpt(l_{i}) \in (0, 0.5]\\
\floor{l_i} & \text{if}~\frcpt(l_i) = 0
\end{cases},~~ 
\psi_{2}(u_{i}) = 
\begin{cases}
\floor{u_{i}} + 0.5 & \text{if}~ \frcpt(u_{i}) \in [0.5, 1) \\
\floor{u_{i}} & \text{if}~ \frcpt(u_{i}) \in [0 ,0.5)
\end{cases}	
\end{equation}
\begin{equation}\label{eq:psi3}
\psi_{3}(l_{i}) = 
\begin{cases}
\ceil{l_{i}} & \text{if}~ \frcpt(l_{i}) \in (0.5, 1) \\
\floor{l_{i}} & \text{if}~ \frcpt(l_{i}) \in [0, 0.5]
\end{cases},~~
\psi_{3}(u_{i}) = 
\begin{cases}
\ceil{u_{i}} & \text{if}~ \frcpt(u_{i}) \in [0.5, 1) \\
\floor{u_{i}} & \text{if}~ \frcpt(u_{i})  \in [0, 0.5)
\end{cases}	
\end{equation}
We define $\Psi_{r}(V) = [\psi_r(l_{1}), \psi_r(u_{1})] \times \dots \times [\psi_r(l_{n}), \psi_r(u_{n})]$ for $r=1,2,3$.
\end{defn}

See \cref{fig:xsifunctiondeffig} for an illustration of these concepts. We now prove that this result always holds. 

\begin{figure}[htp!] 
	\centering
	\begin{subfigure}[t]{2.7in}
		\centering
		\includegraphics[width=2.4in]{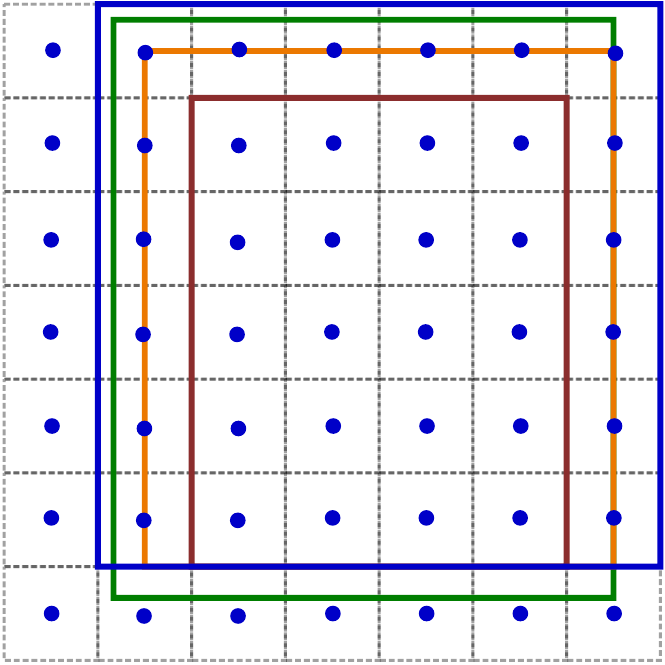}
		\caption{\label{fig:xsifunctiondeffig1}}
	\end{subfigure}
	\begin{subfigure}[t]{2.7in}
		\centering
		\includegraphics[width=2.4in]{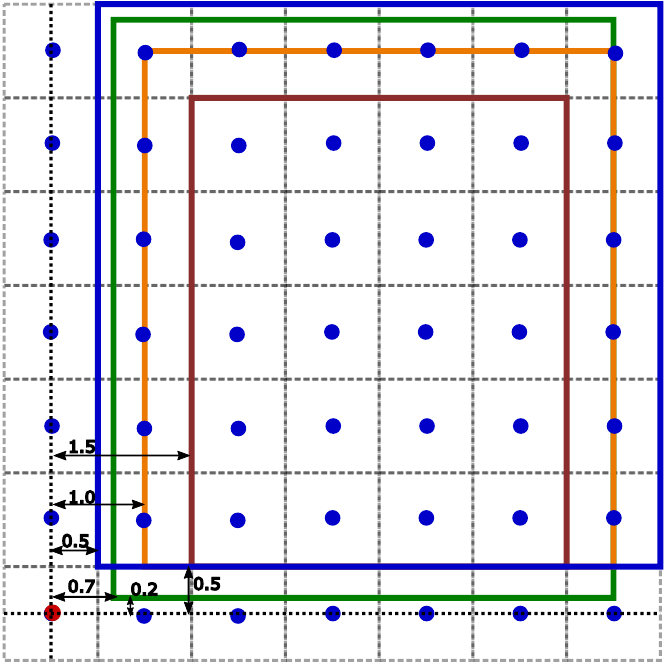}
		\caption{\label{fig:xsifunctiondeffig2}}
	\end{subfigure}	
	\caption{\label{fig:xsifunctiondeffig}
        Illustration of the construction of $\Psi_1(V), \Psi_2(V)$, and $\Psi_3(V)$ in brown, orange, and blue, respectively, for a given $V$ in green.
        Consider the centroid of the bottom-left pixel shown in red.
	This centroid will be weighted with respect to the left side of the box $V$ as shown in \cref{fig:xsifunctiondeffig2}.
        It will still be weighted with respect to the left side of each of $\Psi_{1}(V), \Psi_{2}(V)$, and $\Psi_{3}(V)$.}
\end{figure}

\begin{lem}\label{lem:refdirofsigma}
For a given $\sigma \in \Theta(N)$, let $w_{\sigma}$ and $w_{\sigma}^{r}$ be the weights for corresponding boxes $V$ and $\Psi_{r}(V)$ where $V \subseteq N$.
Then the following results hold for $r = 1, 2, 3$.
\begin{enumerate}
  \item If $w_{\sigma} = m_{\sigma}^{i} - l_{i}$, then $w_{\sigma}^{r} = m_{\sigma}^{i} - \psi_{r}(l_{i})$.
  \item If $w_{\sigma} = u_{i} - m_{\sigma}^{i}$, then $w_{\sigma}^{r} = \psi_{r}(u_{i}) - m_{\sigma}^{i}$.
\end{enumerate}
\end{lem} 
\begin{proof}
   Let $w_{\sigma} = m_{\sigma}^{i} - l_{i}$ for pixel $\sigma$.
   Consider placing the points $m_{\sigma}^{j} - l_{j}$, $u_{j} - m_{\sigma}^{j}$, $~~\forall j \in \clI$ on the number line $\R = (-\infty, \infty)$.
   Based on the optimization problem in \cref{eq:LPobj,eq:LPconstraint1,eq:LPconstraint2}, it must hold that $w_{\sigma}=m_{\sigma}^{i} - l_{i}$ is the left most, non positive point on this number line.

   We first consider the case of $\psi_1(V)$.
   Consider \add{placing} on the number line the points $m_{\sigma}^{j} - \psi_1(l_{j})$, $\psi_1(u_{j}) - m_{\sigma}^{j}$, $~~\forall j \in \clI$.
   Consider the subset of these points in the open interval $(h -0.5, h + 0.5),~ h \in \mathbb{Z}$.
   Following the definition of $\psi_1(V)$, all these points move to $h-0.5$.
   \add{Consider} a point is on \add{$0.5h'$} \add{for some $h' \in \mathbb{Z}$}, then that point will not move \add{for $\psi_1(V)$} (again, following the  definition of $\psi_1(V)$).
   Hence we get that $w^1_\sigma = m_{\sigma}^{i} - \psi_1(l_{i})$ is still the left most, non positive point on the number line.
   We can use a similar argument to show that when $w_{\sigma} = u_{i} - m_{\sigma}^{i}$, we must have $w_{\sigma}^{1} = \psi_1(u_{i}) - m_{\sigma}^{i}$.
	
   Now we consider the case of $\psi_2(V)$.
   Similar to the previous arguments, we now consider points in the open intervals $(h-0.5, h)$ and $(h, h + 0.5)$ for $h \in \mathbb{Z}$.
   Following the definition of $\psi_2(V)$, all these points move to $h-0.5$ and $h$, respectively.
   \add{Consider} \delete{If} a point is on \add{$0.5h'$} \delete{$0.5h$} or on \add{$h'$} \delete{$h$} for some \delete{$h \in \mathbb{Z}$} \add{$h' \in \mathbb{Z}$}, then this points will not move \add{for $\psi_2(V)$} (again, following the definition of $\psi_2(V)$).
   Hence we get that $w_{\sigma}^{2} = m_{\sigma}^{i} - \psi_2(l_{i})$ is still the left most, non positive point on the number line.
   Similarly, we can show that if $w_{\sigma} = u_{i} - m_{\sigma}^{i}$ then we must have $w_{\sigma}^{2} = \psi_2(u_{i}) - m_{\sigma}^{i}$.

   Finally we consider the case of $\psi_{3}(V)$.
   We now consider points in the intervals $(h, h + 0.5]$ and $(h-0.5, h)$ for $h \in \mathbb{Z}$\add{.}
   Following the definition of $\psi_3(V)$, all these points move to $h + 0.5$ and $h-0.5$, respectively.
   \add{Consider} \delete{If} a point is on \add{$h'$}\delete{$h$} for \add{some } \delete{$h \in \mathbb{Z}$} \add{$h' \in \mathbb{Z}$}, then this point will not move \add{for $\psi_3(V)$} (again, by the definition of $\psi_3(V)$).
   Hence we get that $w_{\sigma}^{3} = m_{\sigma}^{i} - \psi_3(l_{i})$ is still the left most, non positive point on the number line.
   A similar argument shows that if $w_{\sigma} = u_{i} - m_{\sigma}^{i}$ then it must be the case that $w_{\sigma}^{3} = \psi_3(u_{i}) - m_{\sigma}^{i}$.
\end{proof}

In general, an optimal box $V^*$ in the pixel cover setting could cut through a pixel, i.e., it could include parts of a pixel.
This setting could render the computation for identifying an optimal box quite expensive.
But the following theorem states that given any optimal box, each of its three roundings as specified in \cref{def:rdgs} is also an optimal solution.
In particular, the roundings $\Psi_1(V^*)$ and $\Psi_3(V^*)$ use only whole pixels.
Hence we could restrict our search for optimal boxes to those that contain only whole pixels.

\begin{thm}\label{thm:lowuplimint}
  Let $V^* = [l^*_{1}, u^*_{1}]\times \dots \times [l_{1}^{*}, u_{1}^{*}] \in \psol_{\alpha}(V, N)$ for the input box $V$ where $\alpha \in [0, 1], \pi \in \bbN$.
  Then $\Psi_{1}(V^{*}), \Psi_{2}(V^{*})$, $\Psi_{3}(V^{*})$ are also in $\psol_{\alpha}(V, N)$ where $N=B(V, \pi)$.  
\end{thm}

\begin{proof}
  If $l_{i}^{*}, u_{i}^{*} \in \bbZ~~\forall i\in \clI$, then the result holds trivially since $\Psi_{1}(V^{*}) = \Psi_{2}(V^{*}) = \Psi_{3}(V^{*}) = V^{*}$.
  Let $\clJ = \{i \mid l_{i}^{*} \notin \bbZ\}, \clK = \{i \mid u_{i}^{*} \notin \bbZ\}$, and 
  let $\delta(s_{j}^{*}) = s_{j}^{*} - \psi_{1}(s_{j}^{*})$ for $s_{j}^{*} \in \{l_{j}^{*}, u_{j}^{*}\}$ and $j \in \clI$.
  Note that $\delta(s_j^*) \in (-1,1)$.
  For any box $V'$, let $\partial \Theta(V')$  be the boundary set of pixels $\sigma$ such that $m_{\sigma} \in V', \sigma \cap V' \subset \sigma, \theta(\sigma)\neq 0$, i.e., the set of non-empty pixels that are partially contained in $V'$.
  Then by \cref{lem:refdirofsigma}, we have that
  \begin{equation}\label{eq:intoptiproof1}
    \begin{aligned}
    \pC(V^{*}) - \pC(\Psi_{1}(V^{*})) & = -\alpha \sum\limits_{\sigma \in \Theta(N) \setminus \Theta(V^{*})} \theta(\sigma) (w_{\sigma}^{*} - w_{\sigma}^{1}) -\alpha \sum\limits_{\sigma \in \partial \Theta(V^{*})} \theta(\sigma)(w_{\sigma}^{*} - w_{\sigma}^{1})\\
    &~~~~ + (1 - \alpha) \sum\limits_{i \in \clI} \delta({u_{i}^{*}}) - \delta({l_{i}^{*}})\\
    &=\alpha \sum\limits_{j \in \clJ}c_{j}\delta({l_{j}^{*}}) -\alpha \sum\limits_{k \in \clK} c_{k}\delta({u_{k}^{*}})+ (1 - \alpha) \sum\limits_{i \in \clI} \delta({u_{i}^{*}}) - \delta({l_{i}^{*}})
    \end{aligned}	
  \end{equation}
  where $c_{j}, c_{k} \in \bbZ_{>0}$ are some constants determined by the set of relevant pixels $\sigma \in \Theta(N) \setminus \Theta(V^{*})$.
  Since $V^{*}$ is optimal, we have that 
  \begin{align}\label{eq:intoptiproof2}
    \pC(V^{*}) - \pC(\Psi_{1}(V^{*})) \leq 0.
  \end{align}
  To show optimality of $\Psi_1(V^*)$, we show that the expression in \cref{eq:intoptiproof2} is also $\geq 0$.
  To this end, we consider another candidate box  $\tV = [\tl_{1}, \tu_{1}] \times \dots \times [\tl_{n}, \tu_n] \subseteq N$ such that
  \begin{itemize}
    \item $\Theta(V^{*}) = \Theta(\tV)$ and also the set of unfilled pixels in both boxes $V^*$ and $\tV$ are identical;
    \item $\tl_{i} = \psi_{1}(l_{i}^{*}) + p\delta({l_{i}^{*}}) ~\forall i \in \clJ$; and 
    \item $\tu_{i} = \psi_{1}(u_{i}^{*}) + p\delta({u_{i}^{*}})~\forall i \in \clK$,
  \end{itemize} 
  where $p > 1$.
  Based on the structure of $\tV$, we get that $\Theta(N)\setminus \Theta(V^{*}) = \Theta(N)\setminus \Theta(\tV)$ and $\partial \Theta(V^{*}) = \partial \Theta(\tV)$.
  Also, the weights of all the pixels in $\Theta(N)\setminus \Theta(V^{*})$ and $\partial \Theta(V^{*})$ for $\tV$ will still be determined by the same direction as in $V^{*}$ by \cref{lem:refdirofsigma}. 

  Now, we get by \cref{lem:refdirofsigma} and \cref{eq:intoptiproof1} that
  \begin{align}\label{eq:intoptiproof3}
    \pC(\tV) - \pC(\Psi_{1}(V^{*})) = \pC(\tV) - \pC(V^*) + \pC(V^*) - \pC(\Psi_{1}(V^{*})) = p (\pC(V^{*}) - \pC(\Psi_{1}(V^{*}))).
  \end{align}	
  Combining \cref{eq:intoptiproof2,eq:intoptiproof3} and the fact that $p > 1$, we get that
  \[ \pC(\tV) - \pC(V^{*}) = (p-1)(\pC(V^{*}) - \pC(\Psi_{1}(V^{*}))) \leq 0,\]
  implying that $\tV$ is optimal as well.
  But that observation and \cref{eq:intoptiproof3} immediately imply that $\Psi_{1}(V^{*})$ is also optimal.

  We can repeat the above argument with $\tV$ defined with respect to $\Psi_2(V^*)$ to show that $\Psi_{2}(V^{*})$ is also optimal.
  Now we show $\Psi_{3}(\tV)$ is optimal. 

  Let $\delta'(s_{j}^{*})= \psi_{1}(s_{j}^{*}) - \psi_{2}(s_{j}^{*})$ and $\delta''(s_{j}^{*})= \psi_{1}(s_{j}^{*}) - \psi_{3}(s_{j}^{*})$ for $s_{j}^{*} \in \{l_{j}^{*}, u_{j}^{*}\}$.
  Then we get that $\delta''(s_{j}^{*}) = 2\delta'(s_{j}^{*})$ based on the definition of functions $\psi_{1}, \psi_{2}$, and $\psi_{3}$.
  Since both $\Psi_1(V^*)$ and $\Psi_2(V^*)$ are optimal, we have that
  \begin{align*}
    \pC(\Psi_{1}(V^{*})) - \pC(\Psi_{3}(V^{*})) = 2 (\pC(\Psi_{1}(V^{*})) - \pC(\Psi_{2}(V^{*}))) = 0,
  \end{align*}	
showing that $\Psi_{3}(V^{*})$ is optimal as well.
\end{proof}

\begin{exm}[for \cref{thm:lowuplimint}] \label{eg:1d2pt}
  Consider the 1D point cloud $X= \{a, b\}$ with $a, b$ being the centers of the pixels $\sigma(a), \sigma(b)$ as shown in \cref{fig:theorem211examplefigurecrop}.
  Further, we can observe $\theta(\sigma(a)) = 1, \theta(\sigma(b)) = 1$.
  Let input box $V = [a - 0.5, a + 0.5] \in \clU(0)$.
  If $\tV \supseteq V$ is some expansion of $V$, then $\tV$ can only be optimal if it expands in $x$-direction and towards $b$ for neighborhood $B(V, \pi = b - a)$.
  Let $\tV = [a - 0.5, x], x \leq b + 0.5$, and $N = B(V, \pi = b - a)$, then 
  \[ \pC(\tV) = -0.5\alpha + \alpha(b - x) + (1 - \alpha)(x - a - 0.5), ~\text{ and } \]
  \[\dfrac{\partial \pC(\tV)}{\partial x} = 1 - 2\alpha.\]

  If $\alpha = 0.5 $, then $\dfrac{\partial \pC(\tV)}{\partial x} = 0~~\forall x$ and $\tV = [a - 0.5, x]$ is an optimal solution $\forall x \leq b + 0.5$.
  Then $\psi_{1}([a - 0.5, x]), \psi_{2}([a - 0.5, x]), \psi_{3}([a - 0.5, x])$ are optimal for all $x \leq b + 0.5$.
\end{exm}
 
\begin{figure}[htp!] 
  \centering
  \includegraphics[scale=0.60]{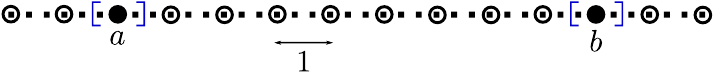}
  \caption{\label{fig:theorem211examplefigurecrop}
  1D point cloud in \cref{eg:1d2pt}.}
\end{figure}

\begin{cor}\label{lem:largestboxboundarypixel}
  Let $\bM$ be the largest optimal solution in $\psol_{\alpha}(V, N)$.
  Then $\bM$ has only pixels $\sigma$ such that either $\sigma \subseteq \bM$ or such that $m_{\sigma} \notin \bM$ and $\sigma \cap \bM \neq \emptyset$.
\end{cor}
\begin{proof}
  Suppose $\bM$ contains a pixel $\sigma$ such that $m_{\sigma} \in \bM, \sigma \cap \bM \subset \sigma$.
  But this implies that $\bM$ is not the largest optimal solution since $\sigma \in \Psi_{3}(\bM)$ and $\Psi_3(\bM)$ is also optimal, and hence $\Psi_3(\bM) \not\subseteq \bM$.
  Hence we get a contradiction.
\end{proof}

\begin{lem}\label{lem:pixelcountalpharelaappendixmod}
  Let $M$ be a largest optimal solution in $\psol_{\alpha}(V, N)$ such that $M \neq V$.
  With $\gamma = (1/\alpha)-1$, we get that
  \[\dfrac{\theta(N\setminus M)}{p} \geq \gamma \geq \dfrac{\theta(N\setminus M)}{q}\]
  where $p,q \in \{1, \dots, 2n\}$ are the numbers of facets of $M$ that do not intersect $V$ and $N$, respectively.   	
\end{lem}
\begin{proof}
  By \cref{lem:largestboxboundarypixel} $M$ has no pixel $\sigma$ such that $m_{\sigma} \in M$ and $\sigma \cap M \subset \sigma$.
  Now we take infinitesimal $\delta$ inward and outward offsets of $M$ just as we did in the proof of \cref{lem:pointcountalpharelaappendixmod}.
  We get two  similar inequalities, but now with no pixels on the boundary of $M$:
  \[\dfrac{\theta(N\setminus M)}{p} \geq \gamma \geq \dfrac{\theta(N\setminus M)}{q}.   \vspace*{-0.2in}
  \]
\end{proof}

We end this Subsection by noting that all the results for pixel covers are valid for any \emph{discretization}, i.e., the corner points of pixels in the pixel cover do not need to have integer coordinates.
Furthermore, all results from \cref{sssec:pointcover} except \cref{lem:pointcountalpharelaappendixmod} also hold for any discretization.

%% file: pntpxlcovers.tex
\subsubsection{Relations of the Point and Pixel Cover Optimal Solutions} \label{sssec:pntpxlcovers}

While the point cover setting may be more accurate, it is also computationally expensive.
An appropriate pixel cover promises computational efficiency.
At the same time, we show that, under certain assumptions, there is no loss of accuracy when working with a pixel cover.
For a given input box, we can always find an optimal box in the pixel cover setting that contains its corresponding optimal solution in the point cover setting.
Complementarily, we can find an optimal solution in the point cover setting that contains its corresponding optimal solution in the pixel cover setting.
In this sense, the point and pixel cover optimal solutions are \emph{interleaved}.
And hence, we can work with the largest optimal solution in the pixel cover setting, since that is guaranteed to also contain the optimal point cover solution.

We \delete{make use of the following property of} \add{work under the following assumption on} the distribution of points in $X$ about how many of them fall inside an open box and how many lie on a facet of the box.
\add{All subsequent results are presented under this assumption.}
\begin{asmn}\label{proper:ptdsbtnsos}
The number of points in any open box is $\leq O(\epsilon)$ where $\epsilon$ is the total width of the box. 
The number of points on any $(\dmn-1)$-dimensional facet of a box is $\leq c$ for some constant $c$.
\end{asmn}

Based on the above \delete{property} \add{assumption} about the distribution of $X$, we get that
\begin{align*}
\theta(\hV \setminus \tV) &\leq \clL (|\hV| - |\tV|) + 2cn~\text{ and }\\
|\Theta(\hV \setminus \tV)| &\leq \clL (|\hV| - |\tV|) + 2cn
\end{align*}
where $\hV \supset \tV$ are open boxes and $\clL>0$ is a constant.
We get the result on the pixels from the observation that there can be at most the same number of pixels with $\theta(\sigma)\neq 0$ as the number of points in $X$.
Any point distribution can be described to have this property for \emph{some} $\clL$ and $c$. 
But it may be desirable for the results we present below to hold for \emph{small} values of $\clL$ and $c$.
In the context of this observation, we note that the number of points in any facet being bounded is not a restrictive one---it is similar to the standard assumption that the points are in general position \cite{EdHa2009}.


We now formalize the \emph{interleaving} of optimal solutions under point and pixel covers.
In \cref{thm:pointpixelcoverrela}, we show that the optimal point cover solution is contained in an optimal pixel cover solution.
Complementarily, in \cref{thm:pixelpointcoverrela}, we show that the optimal pixel cover solution is contained in an optimal point cover solution.
We present a corresponding lemma for each theorem, which in turn specifies the conditions that some of the parameters must satisfy for the results to hold.
Let $\clH$ be the width of a pixel in the pixel cover.

\begin{lem}\label{lem:pixelcoverstab1} 
  Let $\hV \in \sol_{\alpha}(V, N)$, $\bV \in \psol_{\bar{\alpha}}(V, N)$, $\tV = \hV \cap \bV$, and let $\hV \not\subset \bV$.
  Also let $0\leq \alpha < 1 $, $0\leq \bar{\alpha} \leq 1$, $\bar{\alpha} = \alpha + \varepsilon$ for $\varepsilon \geq 0$, and $0 < \clK \leq \Delta$ where $\Delta = |\hV| - |\tV|$.
  Given \cref{proper:ptdsbtnsos} for the distribution of points in $X$, if
  \begin{align}
  & 0 \leq \clL < \dfrac{2(\gamma+1)}{\clH}, \label{eqn:Llmt}  \\
  & \clK > \dfrac{\clH (cn + n\gamma)}{(\gamma+1)-0.5\clL\clH},  \hspace*{0.5in}\text{ and } \label{eqn:Klmt} \\
  & \varepsilon= \min \left\{ \dfrac{n\gamma\clH\alpha + 0.5\clH\alpha\clL\clK + \clH \alpha cn}{((\gamma + 1) - 0.5\clH\clL)\clK-\clH cn -n\gamma\clH}, 1-\alpha\right\}, \label{eqn:epsval}
  \end{align}
  then, $\overline{C}_{\bar{\alpha}}(\hV, \mdelete{V} \add{N}) - \overline{C}_{\bar{\alpha}}(\tV, \mdelete{V} \add{N}) \leq 0$.
\end{lem}

\add{Before presenting the proof, we note that the  upper bound on $\mathcal{L}$ in \cref{eqn:Llmt} tells us about the condition that the point distribution $X$ needs to satisfy for the statement to hold.
  The lower bound on $\mathcal{K}$ in \cref{eqn:Klmt} tells us how  small the difference between $\hat{V}, \tilde{V}$ should be for the statement to hold.
  And $\varepsilon$ in \cref{eqn:epsval} is what we need to add to $\alpha$ to get $\bar{\alpha}$.
  }

\begin{proof}
  First, we note that $C_{\alpha}(\hV, V) - C_{\alpha}(\tV, V) \leq 0$ since $\hV \in \sol_{\alpha}(V, N)$.
  Hence we get that 
  \begin{equation}\label{eq:pixelcoverstab1eq1}
    \begin{aligned}
      -\alpha \sum\limits_{x \in N}(\hat{w}(x) - \tilde{w}(x)) + (1 - \alpha)\Delta\leq 0\\
	\implies~ \gamma\Delta \leq \sum\limits_{x \in N}\hat{w}(x) - \tilde{w}(x),~\text{ with } \gamma = (1/\alpha)-1.
    \end{aligned}
  \end{equation}
	
  The $\ell_{\infty}$-distance of $x \in N\setminus\hV$ to $\hV$ can be increased at most by $0.5\clH$ as we go from point space to pixel space.
  Also, an increase in the $\ell_{\infty}$-distance of $x$ from $\hV$ implies that its $\ell_{\infty}$-distance from $\tV$ cannot decrease, since $\hV \supseteq \tV$.
  Furthermore, for any $x \in \hV$, we have that $\hat{w}(x)$ is increased by $0.5\clH$.
  Similarly, for any $x \in \tV$, we get that $\tilde{w}(x)$ is increased by $0.5\clH$. %
  By \cref{lem:pointcountalpharelaappendixmod}, we get that
  \[\theta(N\setminus\hV) \leq 2n\gamma.\]

  Hence we get that
  \begin{equation*}
    \begin{aligned}
    &\overline{C}_{\bar{\alpha}}(\hV, V) -\overline{C}_{\bar{\alpha}}(\tV, V)~~~\\
    & \leq C_{\bar{\alpha}}(\hV, V) - C_{\bar{\alpha}}(\tV, V) + 0.5\clH\bar{\alpha}\theta(\hV \setminus \tV) + 0.5\clH\bar{\alpha}\theta(N \setminus \hV)\\
    & \leq C_{\bar{\alpha}}(\hV, V) - C_{\bar{\alpha}}(\tV, V) + 0.5\clH\bar{\alpha}\clL \Delta + \clH\bar{\alpha}cn  + \clH\bar{\alpha}n\gamma\\
    &\leq -\bar{\alpha} \sum\limits_{x \in N}(\hat{w}(x) - \tilde{w}(x)) + (1 - \bar{\alpha})\Delta  + 0.5\clH\bar{\alpha}\clL \Delta + \clH\bar{\alpha}cn + \clH\bar{\alpha}n\gamma,\\
    \end{aligned}
  \end{equation*}
  where the third and fourth terms in the second inequality ($0.5\clH\bar{\alpha}\clL \Delta, \clH\bar{\alpha}cn$) come from \cref{proper:ptdsbtnsos}.
Now, using the result in \cref{eq:pixelcoverstab1eq1}, we get that this expression is
  \begin{equation}\label{eq:pixelcoverstab1eq2}
    \begin{aligned}
    \leq &~C_{\alpha}(\hV, V) - C_{\alpha}(\tV, V)-\varepsilon (\sum\limits_{x\in N}(\hat{w}(x) - \tilde{w}(x)) + \Delta) \\
	 &~ + \alpha(0.5\clH\clL \Delta + \clH cn + \clH n\gamma) 
	 + \varepsilon(0.5\clH\clL \Delta + \clH cn +\clH n\gamma)\\
    \leq &~\alpha(0.5\clH\clL \Delta+ \clH cn+ \clH n\gamma) + \varepsilon(0.5\clH\clL \Delta + \clH cn + \clH n\gamma)	-\varepsilon (\gamma + 1)\Delta
    \end{aligned}        
  \end{equation}  
  If the above expression is $\leq 0$ then $\overline{C}_{\bar{\alpha}}(\hV, V) -\overline{C}_{\bar{\alpha}}(\tV, V) \leq 0$.
  Hence we want that
  \begin{equation}\label{eq:pixelcoverstab1eq3}
    \begin{aligned}
    \alpha(0.5\clH\clL \Delta+ \clH cn+ \clH n\gamma)
	&\leq \varepsilon (\gamma + 1)\Delta - \varepsilon(0.5\clH\clL \Delta + \clH cn + \clH n\gamma)\\
    \dfrac{\alpha(0.5\clH\clL \Delta + \clH cn + \clH n\gamma)}{((\gamma + 1) - 0.5\clH\clL) \Delta -\clH cn-\clH n\gamma}&\leq \varepsilon			
    \end{aligned}
   \end{equation}
   Since $\varepsilon \geq 0$, it must hold that $((\gamma + 1) - 0.5\clH\clL) \Delta - \clH cn - \clH n\gamma) >0$.
   Since $\clH cn \geq 0$ and $\clH n\gamma \geq 0$, we should have that $(\gamma + 1) - 0.5\clH\clL > 0$, which holds because of the assumption in the Lemma in \cref{eqn:Llmt}.
   We should also get that
   
   \begin{equation}\label{eq:pixelcoverstab1eq5}
      \begin{aligned}
        & \Delta > \dfrac{\clH (cn + n\gamma)}{(\gamma+1)-0.5\clL \clH}.
      \end{aligned}
   \end{equation}
   But $\clK$ is the minimum value of $\Delta$, and since $\clK$ satisfies the above bound as per the Lemma's assumption in \cref{eqn:Klmt}, $\Delta$ satisfies the bound in \cref{eq:pixelcoverstab1eq5} as well.
	
We show that the lower bound for $\varepsilon$ in \cref{eq:pixelcoverstab1eq3} is a decreasing function of $\Delta$ (\cref{prop:epsilonincdecfunc}).
Hence we get the smallest value for $\varepsilon$ (in \cref{eqn:epsval}) for $\Delta= \clK$.
Finally, noting that $\bar{\alpha} = \alpha + \varepsilon \leq 1$ gives the expression for $\varepsilon$ in \cref{eqn:epsval}.
\end{proof}

\begin{prop}\label{prop:epsilonincdecfunc}
  Given $\clH, \Delta$ as defined in \cref{lem:pixelcoverstab1}, $f(\Delta) = \dfrac{n\gamma\clH\alpha + 0.5\clH\alpha\clL\Delta+ \clH \alpha cn}{((\gamma + 1) - 0.5\clH\clL)\Delta-\clH cn -n\gamma\clH}$ is a decreasing function of $\Delta$. 
\end{prop}
\begin{proof}
   We can rewrite $f$ to isolate $\clH$.
   \[f(\Delta) = \dfrac{(n\gamma\alpha + 0.5\alpha\clL\Delta + \alpha cn)\clH}{(\gamma + 1)\Delta - (0.5\clL\Delta+cn +n\gamma)\clH}.\]
   If $\alpha > 0$                
   \begin{equation*}
   \begin{aligned}
   \dfrac{\partial f}{\partial \Delta} &= \dfrac{0.5\clH\alpha\clL (((\gamma + 1) - 0.5\clH\clL)\Delta-\clH cn -n\gamma\clH) - (n\gamma\clH\alpha + 0.5\clH\alpha\clL\Delta + \clH \alpha cn) ((\gamma + 1) - 0.5\clH\clL)}{(((\gamma + 1) - 0.5\clH\clL)\Delta-\clH cn -n\gamma\clH)^2}\\
   &= \dfrac{-0.5\clH\alpha\clL n\gamma\clH -0.5\clH\alpha\clL \clH cn - n\gamma\clH\alpha ((\gamma + 1) - 0.5\clH\clL) - \clH\alpha cn ((\gamma + 1) - 0.5\clH\clL)}{((\gamma + 1)\Delta - \clH(0.5\clL\Delta+ cn +n\gamma))^2}\\
   &= \dfrac{- n\gamma\clH\alpha (\gamma + 1) - \clH\alpha cn (\gamma + 1)}{((\gamma + 1)\Delta - \clH(0.5\clL\Delta+ cn +n\gamma))^2}\\
   & < 0.			  
   \end{aligned}
   \end{equation*}
   And if $\alpha = 0$, then $\dfrac{\partial f}{\partial \Delta} = 0$.
\end{proof}

The next theorem shows that when the same initial box is chosen in a point and a pixel cover for optimization over the same $\pi$-neighborhood, if the parameters are chosen to satisfy \cref{lem:pixelcoverstab1} then every optimal solution in the point cover setting is contained in some optimal solution in the pixel cover.
 
\begin{thm}\label{thm:pointpixelcoverrela}
  Given $\hV \in \sol_{\alpha}(V, N)$ and $\bV \in \psol_{\bar{\alpha}}(V, N)$, if $\clH$ and $\varepsilon$ are chosen as specified in \cref{lem:pixelcoverstab1} then $\hV\cup\bV \in \psol_{\bar{\alpha}}(V, N)$.
\end{thm}

\begin{proof}
  With $\tV = \hV \cap \bV$, let $S(\tV, \hV), S(\tV, \bV)$, and $S(\tV, \hV\cup\bV)$ be the cost changes at each step of the sequence $\hat{c}_i, \bar{c}_i$, and $c_i$, respectively, in the pixel space with parameter $\bar{\alpha}$,
  where $c_i$ represents the sequence for going from $\tV$ to $\hV\cup\bV$.
  Since the parameters  $\clH$ and $\varepsilon$ are chosen to satisfy \cref{lem:pixelcoverstab1}, we get that
  \[\overline{C}_{\bar{\alpha}}(\hV, V) - \overline{C}_{\bar{\alpha}}(\tV, V) \leq 0\]
  since $\hV \in \sol_{\alpha}(V, N)$.
  Hence $\sum\limits_{i \in \clI_{n} } \hat{c}_{i} \leq 0$.
  We also get that $\sum\limits_{i \in \clI_{n} } \bar{c}_{i} \leq 0$ since $\bV \in \psol_{\bar{\alpha}}(V, N)$.
  By \cref{prop:unioncosteq} we get that $\sum\limits_{i \in \clI_{n} } c_{i} \leq \sum\limits_{i \in \clI_{n} } \hat{c}_{i} + \sum\limits_{i \in \clI_{n} } \bar{c}_{i}$.
  If $\sum\limits_{i \in \clI_{n} } \hat{c}_{i} < 0$ then $\sum\limits_{i \in \clI_{n} } c_{i} < \sum\limits_{i \in \clI_{n} } \bar{c}_{i}$.
  This implies $\bV$ is not an optimal solution, contradicting $\bV \in  \psol_{\bar{\alpha}}(V, N)$.
  Hence we must have $\sum\limits_{i \in \clI_{n} } \hat{c}_{i}=0$, giving $\sum\limits_{i \in \clI_{n} } c_{i} = \sum\limits_{i \in \clI_{n} } \bar{c}_{i}$ since $\bV \in \psol_{\bar{\alpha}}(V, N)$.
  But this implies $\bV \cup \hV \in \psol_{\bar{\alpha}}(V, N)$. 
\end{proof}

We simplify the expressions in \cref{lem:pixelcoverstab1} further when $\clH$ is small.        
These simplifications make the presentation of stability results in \cref{sec:costabnoise} easier to comprehend.

\begin{cor}\label{cor:pixelcoverstab1}
  If $\clH$ is small enough so that $0.5\clH\clL \ll 1$,  then \cref{lem:pixelcoverstab1} holds for $\clK \approx \clH \alpha (cn + n\gamma) + \sqrt{\clH}$ and $\varepsilon \approx \min\{ \alpha^2 \dmn \sqrt{\clH} (\gamma + c + \, 0.5 \clL \sqrt{\clH}/\dmn \,), \,1-\alpha\}$.
\end{cor}
\begin{proof}
  Since $0.5\clH\clL \ll 1$, we get that
  \begin{align}\label{eq:corpixelcoverstab1}
    \dfrac{1}{(\gamma + 1) - 0.5\clL\clH} &\approx \alpha.
  \end{align}
  To satisfy \cref{eqn:Klmt} in \cref{lem:pixelcoverstab1}, we can choose
  \begin{align}\label{eq:corpixelcoverstab2}
     \clK &= \dfrac{\clH(cn +n\gamma)}{(\gamma + 1) - 0.5\clL\clH} + \sqrt{\clH}.
  \end{align}
  Hence we get 
  \begin{align}\label{eq:corpixelcoverstab3}
     \clK & \approx \clH\alpha(cn + n\gamma) + \sqrt{\clH}.
  \end{align}

  Using this value of $\clK$, the expression for $\varepsilon$ in \cref{eqn:epsval} becomes
  \begin{align*}
    \dfrac{n\gamma\clH\alpha + 0.5\clH\alpha\clL\clK+\clH \alpha cn}{((\gamma + 1) - 0.5\clH\clL)\clK-\clH cn -n\gamma\clH}&= \dfrac{n\gamma\clH\alpha+\clH \alpha cn + 0.5\clH\alpha\clL\clK}{((\gamma + 1) - 0.5\clH\clL)\sqrt{\clH}}\\
     & \approx \dfrac{n\gamma\clH\alpha+\clH \alpha cn + 0.5\clH\alpha\clL(\clH\alpha(cn + n\gamma) + \sqrt{\clH})}{((\gamma + 1) - 0.5\clH\clL)\sqrt{\clH}}\\
     & = \dfrac{(n\gamma\clH\alpha +\clH \alpha cn)(1 + 0.5\clH\sqrt{\clH}\alpha\clL)+ 0.5\clH\alpha\clL}{((\gamma + 1) - 0.5\clH\clL)\sqrt{\clH}}\\
     & \approx \dfrac{\alpha(n\gamma\clH\alpha +\clH \alpha cn + 0.5\clH\sqrt{\clH}\alpha\clL)}{\sqrt{\clH}} \\
     & = \alpha^2 \dmn \sqrt{\clH} (\gamma + c + \, 0.5 \clL \sqrt{\clH}/\dmn \,).
  \end{align*}
  Hence we can choose $\varepsilon \approx \min\{ \alpha^2 \dmn \sqrt{\clH} (\gamma + c + \, 0.5 \clL \sqrt{\clH}/\dmn \,), \,1-\alpha\}$ to satisfy \cref{lem:pixelcoverstab1}.
\end{proof}
We note that \cref{thm:pointpixelcoverrela} also holds for the simplified choices for $\clK$ and $\varepsilon$ given in \cref{cor:pixelcoverstab1} (since they satisfy \cref{lem:pixelcoverstab1}).

We present the analogous result for containment of optimal solutions in reverse using similar arguments.

\begin{lem}\label{lem:pixelcoverstab2} 
  Let $\hV \in \sol_{\alpha}(V, N)$, $\bV \in \psol_{\bar{\alpha}}(V, N)$, $\tV = \hV \cap \bV$, and let $\bV \not\subset \hV$.
  Also let $0 \leq \bar{\alpha} < 1 $, $0 \leq \alpha \leq 1$, $\alpha = \bar{\alpha} + \varepsilon$ for $\varepsilon \geq 0$ and $0 < \clK \leq \Delta$ where $\Delta = |\bV| - |\tV|$.
  Given \cref{proper:ptdsbtnsos} for the distribution of points in $X$, if
  \begin{align}
  & 0 \leq \clL < \dfrac{2(\bar{\gamma}+1)}{\clH}, \label{eqn:Llmtrev}  \\
  & \clK > \dfrac{\clH (cn + n\bar{\gamma})}{(\bar{\gamma}+1)-0.5\clL\clH},  \hspace*{0.5in}\text{ and } \label{eqn:Klmtrev} \\
  & \varepsilon= \min \left\{ \dfrac{n\bar{\gamma}\clH\bar{\alpha} + 0.5\clH\bar{\alpha}\clL\clK + \clH \bar{\alpha} cn}{((\bar{\gamma} + 1) - 0.5\clH\clL)\clK-\clH cn -n\bar{\gamma}\clH}, 1-\bar{\alpha}\right\}, \label{eqn:epsvalrev}
  \end{align}
  then, $C_{\alpha}(\hV, \delete{V} \add{N}) - C_{\alpha}(\tV, \delete{V} \add{N}) \leq 0$.
\end{lem}

\begin{proof}
  Follows from arguments similar to those used in the proof of \cref{lem:pixelcoverstab1}.
\end{proof}

\begin{thm}\label{thm:pixelpointcoverrela}
  Given $\hV \in \sol_{\alpha}(V, N)$ and $\bV \in \psol_{\bar{\alpha}}(V, N)$, if $\clH$ and $\varepsilon$ are chosen as specified in \cref{lem:pixelcoverstab2} then $\bV \cup \hV \in \sol_{\alpha}(V, N)$.
\end{thm}

\begin{proof}
  Follows from a similar argument to that used in proving \cref{thm:pointpixelcoverrela}.
\end{proof}

\begin{cor}\label{cor:pixelcoverstab2}
  If $\clH$ is small enough so that $0.5\clH\clL \ll 1$,  then \cref{lem:pixelcoverstab2} holds for $\clK \approx \clH \bar{\alpha} (cn + n\bar{\gamma}) + \sqrt{\clH}$ and $\varepsilon \approx \min\{ \bar{\alpha}^2n\sqrt{\clH}(\bar{\gamma} + c + \,0.5\clL\sqrt{\clH}/n\,), \,1-\bar{\alpha}\}$.
\end{cor}

%% file: boxfiltration.tex
\subsection{Box Filtration}\label{ssec:filtration}

We describe the box filtration algorithm in a unified manner for both the point cover and pixel cover settings.
The only minor modification in the pixel cover setting is that we can apply any of the rounding functions specified in \cref{def:rdgs} to the optimal box without loss of optimality.
Let $m$ be the minimum multiple of $\pi$ to cover $X$ with $B(V, m\pi)$ for any $V \in \clU(0)$, the initial cover.
Let $\clS(V) = \{V[0], V[\pi], \dots, V[m\pi]\}$ be a sequence of expansions of $V$, where $V[0] = V$ and $V[j\pi] \in \sol(V[(j-1)\pi], N[j])$ for $N[j]=B(V,j\pi)$ for $j=1,\dots,m$.
Then the cover of $X$ at the $j$th expansion is $\clU(j\pi) = \{ V[j\pi] \,|\,V \in \clU(0)\}$.
Finally, let $\srS(\clU) = \{\clU(0), \dots,$ $ \clU(m\pi)\}$ be a sequence of covers of $X$ and $K(\clU) =\{\Nrv(\clU(0)), \dots, \Nrv(\clU(m\pi))\}$ be the corresponding sequence of simplicial complexes specified as the nerves of the covers.
\add{The filtration value of a simplex is $j\pi$, when the boxes involved in the intersection defining the simplex first intersect as the neighborhood is grown by $j\pi$ starting from their respective initial boxes.
  While $j$ can be any nonnegative real value in theory, we assume it takes nonnegative integer values with the ease of computation in mind.
}%

Using boxes as cover elements provides the following result that has implications on the efficient computation of nerves.

\begin{lem} \label{lem:boxcechVR}
  Given $V_p$ is subset of $\bbR^{n}$ and $\{V_p\}_{p \in \clP}$ be a finite collection of boxes such that $V_i \cap V_j \neq \emptyset$ for every $i, j \in \clP$ with $i \neq j$.
  Then $\cap_{q \in \clQ} V_q \neq \emptyset$ for all $\clQ \subseteq \clP$ with $|\clQ| \geq 2$.
\end{lem}
\begin{proof}
  Since every pair of the boxes intersect, we should have that the intervals defining the boxes in each dimension, also intersect in pairs.
  Hence we have that
  \[ [l^p_i,u^p_i] \cap [l^{p'}_i,u^{p'}_i] \neq \emptyset~~\forall p, p' \in \clP.\]
  Hence, by properties of intersection of closed intervals in $\R$, it must be the case that
  \[
     \bigcap\limits_{ p\in \clP}[l^p_i,u^p_i]  \neq \emptyset
  ~~\forall i \in [1, n].\]
  This implies that $\cap_{q \in \clQ} V_q \neq \emptyset$. 
\end{proof}

Hence we get all higher order intersections of the boxes as soon as every pair of them intersect.
In this sense, the box filtration defined using pairwise intersections of the boxes is identical to the one defined using intersections of all, i.e., higher, orders. 
Recall that in the default case using Euclidean balls, we get only the inclusion of the VR complex as a subcomplex of the \v{C}ech complex at a larger radius \cite{EdHa2009}.

We specify two algorithms to build this sequence of simplicial complexes.
Our goal is to find the largest optimal solution at each step, which is what we do in the {\bfseries largest optimal expansion} algorithm.
At the same time, we may want to minimize the computational cost for finding the largest optimal box at each step.
With this goal in mind, we limit the search for the largest box to $k$ steps in the {\bfseries $\boldsymbol{k}$-optimal expansion} algorithm.

\paragraph{Largest optimal expansion algorithm:}
We find the largest optimal box at each step $j$ using a two-step approach.
In the first step, we find $V^* \in \sol(V[(j-1)\pi], N[j])$ by solving the linear program (\cref{eq:ConLP}) for input box $V[(j-1)\pi]$ and set $V[j\pi]=V^*$.
In the second step, we find a new optimal solution $V' \in \sol(V[(j-1)\pi], N[j])$ such that $V'$ strictly contains the current $V[j\pi]$ by solving the linear program with the additional constraint
\begin{equation} \label{eq:LPconstraint3}
  \sum\limits_{i \in \mathcal{I}}(u'_i - l'_i) \geq \sum\limits_{i \in \mathcal{I}}(u_i^{*} - l_i^{*}) + P
\end{equation}
where $P$ is the largest value such that $C(V^{*}, N) = C(V', N)$.
We identify $P$ using a binary search.
We specify the range of values $P \in [l_P, u_P]$ and start with $l_P=0$ in the point cover setting and $l_P=\clH$ in the pixel cover setting, and $u_P = m \dmn \pi$ (unless we have a smaller estimate).
At each step of the binary search, we solve the linear program with additional constraint in \cref{eq:LPconstraint3} with  $P= (l_P+u_P)/2$.
If $C(V^{*}, N) = C(V', N)$ then we set $V[j\pi] = V'$ and update $l_P = P$.
Else, we update $u_P = P$.
We repeat the binary search  while $|u_P - l_P| > \varepsilon$ for $\varepsilon=0.1$ or another suitably small tolerance.
This algorithm terminates in at most $\log(m \dmn \pi)$ steps.

\paragraph{$\boldsymbol{k}$-optimal expansion algorithm:}
This algorithm is the same as largest optimal expansion solution except that we stop the second step of binary search after $k$ iterations.
Note that while we are not guaranteed to find the optimal $P$, especially when $k$ is chosen small, we always find a valid $V[j\pi]$ using this algorithm as well.

\begin{lem}\label{lem:filtration}
  $K(\clU)$ is a filtration.
\end{lem}

\begin{proof}
  Suppose the $p$-simplex $\sigma \in \Nrv(\clU(j\pi))$ be determined by $\cap_{i=0}^p V_i[j\pi] \neq \emptyset$ where $V_i[j\pi] \in \clU(j\pi)$ for corresponding initial boxes $V_i \in \clU(0)$ for $i=0,\dots,p$.
  The sequences of expansions $\clS(V_i)$ of each box $V_i$ are all non-decreasing as determined by either the largest optimal or the $k$-optimal expansion algorithms, i.e., $V_i[(j+1)\pi] \supseteq V_i[j\pi]$ for each $V_i \in \clU(0)$ and for all $j$.
  Hence $\cap_{i=0}^p V_i[j\pi] \neq \emptyset$ implies also that $\cap_{i=0}^p V_i[(j+1)\pi] \neq \emptyset$, showing that $\sigma \in \Nrv(\clU((j+1)\pi))$.
  Hence $K(U)$ is a filtration.
\end{proof}

\paragraph{Correctness of expansion algorithms:}
By \cref{lem:pichangeforx}, if $V[(j+1)\pi] \in \sol(V[j], N[(j+1)\pi])$ then $V[(j+1)\pi] \in \sol(V, N[(j+1)\pi])$ where $V \in \clU(0)$ is the corresponding initial box.
Further, largest optimal and $k$-optimal expansion algorithms return optimal solutions in $\sol(V[j\pi], N[(j+1)\pi]) \forall j$.
Hence we are guaranteed to obtain a global optimal box for the input box $V$ using either approach, even if the $k$-optimal expansion algorithm may not identify a largest optimal box.
We also know from \cref{lem:filtration} that $K(\clU)$ is a filtration.

\paragraph{Complexity of expansion algorithms}
The total number of variables in the linear optimization problem is $O(q)$ where $q= n \times |X|$ with $|X|$ being the usual cardinality of point cloud $X$ in the case of point cover, and the number of filled pixels in $X$ in the case of pixel cover.
Further, under \cref{proper:ptdsbtnsos} for the distribution of points in $X$ and based on \cref{lem:pointcountalpharelaappendixmod}, we get that $q = O(2n\gamma \times n) = O(\gamma n^2)$.
With the complexity of solving the linear optimization problem given as $O(\lpc(q))$, we get the complexity of the largest optimal expansion algorithm for one expansion step as $O(\log(mn\pi) \lpc(q))$ since the largest optimal solution will be found in at most $O(\log(mn\pi))$ steps.
Hence the overall complexity of the largest optimal expansion algorithm is $O(m|\clU(0)|\log(mn\pi)\lpc(q))$, where $m$ is length of the expansion sequence $S(V)~\forall V \in \clU(0)$.
Similarly, the complexity of the $k$-optimal expansion algorithm is $O(m|\clU(0)|k\lpc(q))$.
Note that $|\clU(0)| \leq |X|$, with the highest cardinality realized when each point in $X$ is assigned a unique box in the initial cover $\clU(0)$.
The fastest algorithm for linear optimization gives $O(\lpc(q)) = O(q^{\omega} \log^2(q) \log(q/\delta))$ where $\omega \approx 2.38$ specifies the complexity of matrix multiplication and $0 < \delta \leq 1$ is an accuracy term \cite{Br2020}.
Furthermore, given \cref{lem:boxcechVR}, we get the box filtration with higher order intersection at the complexity of building the box filtration with pairwise intersection.
Furthermore, the direct dependence of time complexity on $|X|$ is linear, and hence the box filtration algorithms will scale better than classical VR approaches when the dimension is fixed but the data set size increases significantly.

%% file: stable.tex
\section{Stability}\label{sec:costabnoise}

The notion of stability implies that ``small'' changes in the input data set $X$ should lead to only ``small'' changes in its filtration.
We first show such a stability result for point covers---if the input data set $X$ is changed by a small amount to $Y$, then the point cover filtrations of $X$ and \emph{a slightly enlarged} point cover filtration of $Y$ are close.
We need the slight enlargement since the growth of boxes in any box filtration is always guided by the points in the PCD.
Since $Y$ is obtained as a perturbation including the addition/deletion of points of $X$, the points that guide the growth of specific boxes in the box filtration of $X$ may not be available to act in the same manner in the box filtration of $Y$, even if we were to enlarge the neighborhoods for the expansion steps.
But we need only a constant enlargement by a factor determined by the Gromov-Hausdorff distance between $X$ and $Y$.
We then show a similar result for pixel covers using this result for point covers by showing that the pixel cover filtration of $\bar{X}$, the pixel discretization of $X$, is close to its point cover filtration and a similar result holds for $\bar{Y}$ and $Y$.

We assume that the input data set $X$ and its changed version $Y$ are finite.
We use certain standard definitions and results used in the context of persistence stability for geometric complexes \cite{ChdeSOu2014}.

\subsection{Stability of Point Cover Filtration}\label{ssec:costabnoise}

 \begin{defn}
   A multivalued map $\clC: X \rightrightarrows Y$ from set $X$ to set Y is a subset of $X \times Y$ that projects surjectively onto $X$ using $\pi_X: X \times Y \to X$.
   The image $\clC(\sigma)$ of a subset $\sigma$ of $X$ is the projection onto $Y$. 
 \end{defn}
 \noindent We denote by $\clC^T$ the transpose of $\clC$.
 Note that $\clC^T$ is subset of $Y \times X$, but not always a multivalued map. 

 \begin{defn}
   A multivalued map $\clC: X \rightrightarrows Y$ is a correspondence if the canonical projection $\clC \to Y$ is surjective, or equivalently, if $\clC^T$ is also a multivalued map. 
\end{defn}

 \begin{defn}[$\epsilon$-net]
   Let $(X, d_X)$ be a metric space and $\epsilon > 0$. 
   A set $S \subset X$ is called an {\it $\epsilon$-net} of $X$ if for every $x \in X$ there exists an $s \in S$ such that $d_X(x, s) \leq \epsilon$. 
 \end{defn}

 When $(X, d_{X})$ and $(Y, d_{Y})$ are metric spaces, the \emph{distortion} of a correspondence $\clC: X \rightrightarrows Y$ is defined as $\dis(\clC) = \sup\{|d_{X}(x, x') - d_{Y}(y, y')| : (x, y), (x', y') \in \clC\}$.
 Then $\dGH(X, Y) = \frac{1}{2}\inf\{\dis(\clC): \clC~~\text{is a correspondence}~~X \rightrightarrows Y\}$ is the Gromov-Hausdorff distance between $(X, d_{X})$ and $(Y, d_{Y})$.

 We present a definition of $\epsilon$-simplicial maps specific to our context of box filtrations.

 \begin{defn}[$\epsilon$-simplicial map]
   Let $\bbS, \bbT$ be two box filtered simplicial complexes with vertex sets $X$ and $Y$, respectively.
   A multivalued map $C: X \rightrightarrows Y$ is $\epsilon$-simplicial from $\bbS$ to $\bbT$ if for any $\pi \in \R_{\geq 0}, 0 \leq \alpha \leq 1$ and any simplex $\sigma \in \bbS_{\pi, \alpha}$, every finite subset of $\clC(\sigma)$ is a simplex of $\bbT_{\hpi, \halpha}$ for $\hpi \leq \pi + \delta$ and $\halpha \leq \alpha + \delta'$ and $\epsilon = \max\{\delta,\delta'\}$.
 \end{defn}
 
 \begin{prop}\label{prop:stability1}
   Let $2\dGH(X, Y) < \epsilon$ and let the correspondence $\clC : X \rightrightarrows Y$ be such that $\dis(\clC) = 2\dGH(X, Y)$.
   Then points in the $\epsilon$-neighborhood of $X$ can be mapped to a single point in $Y$ by $\clC$.
 \end{prop}
 \begin{proof}
   If the point $x'$ does not belong to the $\epsilon$-neighborhood of $x\in X$, then $d_X(x', x) \geq \epsilon$.
   For any $(x', y)$, $(x, y) \in \clC$, we get that $|d(x', x) - d(y, y)| \geq \epsilon$.
   But this is a contradiction since $2\dGH(X, Y) < \epsilon$. 
 \end{proof}

 \begin{prop}\label{prop:stability2}
   If $S$ is a minimal $\epsilon$-net of $X$ and $2\dGH(X, Y) < \epsilon$, then $|Y| \geq |S|$.
 \end{prop}
 \begin{proof}
   By the definition of an $\epsilon$-net, we know that the $\epsilon$-neighborhood of all $x \in S$ will cover $X$.
   And by \cref{prop:stability1}, up to all points in the $\epsilon$-neighborhood of $X$ can be mapped to a single point in $Y$ if $2\dGH(X, Y) < \epsilon$.
   Hence we get that $|Y| \geq |S|$. 
 \end{proof}
 Similarly, we can show $|X| \geq |S|$ if $S$ is a minimal $\epsilon$-net of $Y$.

 \add{
   In many of the subsequent results presented in this Section, we denote constants of proportionality by $t$.
   The exact values of these constants are not crucial for the arguments.
   Hence, we use the generic notation $t$ in multiple places for the sake of brevity.
 }
 
 \begin{prop}\label{prop:stability3}
   If $2\dGH(X, Y) < \epsilon$ then \add{$||Y|-|X||\leq t\epsilon$ for some constant $t$}.
 \end{prop}
 \begin{proof}
   An $\epsilon$-ball is contained in a box of total width $2\epsilon\times n=2\epsilon n$.
   Then based on \cref{proper:ptdsbtnsos}, we get that \add{$|X|\geq|S|\geq |X|-t\epsilon$ for some constant $t$} for a minimal $\epsilon$-net $S$ of $X$.
   \cref{prop:stability2} gives that $|Y| \geq |S|$.
   Hence \add{$|Y|\geq |X| - t\epsilon$}. \add{Similarly we can show $|X|\geq |Y| - t\epsilon$}. 
 \end{proof}

 For the next set of results, we consider the setting where $2 \dGH(X,Y) < \epsilon$.
 We will also consider a correspondence \add{$\clC$} between $X$ and $Y$ that realizes this distance (as specified in \cref{prop:stability1}, for instance) and \add{$(x, y) \in \clC$}.
 \add{Given $V_1 \in X$ and $V_2 \in Y$, consider a transformation $T:(x, V_1) \to (y, V_2)$ such that the distance to faces of $V_1$ from $x$ under $d_X$ is the same as the distance to the corresponding faces of $V_2$ from $y$ under $d_Y$.
   Such a transformation will always exist under translation and scaling.
 }

 \add{
   \begin{prop}\label{prop:newaddstability3}
     Let $V_1 \in X, V_2 \in Y$, and let $(y, V_2)$ be the image of $(x, V_1)$ under $T$.
     Then $|\theta(V_1) - \theta(V_2)| \leq t\epsilon$ for some constant $t$.
   \end{prop}
   \begin{proof}
     Let $S_1, S_2$ be the subsets of points in $V_1$ and $V_2$, respectively.
     For each $x \in S_1~\exists y\in S_2$ such that $(x, y) \in \clC$ and for each $y \in S_2 ~\exists x\in S_1$ such that $(x, y) \in \clC$.
     This implies that $2\dGH(S_1, S_2) \leq \epsilon$.
     Then, similar to \cref{prop:stability3}, we get
     \[||S_1|-|S_2||\leq t'\epsilon\]
     where $t'$ is some constant.
     The set of points in $V_1$ that may not be mapped to points in $V_2$ by the correspondence $\clC$ is $\epsilon$ distance away from the boundary of $V_1$.
     This implies
     \[\theta(V_1) - |S_1|\leq t'\epsilon.\]
     Similarly for $V_2$, we get
     \[\theta(V_2) - |S_2|\leq t'\epsilon.\]
     Then we get
     \begin{align*}
       |\theta(V_1) - \theta(V_2)| &= |\theta(V_1) - |S_1| + |S_1|-|S_2| +|S_2|  -\theta(V_2)|\\
       &\leq |\theta(V_1) - |S_1|| + ||S_1|-|S_2|| +||S_2|  -\theta(V_2)|\\
       &\leq 3t'\epsilon = t\epsilon ~\text{ for some constant }t. \qedhere
     \end{align*}
   \end{proof}
 }
 
 For $x \in X$, let $V(x) \in \clU(0)$ be its box in the initial cover of $X$.
 Let $M$ be the largest optimal solution in $\sol_{\alpha}(V(x), N)$ where $N = B(V(x), \pi)$.
 \add{Let $M_{\epsilon} = \{z \in M : \mathrm{dist}(z, N\setminus M) \geq \epsilon\}$, i.e., the $\epsilon$-inward offset of $M$.}
 Let $(y, \tM)$ \add{in $Y$ be} the image of \add{($x, M_{\epsilon}$)} \add{ in $X$ under $T$}, and let $\hat{M} = V(y) \cup \tM$.
 Here, $V(y)$ is the box of $y$ in the initial cover of $Y$.
 We assume $|M| \geq 2\epsilon$. 

 \begin{lem}\label{lem:coverstability1}
   Let $N_1 = B(V(x),\pi), N_2 = B(V(y),\pi)$, and $N_3 = B(V(y),\pi + \epsilon)$.
   Then the following bounds hold.
   \begin{equation}
     \begin{aligned}
       |\theta(\hat{M}) - \theta(M)|&\leq \add{t\epsilon } \mdelete{O(\epsilon)} + 4cn\\
       |\theta(N_3 \setminus \hat{M}) - \theta(N_1 \setminus M)|&\leq \add{t\epsilon } \mdelete{O(\epsilon)}\\
       \partial \theta(\hat{M}) &\leq 2cn\\
     \end{aligned}
   \end{equation}
   \add{where $t$ is some constant.} 
 \end{lem}

 \begin{proof}
 	
   Since \add{$(y, \tM)$} is the image of the \add{$(x, M_\epsilon)$} \add{under $T$, where $M_\epsilon$ is the} negative $\epsilon$-offset of $M$, using the upper bound on boundary points of a box based on \cref{proper:ptdsbtnsos}, we get that
   \[\theta(\tM) \geq \theta(\hat{M}) - \add{t'\epsilon} \mdelete{O(\epsilon)} -2cn\]
   \add{where $t'$ is some constant.}  
   \cref{proper:ptdsbtnsos} also gives that
   \[\theta(\tM) \leq \theta(\hat{M}) + \add{t'\epsilon}\mdelete{O(\epsilon)}+2cn.\]
   \add{Based on \cref{prop:newaddstability3} we get that
     \[|\theta(\tM) - \theta(M_{\epsilon})| \leq \add{t'\epsilon.}\mdelete{O(\epsilon)}\]}
   \add{Then by \cref{proper:ptdsbtnsos}} we get that
   \[|\theta(\tM) - \theta(M)|\leq  |\theta(\tM) - \theta(M_{\epsilon})| + |\theta(M_{\epsilon}) - \theta(M)|\leq \add{t''\epsilon}\mdelete{O(\epsilon)} + 2cn.\]
   \add{where $t''$ is some constant.}
   Hence,
   \[|\theta(\hat{M}) - \theta(M)|\leq |\theta(\hat{M})-\theta(\tM)| + |\theta(\tM) - \theta(M)|\leq \add{t\epsilon}\mdelete{O(\epsilon)}+ 4cn \add{\text{ for some constant } t}.\]
   Since neighborhoods are open boxes, similar to above set of results we get
   \[|\theta(N_3 \setminus \hat{M}) - \theta(N_1 \setminus M)| \leq |\theta(N_3 \setminus \hat{M}) - \theta(N_2 \setminus M)| + |\theta(N_2 \setminus M) - \theta(N_1 \setminus M)|\leq \add{t\epsilon}\mdelete{O(\epsilon)}. \]
   Also, \cref{proper:ptdsbtnsos} gives
   \[\partial \theta(\hat{M}) \leq 2cn. \qedhere \] 
 \end{proof}

 \begin{lem}\label{lem:coverstability2}
   $\hat{M}$ is contained in a largest optimal solution of $\,\sol_{\hat{\alpha}}(V(y), N_{\epsilon})$ where
   \[ \hat{\alpha} =\min\left\{\alpha \big/ \left(1 - \alpha \left(\dfrac{\add{t\epsilon}\mdelete{O(\epsilon)}}{2n}\right)\right), 1\right\} \text{ and } N_{\epsilon} = B(V(y), \pi+\epsilon).
   \]
   \add{where $t$ is some constant.}
 \end{lem}

 \begin{proof}
   Since $M$ is the largest optimal solution in $\sol_{\alpha}(V(x), N)$ with $N=B(V(x), \pi)$, \cref{lem:pointcountalpharelaappendixmod} gives 
   \begin{align}\label{eq:coverstability1}
     \dfrac{\theta(N\setminus M) + \theta(\partial M)}{p}&\geq \gamma \geq \dfrac{\theta(N\setminus M)}{q}
   \end{align}
   where $\gamma = 1/\alpha - 1$.
   These inequalities give bounds for $\theta(N\setminus M)$ as 
   \begin{align}\label{eq:coverstability2}
     q\gamma &\geq \theta(N\setminus M) \geq p\gamma - \theta(\partial M).
   \end{align}
	
   Let $\hat{M}$ be the largest optimal solution in $\sol_{\alpha'}(V(y), N_{\epsilon})$ for some $\alpha'$ and $N_{\epsilon}=B(V(y), \pi + \epsilon)$.
   Then similar to \cref{eq:coverstability1} we get
   \begin{align}\label{eq:coverstability3}
     \dfrac{\theta(N_{\epsilon}\setminus  \hat{M}) + \theta(\partial \hat{M})}{\hat{p}} & \geq \gamma' \geq \dfrac{\theta(N_{\epsilon}\setminus  \hat{M})}{\hat{q}}
   \end{align}
   where $\gamma' = 1/\alpha' - 1$.
   Based on \cref{lem:coverstability1} \add{we get $\theta(N_{\epsilon}\setminus  \hat{M}) \geq \theta(N\setminus M)-t\epsilon$}. \add{Based on} \cref{eq:coverstability2} \add{$\theta(N\setminus M) \geq p\gamma - \theta(\partial M)$}. \add{Then based on \cref{eq:coverstability3}} we get
   \begin{equation}\label{eq:coverstability4}
     \begin{aligned}
       \gamma' &\geq \dfrac{\theta(N\setminus M)-\add{t\epsilon}\mdelete{O(\epsilon)}}{\hat{q}}\\
       \gamma' &\geq \dfrac{p\gamma - \theta(\partial M)-\add{t\epsilon}\mdelete{O(\epsilon)}}{\hat{q}}\\
       \gamma' &\geq \dfrac{p\gamma}{\hat{q}} - \dfrac{\theta(\partial M) + \add{t\epsilon}\mdelete{O(\epsilon)}}{\hat{q}} \\
       \implies &\dfrac{\theta(\partial M) + \add{t\epsilon}\mdelete{O(\epsilon)}}{\hat{q}} \geq \dfrac{p\gamma}{\hat{q}} -  \gamma'.
     \end{aligned}
   \end{equation}
   Since we have a difference of at least $\epsilon$ between the upper and lower limits of the $\hat{M}$ and $N_{\epsilon}$, we must have $\hat{q} = 2n$.
   We can apply an inward offset of $\epsilon$ on $M$, and since the bounds hold for all valid values of $p$, we take $p=2n$.
   And for \cref{eq:coverstability4} to hold, we need $\theta(\partial M)\geq 0$.
   Then we get that
   \begin{align}\label{eq:coverstability6}
     \dfrac{\add{t\epsilon}\mdelete{O(\epsilon)}}{2n} &\geq \gamma - \gamma'.
   \end{align}
	
   Rewriting \cref{eq:coverstability6} terms of $\alpha, \alpha'$ gives
   \begin{equation}\label{eq:coverstability5}
     \begin{aligned}
       &\alpha \geq \alpha' - \alpha \alpha'\left(\dfrac{\add{t\epsilon}\mdelete{O(\epsilon)}}{2n}\right)\\
      \implies &\alpha \alpha'\left(\dfrac{\add{t\epsilon}\mdelete{O(\epsilon)}}{2n}\right)\geq \alpha' - \alpha\\
      \implies &\alpha (\alpha + \Delta)\left(\dfrac{\add{t\epsilon}\mdelete{O(\epsilon)}}{2n}\right)\geq \Delta\\
      \implies &\dfrac{\alpha^2\left(\dfrac{\add{t\epsilon}\mdelete{O(\epsilon)}}{2n}\right)}{1 - \alpha \left(\dfrac{\add{t\epsilon}\mdelete{O(\epsilon)}}{2n}\right)} \geq \Delta. 		
     \end{aligned}
   \end{equation}
   Hence the largest possible value of $\hat{\alpha} = \alpha + \Delta = \alpha \big/ \left(1 - \alpha \left(\dfrac{\add{t\epsilon}\mdelete{O(\epsilon)}}{2n}\right)\right)$.
   We also require, by definition, that $\hat{\alpha} \leq 1$.
\end{proof}

 We denote the box filtration of $X$ with parameters $\alpha, \pi$ (see \cref{lem:filtration}) as $\BoxF(X,\pi,\alpha)$.
 And we denote the $\epsilon$-enlarged version of $\BoxF(X,\pi,\alpha)$ as $\eBoxF(X,\pi,\alpha)$, where we start with the entire $\BoxF(X,\pi,\alpha)$ and enlarge the largest optimal box constructed at each step by $2\epsilon$ in both directions in each dimension.
 Note that such $\epsilon$-enlargements for small $\epsilon$ values are considered only for theoretical purposes to prove the stability results.
 We use the following standard result on $\epsilon$-interleaving of filtrations.

 \begin{prop}(\cite{ChdeSOu2014})\label{prop:chazinterleaveprop}
   Let $\bbS, \bbT$ be filtered complexes with vertex sets $X, Y$, respectively.
   If $\mathcal{C}:X \rightrightarrows Y$ is a correspondence such that $\mathcal{C}$ and $\mathcal{C}^{T}$ are both $\epsilon$-simplicial, then together they induce a canonical $\epsilon$-interleaving between the persistence modules $H(\bbS)$ and $H(\bbT)$, the interleaving homomorphisms being $H(\mathcal{C})$ and $H(\mathcal{C}^{T})$.
 \end{prop}

 \begin{thm}\label{lem:2epsinterleave}
    Let $X$ and $Y$ be finite subsets of $\R^n$ endowed with the induced metric of $\R^n$. 
    If $2\dGH(X, Y)$ $<\epsilon$, then the persistence modules $H(\BoxF(X,\pi,\alpha))$ and $H(\BoxF(Y,\hpi,\halpha))$ are $\lambda-$interleaved where 
	$\lambda = \max \{ \alpha^2\left(\add{t\epsilon}\mdelete{O(\epsilon)}/(2n)\right) \, \big/ \, [ 1 - \alpha \left(\add{t\epsilon}\mdelete{O(\epsilon)}/(2n)\right)], \, 2\epsilon \}$ \add{ for some constant $t$} and $\hpi \leq \pi + \lambda$ and $\halpha \leq \min\{\alpha + \lambda, 1\}$.
 \end{thm}
 
 \begin{proof}
   Let $\mathcal{C}:X \rightrightarrows Y$ be a correspondence with distortion at most $\epsilon$. 
   To show that $H(\BoxF(X,\pi,\alpha))$ and $H(\BoxF(Y,\hpi,\halpha))$ are $\lambda$-interleaved, we will show that $\mathcal{C}$ and $\mathcal{C}^T$ are $\lambda$-simplicial. 

   Let $\sigma \in \BoxF(X, \pi, \alpha)$ and let $\tau$ be any finite subset of $\clC(\sigma)$. 
   We show that $\tau \in \eBoxF(Y,\hpi, \halpha)$. 
   For $y, y' \in \tau$ pick corresponding vertices $x, x' \in \sigma$ such that $(x, y)$ and $(x', y')$ are in $\clC$. 
   Let $M(x)$ and $M(x')$ be the largest optimal solutions in $\sol_{\alpha}(V(x), N_1)$ and $\sol_{\alpha}(V(x'), N_2)$ where $N_1 = B(V(x), \pi), N_2 = B(V(x'), \pi)$, respectively.

   Let $\hat{M}(y)$ be the box that is equal to $V(y)$ union with the image of the box $\{z \in M(x) : \mathrm{dist}(z, N_1\setminus M(x)) \geq \epsilon\}$, under the unique translation of $\bbR^n$ that takes $x$ to $y$, and let $\hat{M}(y')$ be similarly defined. 
   Then based on \cref{lem:coverstability2},
   \begin{align*}
     \hat{M}(y)  & \subseteq \text{largest optimal solution } M(y)  \in \sol_{\halpha}(\Vy, N_3)   \text{ for } N_3 = B(\Vy,  \pi+\epsilon), \text{ and }\\
     \hat{M}(y') & \subseteq \text{largest optimal solution } M(y') \in \sol_{\halpha}(V(y'), N_4) \text{ for } N_4 = B(V(y'),\pi+\epsilon),
   \end{align*}
   where $\halpha = \alpha^2\left(\add{t\epsilon}\mdelete{O(\epsilon)}/(2n)\right) \, \big/ \, [ 1 - \alpha \left(\add{t\epsilon}\mdelete{O(\epsilon)}/(2n)\right)]$.	
   Then $B(\hat{M}(y), 2\epsilon) \, \bigcap \, B(\hat{M}(y'), 2\epsilon) \neq \emptyset$ since we are given that $2\dGH(X, Y) < \epsilon$.
   Therefore $\tau \in \eBoxF(Y,\hpi, \halpha)$.
   Since $\lambda = \max \{ \alpha^2\left(\add{t\epsilon}\mdelete{O(\epsilon)}/(2n)\right) \, \big/ \, [ 1 - \alpha \left(\add{t\epsilon}\mdelete{O(\epsilon)}/(2n)\right)], \, 2\epsilon \}$,  
   $\mathcal{C}$ is $\lambda$-simplicial from $\BoxF(X, \pi, \alpha)$ to $\BoxF(Y,\hpi, \halpha)$.
   Similarly $\mathcal{C}^T$ is $\lambda$-simplicial from $\BoxF(Y, \pi, \alpha)$ to $\BoxF(X,\hpi, \halpha)$.
   The result now follows from \cref{prop:chazinterleaveprop}. 
 \end{proof} 

 \add{The persistence module $H(\BoxF(X,\pi,\alpha))$ is $q$-tame since $(X, d_{X})$ is a finite metric space.}

 Finally, we use the following standard result on the bottleneck distance of persistence diagrams of interleaved modules being bounded to obtain the stability result for box filtration in the point cover setting.
 
 \begin{thm}\label{thm:orgstabilitytheorem}(\cite{ChdeSGlOu2016})
   If $\bbU$ is a $q$-tame module then it has a well-defined persistence diagram $\dgm(\bbU)$.
   If $\bbU, \bbV$ are $q$-tame persistence modules that are $\epsilon$-interleaved then there exists an $\epsilon$-matching between the multisets $\dgm(\bbU), \dgm(\bbV)$.
   Thus, the bottleneck distance between the diagrams satisfies the bound $\mathrm{d_b}(\dgm(\bbU), \dgm(\bbV)) \leq \epsilon$. 
 \end{thm}

 \begin{thm}\label{thm:pointstabilitytheorem}
   If $\bbU, \bbV$ are persistence modules using box filtrations for $X, Y$ such that $2\dGH(X, Y) < \epsilon$, then $\mathrm{d_b}(\dgm(\bbU), \dgm(\bbV))\leq \lambda$.
 \end{thm}
 \begin{proof}
   $\bbU, \bbV$ are $q$-tame.
   By \cref{lem:2epsinterleave} $\bbU, \bbV$ are $\lambda-$interleaved.
   Hence by \cref{thm:orgstabilitytheorem}, we get that
   $\mathrm{d_b}(\dgm(\bbU), \dgm(\bbV))\leq \lambda$.   
 \end{proof}

 \subsection{Stability of Pixel Cover Filtration}\label{ssec:pixcostabnoise}

 Let $\bX$ be a pixel discretization of the PCD $X$ with pixels of side length $\clH$.
 Recall that the distance between points $x,x' \in \bX$ is measured as $\bar{d}(x,x')=d(m_{\sigma}, m_{\sigma'})$ for pixels $\sigma, \sigma'$ such that  $x \in \sigma, x' \in \sigma'$.
 The correspondence $\clC: X \rightrightarrows \bX$ is defined by the condition that $\sigma \in \bX$ if $\theta(\sigma)\neq 0$ $\iff$ $\exists x \in X$ such that $x \in \sigma$.
 We denote the pixel cover filtration of $\bX$ as $\pBoxF(\bX,\bpi, \balpha)$.
 As a main step towards showing the stability of the pixel cover filtration, we show that the point cover filtration of $X$ and the pixel cover filtration of $\bX$ are close.

 \begin{prop}\label{prop:pixelcoverstab1}
   $2\dGH(X, \bX) < \clH\sqrt{n}$.
 \end{prop}

 \begin{proof}
   When measuring in the pixel cover setting, the distance between $x, x' \in X$ in any one dimension (as measured in the point cover setting) can increase by at most $\clH$.
   Hence by Minkowski's inequality, we get that
   \[|\bar{d}(x,x') - d(x,x')| = |d(m_{\sigma}, m_{\sigma'}) - d(x, x')| \leq \clH \sqrt{n}.\]
   
   \vspace*{-0.2in}
 \end{proof}

 Based on this Proposition, we fix $\clH = \epsilon/\sqrt{\dmn}$ so that we have $\dGH(X,\bX)=\epsilon$.
 We assume the settings specified in \cref{cor:pixelcoverstab1}.
 This choice of $\clH$ gives $\clK \approx \alpha\epsilon\sqrt{n}(c+ \gamma) + \epsilon^{1/2}/n^{1/4}$ and $\varepsilon \approx \min\{\alpha^2n^{3/4}\epsilon^{1/2}(\gamma +c + (0.5\clL\epsilon^{1/2})/n^{5/4}), 1-\alpha\}$.
 Thus \cref{thm:pointpixelcoverrela} holds for these choices of parameters, and we get the following result.

 \begin{lem}\label{lem:pixelcoverstab5}
   Let $ \clH = \epsilon/\sqrt{\dmn}, \varepsilon \approx \min\{\alpha^2n^{3/4}\epsilon^{1/2}(\gamma +c + (0.5\clL\epsilon^{1/2})/n^{5/4}), 1-\alpha\}, \clK \approx \alpha\epsilon\sqrt{n}(c+ \gamma) + \epsilon^{1/2}/n^{1/4}$, and $\balpha = \alpha + \varepsilon$.
   Suppose $M, \bM$ are the largest optimal solutions in $\sol_{\alpha}(V, N), \psol_{\balpha}(V, N)$, respectively, for a given input box $V$.
   Then	$|M| -|M\cap\bM| < \clK$.
 \end{lem}
 \begin{proof}
   Direct consequence of \cref{thm:pointpixelcoverrela}.
 \end{proof}
  
 \begin{thm}\label{lem:pixel2epsinterleave}
   Let $X$ be a finite subset of $\R^n$ endowed with the induced metric of $\R^n$ and let $\bX$ be a pixel discretization of $X$ such that $\dGH(X,\bX) < \epsilon$ and $0.5\clH\clL \ll 1$. 
   Then the persistence modules $H(\BoxF(X,\pi,\alpha))$ and $H(\pBoxF(\bX,\bpi,\balpha))$ are $O(\bgl)$-interleaved where  $\bgl = \max\{\epsilon^{1/2}, \epsilon\}$.
 \end{thm}

 \begin{proof}
   Let $\mathcal{C}:X \rightrightarrows \bX$ be a correspondence with distortion at most $\epsilon$. 
   To show that $H(\BoxF(X,\pi,\alpha))$ and $H(\pBoxF(\bX,\bpi,\balpha))$ are $O(\bgl)$-interleaved, we show that $\mathcal{C}$ and $\mathcal{C}^T$ are $O(\bgl)$-simplicial. 
   Let $\varrho \in \BoxF(X, \pi, \alpha)$ and let $\tau$ be any finite subset of $\mathcal{C}(\varrho)$. 
   We show that $\tau \in \pBoxF(\bX, \bpi, \balpha)$. 

   For $m_{\sigma}, m_{\sigma'} \in \tau$ pick corresponding vertices $x \in \sigma$ and $x' \in \sigma'$ such that $(x, m_{\sigma})$ and $(x', m_{\sigma'})$ are in $\mathcal{C}(\varrho)$. 
   Let $M(x)$ and $M(x')$ be the largest optimal solutions in $\sol_{\alpha}(V(x), N)$ and $\sol_{\alpha}(V(x'), N')$ where $N = B(V(x), \pi)$ and $N' = B(V(x'), \pi)$, respectively.

   Let $\bM(m_{\sigma}), \bM(m_{\sigma'})$ be the largest optimal solutions in  $\psol_{\balpha}(V(m_{\sigma}), N)$ and $\psol_{\balpha}(V(m_{\sigma'}), N')$, respectively.
   We assume that $V(m_{\sigma}) = V(x)$ and $V(m_{\sigma'}) = V(x')$, i.e., the initial boxes are same in the point and pixel cover settings.
   Then by \cref{lem:pixelcoverstab5}, we get that $|M(x)| - |\bM(m_{\sigma}) \cap M(x)| \leq \clK$ and  $|M(x')| - |\bM(m_{\sigma'}) \cap M(x')| \leq \clK$ for the values of $\balpha, \varepsilon$, and $\clK$ specified in \cref{lem:pixelcoverstab5}.
   And $O(\bgl) =\max\{\varepsilon, \clK\}$.

   Then the versions of $\bM(m_{\sigma})$ and $\bM(m_{\sigma'})$ in $\clK$-$\pBoxF(\bX,\bpi,\balpha)$, the $\clK$-enlargement of $\pBoxF(\bX,\bpi,\balpha)$, contain $M(x)$ and $M(x')$, respectively.
   Hence $\mathcal{C}$ is $O(\bgl)$-simplicial from $\BoxF(X, \pi, \alpha)$ to $\pBoxF(\bX,\bpi, \balpha)$.
   Similarly $\mathcal{C}^T$ is $O(\bgl)$-simplicial from $\pBoxF(\bX, \pi, \alpha)$ to $\BoxF(X,\bpi, \balpha)$.
   The result now follows from \cref{prop:chazinterleaveprop}. 
 \end{proof}

 \add{The persistence module $H(\pBoxF(\bX,\pi,\alpha))$ is $q$-tame since $(\bX, d_{\bX})$ is a finite metric space.}

 \begin{thm}\label{thm:pointpixelstab}
   If $\bbU, \bar{\bbU}$ are persistence modules of the box filtrations for $X, \bX$ such that $2\dGH(X, \bX) < \epsilon$ then, $\mathrm{d_b}(\dgm(\bbU), \dgm(\bar{\bbU}))\leq \mdelete{O(\bgl)}\add{t\bgl}$ \add{for some constant $t$}.
 \end{thm}
 \begin{proof}
   $\bbU, \bar{\bbU}$ are $q$-tame.
   By \cref{lem:pixel2epsinterleave} $\bbU, \bar{\bbU}$ are $O(\bgl)-$interleaved.
   Hence by \cref{thm:orgstabilitytheorem}, we get that
   $\mathrm{d_b}(\dgm(\bbU), \dgm(\bar{\bbU}))\leq \mdelete{O(\bgl)}\add{t\bgl}$ \add{for some constant $t$}.   
 \end{proof}

 \begin{thm}\label{thm:pointpixelstabilitytheorem}
   Let $X, \bX, Y, \bY$ are finite metric spaces where  $\bX$ and $\bY$ are pixel discretization of $X$ and $Y$, respectively.
   Also, let $\bbU, \bar{\bbU}, \bbV, \bar{\bbV}$ be the persistence modules of $X, \bX, Y, \bY$, respectively.
   Then we have 
   $\mathrm{d_b}(\dgm(\bar{\bbU}), \dgm(\bbV)) \leq \mdelete{O(\bgl)}\add{t\bgl} + \lambda$,
   $\mathrm{d_b}(\dgm(\bbU), \dgm(\bar{\bbV})) \leq \mdelete{O(\bgl)}\add{t\bgl} + \lambda$, and 
   $\mathrm{d_b}(\dgm(\bar{\bbU}), \dgm(\bar{\bbV})) \leq \mdelete{O(\bgl)}\add{t\bgl} + \lambda$ \add{ for some constant $t$}.	
 \end{thm}
 \begin{proof}
   By \cref{thm:pointstabilitytheorem,thm:pointpixelstab} we know $\mathrm{d_b}(\dgm(\bbU), \dgm(\bbV))\leq \lambda$, $\mathrm{d_b}(\dgm(\bbU), \dgm(\bar{\bbU}))\leq \mdelete{O(\bgl)}\add{t\bgl}$, and $\mathrm{d_b}(\dgm(\bbV), \dgm(\bar{\bbV}))\leq \mdelete{O(\bgl)}\add{t\bgl}$.
   Since the bottleneck distance satisfies metric axioms, we get that \\
   $\mathrm{d_b}(\dgm(\bar{\bbU}), \dgm(\bbV)) \leq \mdelete{O(\bgl)}\add{t\bgl} + \lambda$,
   $\mathrm{d_b}(\dgm(\bbU), \dgm(\bar{\bbV})) \leq \mdelete{O(\bgl)}\add{t\bgl} + \lambda$, and 
   $\mathrm{d_b}(\dgm(\bar{\bbU}), \dgm(\bar{\bbV})) \leq \mdelete{O(\bgl)}\add{t\bgl} + \lambda$.
 \end{proof}

\subsection{Sensitivity to Parameter Choices}\label{ssec:snstparam}

\paragraph{Parameter $\alpha$:} Based on \cref{lem:pointcountalpharelaappendixmod}, a small change in $\alpha$ implies small changes in the lower and upper bounds of $\theta(N\setminus M)$.
And this implies small changes in the lower and upper bounds of $\theta(M)$ as well.
Furthermore, when $\alpha$ is increased by a small value, the new \delete{largest} optimal \delete{solution} \add{box } will contain the previous \add{largest optimal box by \cref{lem:alphachangeforx}} \delete{one}.
Similarly,  by \cref{lem:alphachangeforx}, a small decrease in $\alpha$ will result in a new largest optimal \delete{solution} \add{box } that is contained in the previous \add{optimal box} \delete{one}.
We get these results for pixel cover as well, since \cref{prop:unioncosteq} is true also for pixel discretization.   

\paragraph{Parameter $\pi$:} 
A small change in $\pi$ \add{implies a small number of new points} \delete{added  a small change in $\theta(N\setminus M)$} based on \cref{proper:ptdsbtnsos}.
This implies that the change in $\theta(M)$ is small.
Furthermore, when $\pi$ increases by a small value, the new \delete{largest} optimal \add{box} \delete{solution} contains the previous \add{largest optimal box} \delete{one} by \cref{lem:pichangeforx}.
Similarly, a small decrease in $\pi$ results in a new largest optimal \add{box} \delete{solution} that is contained in the previous \add{optimal box} \delete{one} by \cref{lem:pichangeforx}.
Again, we get similar results for pixel cover as well, since \cref{prop:unioncosteq} is true also for pixel discretization.

While we get these inclusions of largest optimal solutions, the sizes of the largest boxes could change significantly.
For instance, consider \cref{exm:notunique} with $\alpha < 0.5$.
Then optimal solution is $[a, a]$ for the input box $V = [a, a]$ and neighborhood $B(a, b-a+\delta)$ for $\delta > 0, \delta \approx 0$, since $\dfrac{\partial C(\tV, \alpha)}{\partial x} > 0$ for $\tV \supset V$.
If $\alpha = 0.5$, then $[a, b]$ is the largest optimal solution in the neighborhood $B(a, b-a+\delta)$.
Now, if $b \gg a$ then the difference in total width of the optimal solutions is $(b - a) \gg 0$ even for a small change in $\alpha$ from $0.5-\delta$ to $0.5$.

Now, consider \cref{exm:notunique} with $\alpha = 0.5$.
Then the optimal solution is $[a, a]$ for input box $V = [a, a]$ and neighborhood $B(a, b-a)$.
And the largest optimal solution is $[a, b]$ for input box $V = [a, a]$ with neighborhood $B(a, b-a + \delta)$ for $\delta > 0, \delta \approx 0$.
Now, if $b \gg a$ then difference in total width of the optimal solutions is $(b - a) \gg 0$ even for a $\delta \approx 0$ change in the parameter $\pi$.
Thus, for a small change in the $\pi$ parameter we can have a large change in the optimal solution.

%% file: examples.tex
\section{Examples}\label{sec:examples}

We present several 2D examples on which we compare the performances of Vietoris-Rips (VR) and distance-to-measure (DTM) filtrations with that of the box filtration.

\subsection{Point Clouds and Comparison of Filtrations} \label{ssec:pcds}

The four point clouds are shown in \cref{fig:entropyandpmaxexampleh1}, and details highlighted below.

\begin{figure}[hb!] 
  \centering
  \begin{subfigure}[t]{1.7in}
    \centering
    \includegraphics[width=1.7in]{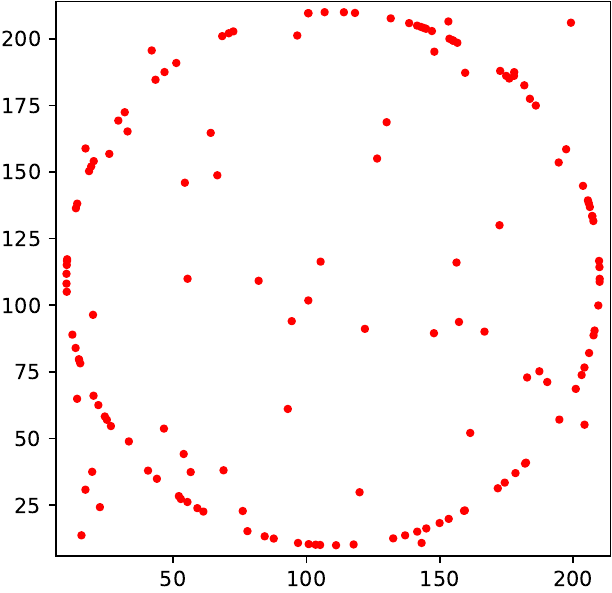}
    \caption{\label{fig:testexample2}}
  \end{subfigure}
  \quad
  \begin{subfigure}[t]{1.7in}
    \centering
    \includegraphics[width=1.7in]{noisyThinCircleResults/dataPoints_crop}
    \caption{\label{fig:testexample1}}
  \end{subfigure}\\
  \smallskip
  \begin{subfigure}[t]{1.7in}
    \centering
    \includegraphics[width=1.7in]{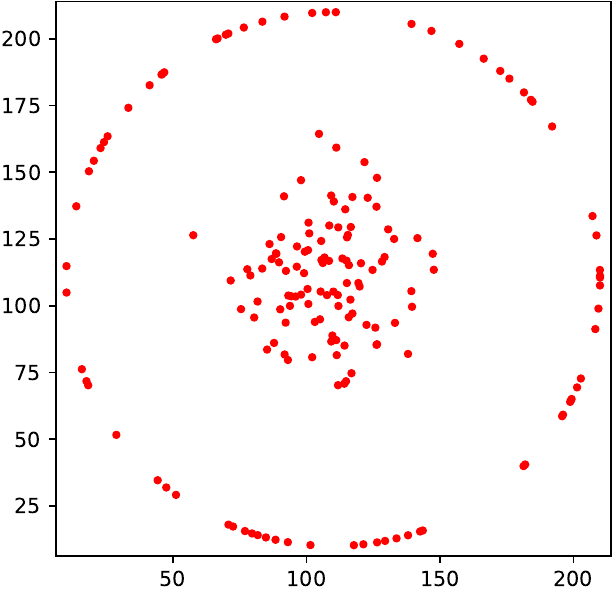}
    \caption{\label{fig:testexample3}}
  \end{subfigure}
  \quad
  \begin{subfigure}[t]{1.7in}
    \centering
    \includegraphics[width=1.7in]{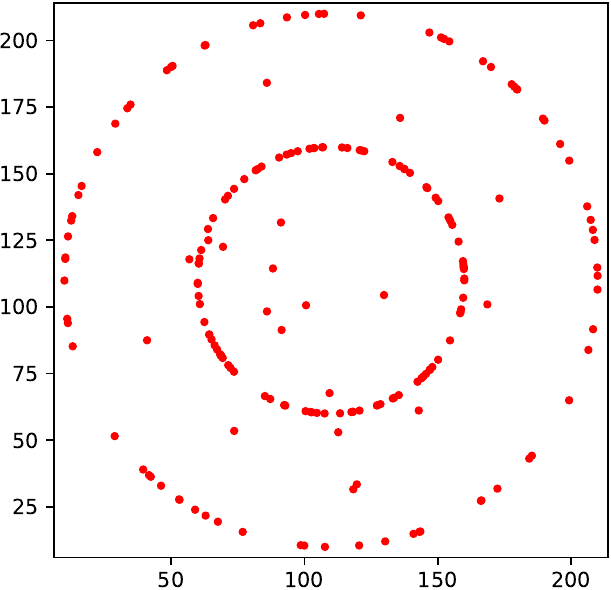}
    \caption{\label{fig:testexample4}}
  \end{subfigure}

  \caption{\label{fig:entropyandpmaxexampleh1} The four PCDs considered in our study (see \cref{ssec:pcds} for details).}
\end{figure}

\begin{enumerate}[leftmargin=*,itemsep=0ex,label=(\alph*)]
  \item {\bfseries Noisy circle:}
    The PCD consists of 100 points sampled from the uniform distribution on the unit circle centered at $(1, 1)$ and 50 points sampled from $[0, 2]^2$.
    We then scaled the PCD by $100$, 
    as shown in \cref{fig:testexample2}.
    The $1$-homology ($H_1$) persistence diagrams for VR, DTM, and BF are shown in \cref{fig:testexample2vr}, \cref{tab:noisycircleDtm}, and \cref{tab:noisycircleBf}.

    In the VR filtration PD, the $H_1$ feature of the circle shows only a small separation from the noise (see \cref{fig:testexample2vr}) due to the the presence of outliers.
    The DTM filtration PDs show more separation of the $H_1$ feature of circle from noise compared to VR for the values $\{0.1, 0.2, 0.3\}$ of its parameter $m$ (see \cref{tab:noisycircleDtm}).
    On the other hand, the BF PDs separate the $H_1$ feature of the circle from the noise for most values of its parameter $\alpha= \{0.2, 0.3, 0.4, 0.5, 0.6, 0.7, 0.8, 0.9\}$ (see \cref{tab:noisycircleBf}). 

  \item {\bfseries Noisy ellipse:}
    The PCD consists of 100 points sampled from a uniform distribution on the unit circle centered at $(1, 1)$ and 50 points sampled from $[0, 2]^2$.
    Then the $y$-coordinates of the points are scaled up by $20$ and the $x$-coordinate by $100$.
    Finally, the scaled PCD is rotated by $45\deg$, as shown in \cref{fig:testexample1}.
    The $H_1$-persistence diagrams of VR, DTM, and BF are shown in \cref{fig:testexample1vr}, \cref{tab:noisethincircleDtm}, and \cref{tab:noisethincircleBf}, respectively.

    In VR filtration persistence diagram (PD), the $H_1$ feature of the ellipse shows only a small separation from the rest of the points (see \cref{fig:testexample1vr}).
    Among all the values considered for the parameter $m \in [0, 1)$ for the DTM filtration, only $m = 0.2$ shows a clear separation of the $H_1$ feature of the ellipse from the rest of the points (see \cref{tab:noisethincircleDtm}).
      On the other hand, for the BF PDs, we observe the $H_1$ feature of the ellipse clearly separated from the rest of the points for a majority of choices $\{0.1, 0.5, 0.6, 0.7, 0.8, 0.9\}$ of its parameter $\alpha$ (see \cref{tab:noisethincircleBf}).

      The ellipse has more variation in one direction (its major axis) than the other.
      Since using Euclidean balls leads to symmetry bias, going from the noisy circle (\cref{fig:testexample2}) to the noisy ellipse (\cref{fig:testexample1}) decreases the separation of the $H_1$ feature from the noise in the case of DTM and VR.
      At the same time, in the case of BF we can still see clear separation of the $H_1$ feature and noise for many choices of its parameter.

    \item {\bfseries Circle with central cluster:} %
      The PCD consists of 75 points sampled from a uniform distribution on the unit circle centered at $(1, 1)$ and 100 points sampled from a cluster centered at $(1, 1)$ with a noise of $0.2$.
      The PCD is then scaled by $100$, as shown in \cref{fig:testexample3}.
      The $H_1$ PDs of VR, DTM, and BF are shown in \cref{fig:testexample3vr}, \cref{tab:noisecirclewithcentralclusterDtm}, and \cref{tab:noisecirclewithcentralclusterBf}, respectively.
      VR filtration PD was able to identify $H_1$ feature of the circle but shows small separation from the noise (see \cref{fig:testexample3vr}).
      DTM fails to identify the circle for any value of its parameter $m$ (see \cref{tab:noisecirclewithcentralclusterDtm}).
      On the other hand, we observe the $H_1$ feature of circle showing more separation from the noise in the BF PDs for the choices of $\alpha$ in $ \{0.2, 0.6, 0.7, 0.8, 0.9\}$ (see \cref{tab:noisecirclewithcentralclusterBf}). 

      This dataset is similar to the noisy circle of same radius, but random noise is replaced by a dense cluster in the middle.
      For DTM, balls in the middle grow faster than the balls centered at the points on the circle because it uses density as a measure to grow balls. 
      This behavior along with symmetry bias due to the use of balls significantly reduces the life of the hole.
      For BF too, the boxes belonging to the central cluster points grow faster than those belonging to points on the circle.
      But the boxes belonging to the cluster should cover points in the cluster before covering points in the circle based on the definition of BF and the non-symmetric nature of boxes.
      The VR filtration was able to identify the $H_1$ feature of the circle but also shows several other $H_1$ features corresponding to noise (see \cref{fig:testexample3vr}) whereas BF persistence diagram has significantly less $H_1$ features corresponding to noise.
      But more generally, VR is sensitive to outliers as shown in \cref{fig:testexample1vr}.   
      
    \item {\bfseries Concentric circles with noise:}
      The PCD is sampled from two concentric circles with 100 points from the inner circle and 75 points from the outer circle, along with 20 random points sampled from $[0, 2]^2$ as shown in \cref{fig:testexample4}.
      The $H_1$-Persistence diagrams of VR, DTM, and BF are shown in \cref{fig:testexample4vr}, \cref{tab:noisecirclewithcentralclusterwithsquareDtm}, and \cref{tab:noisecirclewithcentralclusterwithsquareBf}, respectively.

      VR was not able to identify the two circles (see \cref{fig:testexample4vr}) since it is sensitive to noise.
      For DTM, the balls corresponding to the inner circle grow relatively faster than those corresponding to the outer circle since inner circle has a higher density of points.
      The balls also induce symmetry bias.
      While the DTM filtrations show the inner circle as a $H_1$ feature with significant life, the outer circle feature is not clearly visible for all values of its parameter (see \cref{tab:noisecirclewithcentralclusterwithsquareDtm}).
      On the other hand, the BF PDs shows two significant $H_1$ features for $\alpha = \{0.5, 0.6, 0.7, 0.8, 0.9\}$.
      While the boxes corresponding to the inner circle grow faster than those corresponding to the outer circle, they grow non-symmetrically to better capture both circles (see \cref{tab:noisecirclewithcentralclusterwithsquareBf}).

\end{enumerate}

\subsection{Point Clouds with Noise} \label{ssec:pcdswnoise}
\vspace*{-0.02in}	

We further explore how the filtrations perform on the data sets under increasing levels of noise.
For each class, we create three versions of the PCD by adding Gaussian noise with a mean of $0$ and standard deviations of $0, 2, 5$.
The first version with standard deviation of $0$ is just the original PCD.
Each version of the PCD is represented by its respective $H_1$-persistence diagram and we use the bottleneck distance between the PDs to cluster the 12 data sets (3 versions each of the 4 point clouds).
Ideally, we should see four clusters with 3 points each corresponding to the respective noisy versions of each class.  
	
We first compute the bottleneck distance between each pair of the $H_1$ PDs of the 12 PCDs.
To visualize how well these distances separate the PCDs, we then use multidimensional scaling (MDS) \cite{CoCo2000} to obtain a 2D representation of the data set.
Finally, we apply $k$-means clustering with $k=4$ to the data set.
We repeat the steps for BF and DTM for the respective parameters $\alpha, m$ set as $\{0.1, 0.2, 0.3, 0.4, 0.5, 0.6, 0.7, 0.8, 0.9\}$.
As a measure of effectiveness of the clustering, we compute the Rand score \cite{HuAr1985} for each choice of $\alpha$ and $m$ (see \cref{fig:clusterexamplebfdtm}).
The Rand score between a pair of clusterings of a given data set is defined as the fraction (number of agreeing pairs)/(number of pairs).
A score of $1$ indicates that all pairs are agreeable, i.e., they are either in the same cluster in both clusterings or are in different clusters in both clusterings
(a score of $0$ means no pair is agreeable).
As shown in \cref{fig:clusterexample}, the Rand score for BF is higher that or equal to that for DTM for all choices of parameters.
For $\alpha=0.6$, BF separates all $4$ classes of data into unique clusters as shown in \cref{fig:clusterexamplebf}.
On the other hand, DTM was not able to fully separate the classes of data for any choice of $m$
(clustering for $m=0.6$ is shown in \cref{fig:clusterexampledtm} as an example).
To avoid the dependence on the choice of $\alpha$ or $m$, we repeat the analysis using the pairwise distance set as the \emph{sum} of distances for all choices of the parameters (see \cref{fig:clusterconcatexamplebf} and \cref{fig:clusterconcatexampledtm}).
BF is able to separate all classes into unique clusters except for one PCD of class 3 (circle with central clusters) whereas DTM was not able to separate classes 2 and 3 (noisy ellipse and circle with central cluster).   

\begin{figure}[htp!] 
	\centering
	\begin{subfigure}[t]{2.2in}
		\centering
		\includegraphics[width=2.2in]{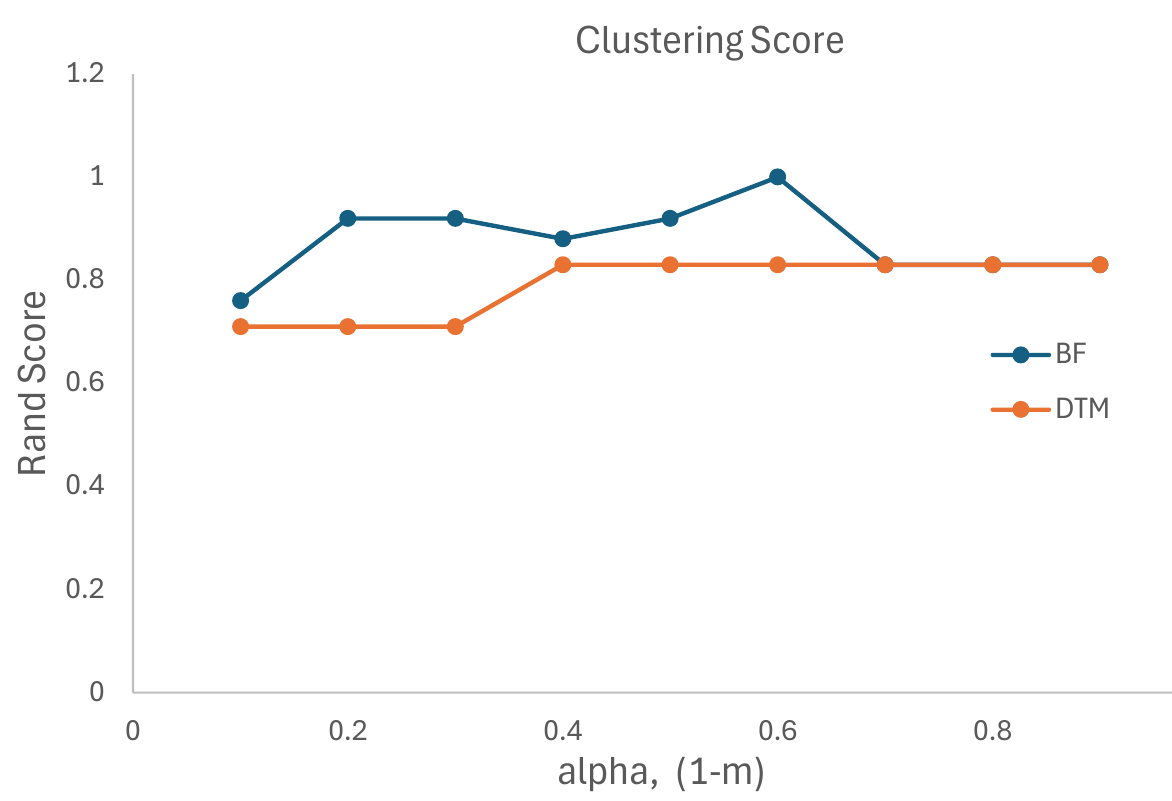}
		\caption{\label{fig:clusterexample}}
	\end{subfigure}
	\begin{subfigure}[t]{2.1in}
		\centering
		\includegraphics[width=2.1in]{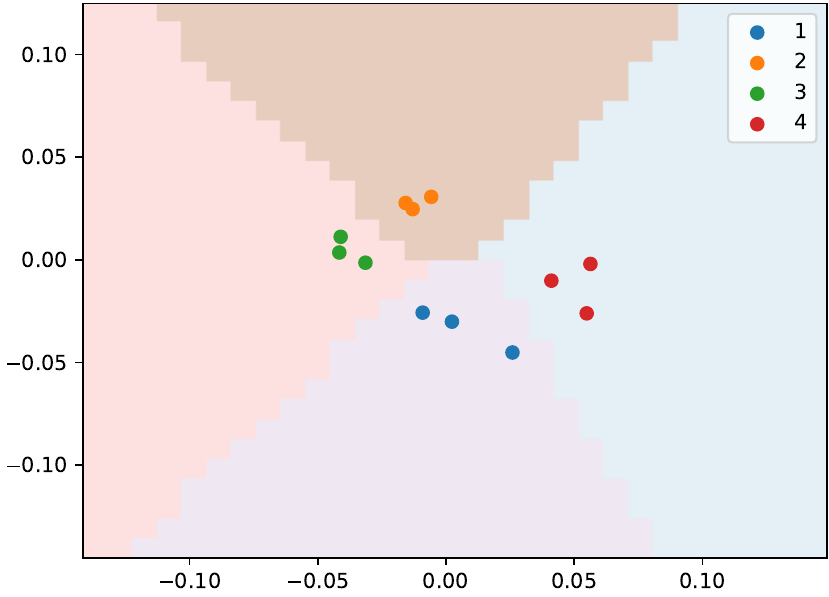}
		\caption{\label{fig:clusterexamplebf}}
	\end{subfigure}
	\begin{subfigure}[t]{2.1in}
		\centering
		\includegraphics[width=2.1in]{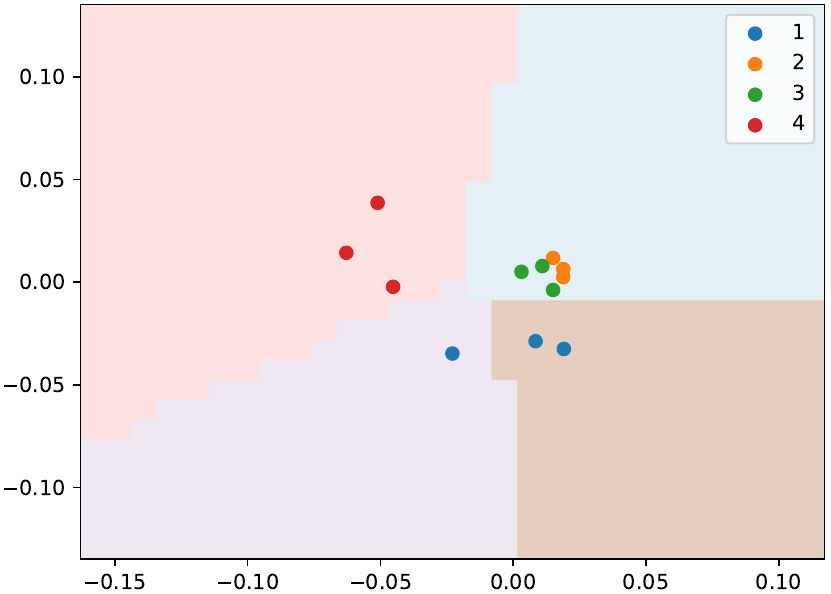}
		\caption{\label{fig:clusterexampledtm}}
	\end{subfigure}
	\begin{subfigure}[t]{2.1in}
		\centering
		\includegraphics[width=2.1in]{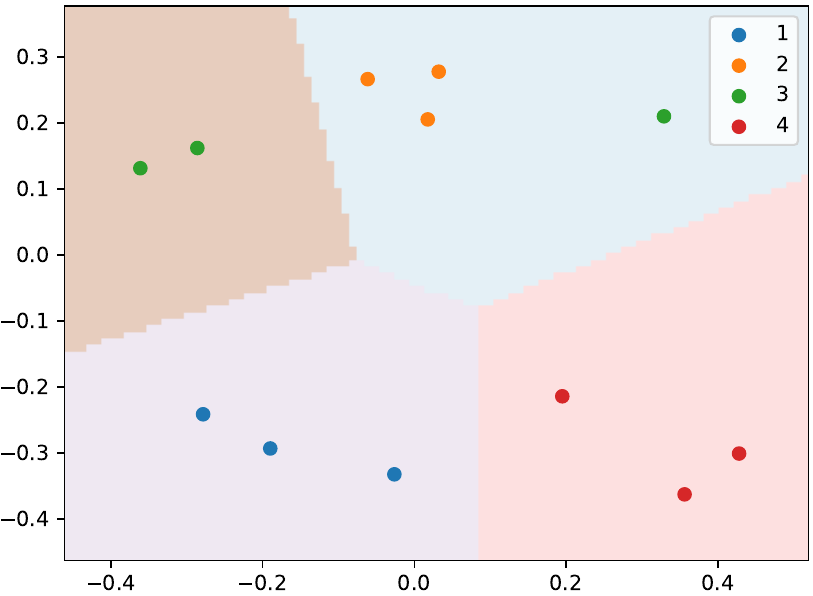}
		\caption{\label{fig:clusterconcatexamplebf}}
	\end{subfigure}
	\begin{subfigure}[t]{2.1in}
		\centering
		\includegraphics[width=2.1in]{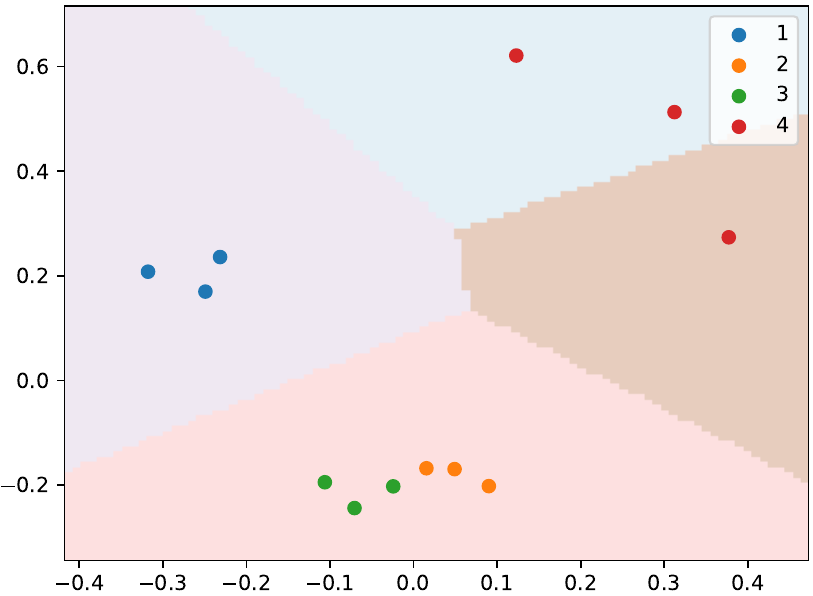}
		\caption{\label{fig:clusterconcatexampledtm}}
	\end{subfigure}
	
	\caption{\label{fig:clusterexamplebfdtm}
          \cref{fig:clusterexample} shows Rand score for choices of $\alpha$ for BF and $1-m$ for DTM.
          \cref{fig:clusterexamplebf} and \cref{fig:clusterexampledtm} show the MDS representation of BF and DTM for $\alpha=0.6$, $m=0.6$, with legends $1, 2, 3, 4$ representing noisy circle, noisy ellipse, circle with central cluster, and concentric circles with noise.
          \cref{fig:clusterconcatexamplebf} and \cref{fig:clusterconcatexampledtm} show the MDS representations using total distance from all parameters for BF and DTM.}
\end{figure}

%% file: mapper.tex
\section{Box Mapper: A Mapper Algorithm Using Box Filtration} \label{sec:mapper}

The mapper provides a compact summarization of a PCD $X$.
Let $f : X \to Z$ be a continuous \emph{filter function} and $\clU=\{U_{\alpha}\}$ a finite cover of $Z$.
Then the connected components of $f^{-1}(U_{\alpha}) = \bigcup\nolimits_{i = 1}^{k_{\alpha}} V_{\alpha, i}$ form a cover $f^{*}(\clU)$ of $X$.
The \emph{mapper} of $X$ is then defined \cite{SiMeCa2007} as its nerve: 
$M(\clU, f) := \Nrv(f^{*}(\clU))$.
By default, $Z$ is taken as the range space containing $X$, i.e., the bounded $\dmn$-dimensional box containing $X$.
And the cover $\clU$ is usually comprised of hypercubes that overlap.
The length of the hypercubes (resolution) and their overlap percentage (gain) as well as the the clustering algorithm to find the connected components of the pullback cover elements are chosen by the user.
While the framework of persistent homology has been used to derive theoretical stability results for mapper constructions \cite{CaOu2017,DeMeWa2016}, implementations of such constructions are not known.
In practice, users work with a single mapper constructed for specific choices of parameters \cite{CaMiOu2018}.

We note that any simplicial complex in a box filtration of $X$ built using a pixel cover automatically gives a mapper of $X$.
Our framework naturally avoids any pixels $\sigma$ that do not contain points from $X$ ($\theta(\sigma)=0$) from further consideration, thus providing computational savings.
The same observation could be made when using a point cover as well, under the modification that the union of boxes may not form a cover of $Z$ at all stages of growth.
At the same time, all relevant portions of $Z$, i.e., regions that have points from $X$ within them, are always covered by the union of boxes.
Furthermore, our stability results for box filtration (\cref{thm:pointstabilitytheorem,thm:pointpixelstabilitytheorem}) also imply a stability for the associated mapper constructions.

With quick applicability on large PCDs in mind, we present a mapper algorithm that uses a single growth step per box.
We start by applying $k$-means clustering \cite{HaWo1979} to $X$.
We then find the minimal box enclosing each cluster.
Using these $k$ boxes as the pivot boxes, we apply the pixel cover box filtration framework for a single value of $\pi$ and output the nerve of the enlarged boxes as the box mapper of $X$.
Unlike the conventional mapper, we let these optimizations determine the sizes of individual boxes as well as their overlaps.
\cref{fig:shapeFilter} shows two instances of the box mapper constructed on PCDs of \emph{elephant} and \emph{flamingo} with 42321 and 26907 points, respectively \cite{SuPo2004}.
We used $k=80$ for elephant and $k=25$ for flamingo, with $\pi=3$ and $\alpha=0.1$ for both cases.
\begin{figure}[hb!]
  \centering
  \includegraphics[scale=0.57]{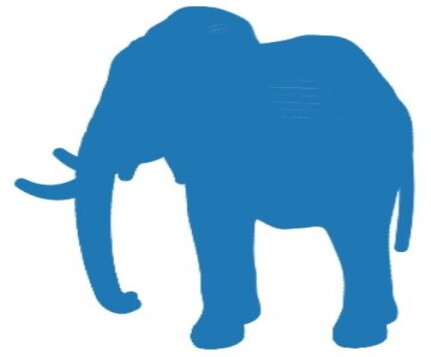} \hspace*{0.6in}
  \includegraphics[height=1.83in, width=1.83in]{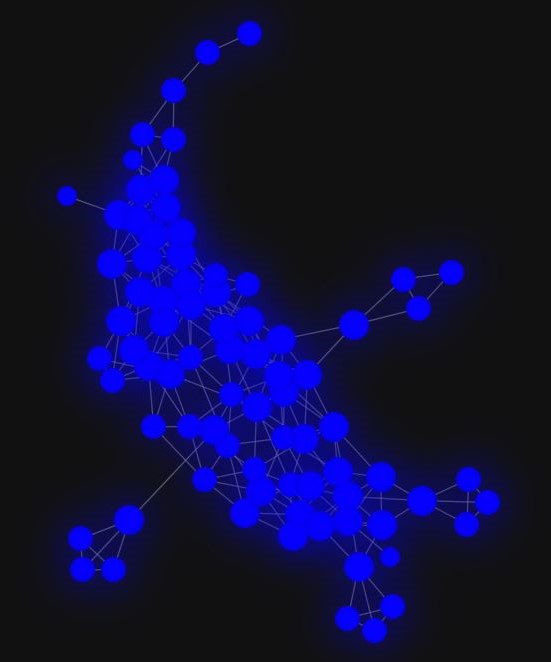} \\
  \quad\quad\quad~ \includegraphics[scale=0.50,angle=90]{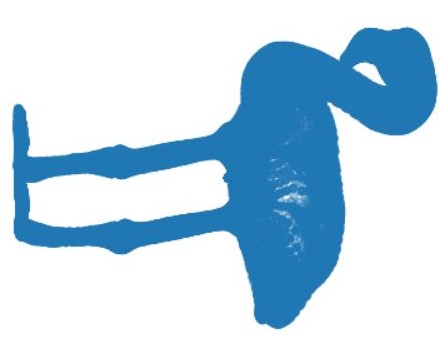} \hspace*{0.63in}
  \includegraphics[height=1.83in, width=1.83in]{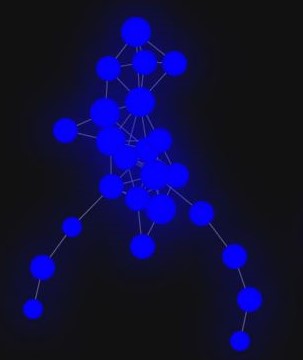}
  \caption{\label{fig:shapeFilter}
    Point cloud (left) and box mapper (right) for the elephant and flamingo PCDs.}
\end{figure}

%% file: disc.tex
\section{Discussion} \label{Sec:disc}

The main bottleneck for computing box filtrations is the solution of linear programs that determine the extent of growth of each box at each growth step of the filtration.
While we currently model this growth step as a linear program, more efficient approaches could be developed for the same.
\add{Since the cost function is additive in nature, we have shown stability with respect to Gromov-Hausdorff distance.
  We would like to explore stability with respect to Gromov-Wasserstein distance in future work.}

While pixel covers offer computational efficiency over point covers, using too large of a pixel size could result in smaller scale features being missed.
It could be useful to identify guidelines for choosing the pixel size based on properties of the PCD.
While we presented the box filtration for PCDs, it can be naturally adapted to build sublevel set filtrations.
It would be interesting to consider extending the box filtration approach to multiparameter persistence.

%% file: app.tex
\renewcommand\appendixpagename{Appendix}
\appendixpage

\begin{appendices}

\section*{Persistence Diagrams for PCDs} \label{appsec:PDs}

\begin{figure}[htp!] 
  \centering
  \begin{subfigure}[t]{1.4in}
  \centering
  \includegraphics[width=1.4in]{noisyThinCircleResults/vrPersistence_crop}
  \caption{\label{fig:testexample1vr}}
  \end{subfigure}
  \begin{subfigure}[t]{1.4in}
  \centering
  \includegraphics[width=1.4in]{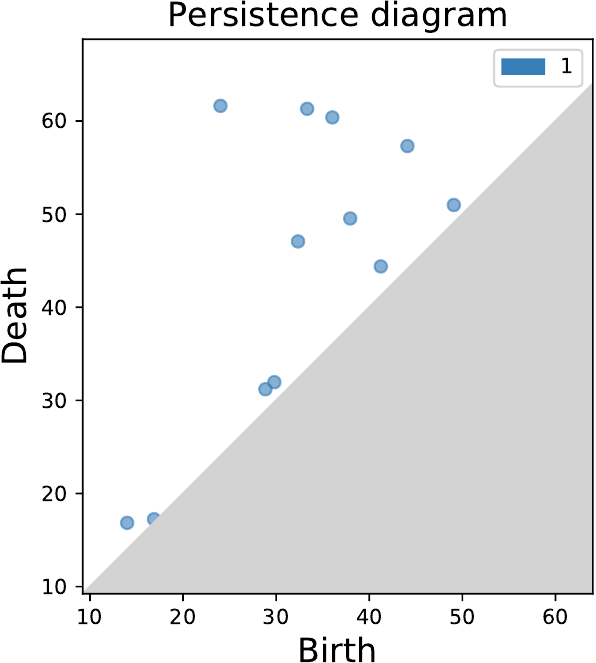}
  \caption{\label{fig:testexample2vr}}
  \end{subfigure}
  \begin{subfigure}[t]{1.4in}
  \centering
  \includegraphics[width=1.4in]{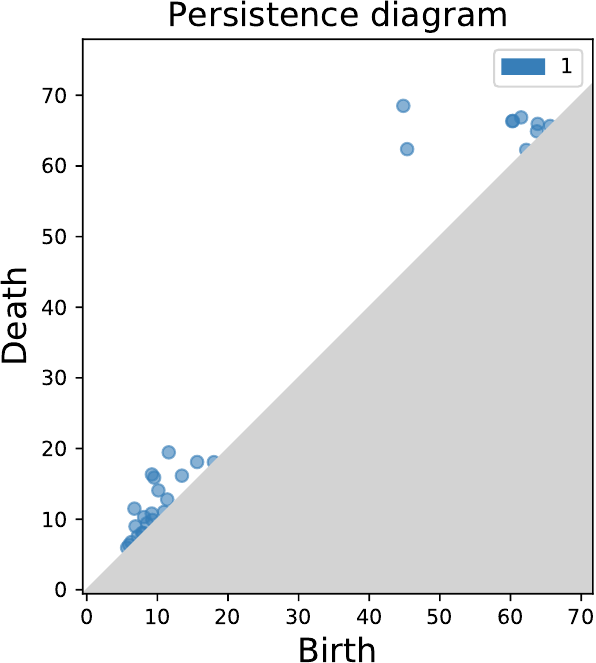}
  \caption{\label{fig:testexample3vr}}
  \end{subfigure}
  \begin{subfigure}[t]{1.4in}
  \centering
  \includegraphics[width=1.4in]{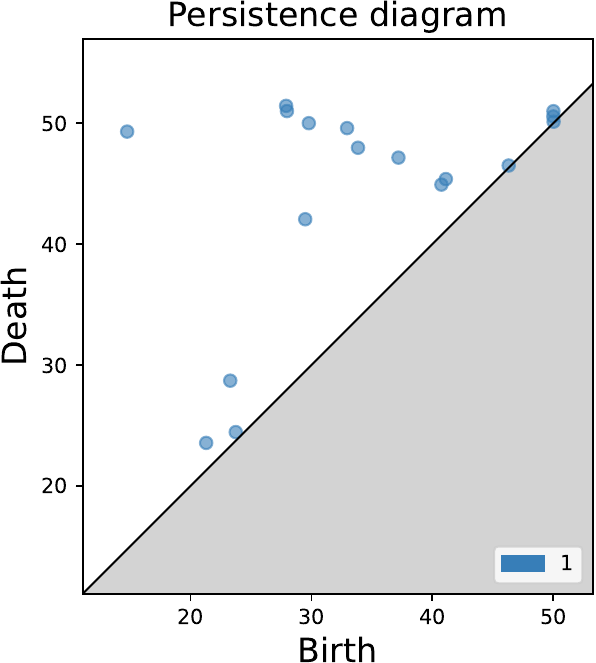}
  \caption{\label{fig:testexample4vr}}
  \end{subfigure}
  \caption{\label{fig:testexamplevr}
  Vietoris-Rips PDs for the four PCDs in \cref{fig:entropyandpmaxexampleh1}.
  }
\end{figure}

\begin{figure}[htp!]
  \begin{center}
  \begin{tabular}{m{4.5cm} m{4.5cm} m{4.5cm} m{4.5cm}}
  \includegraphics[width=1.8in]{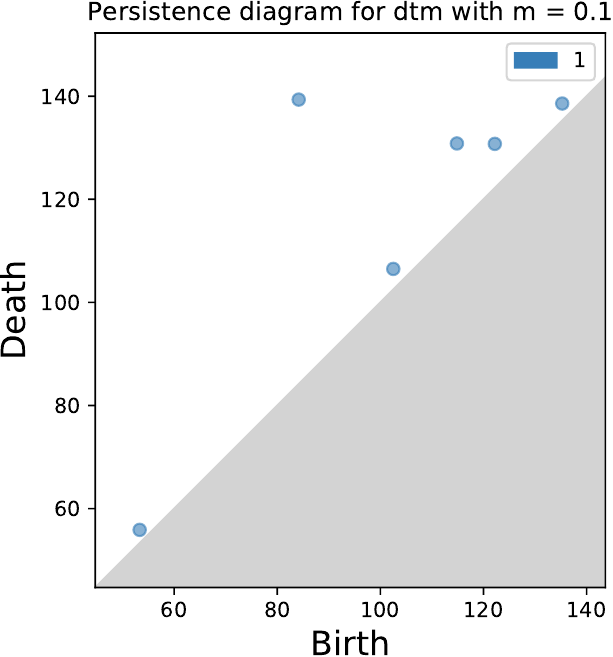}&
  \includegraphics[width=1.8in]{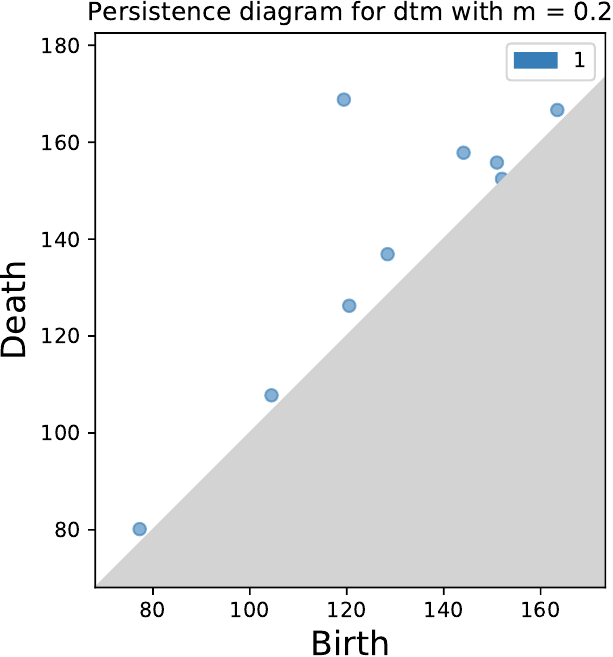}&
  \includegraphics[width=1.8in]{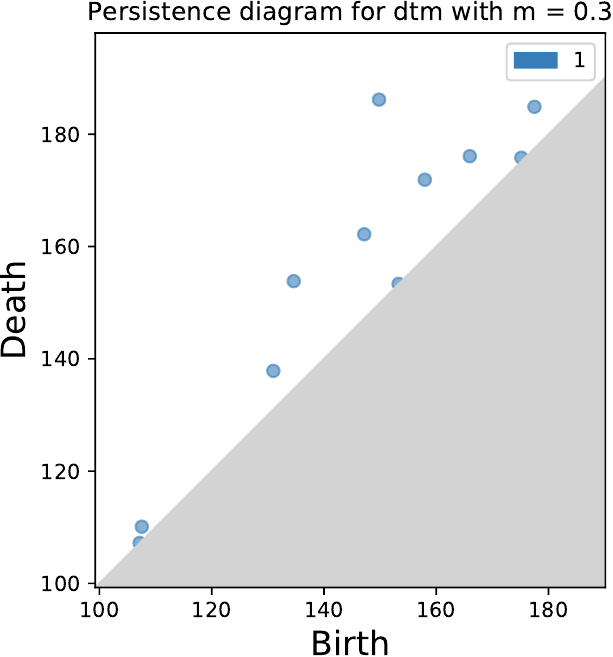}\\
  \includegraphics[width=1.8in]{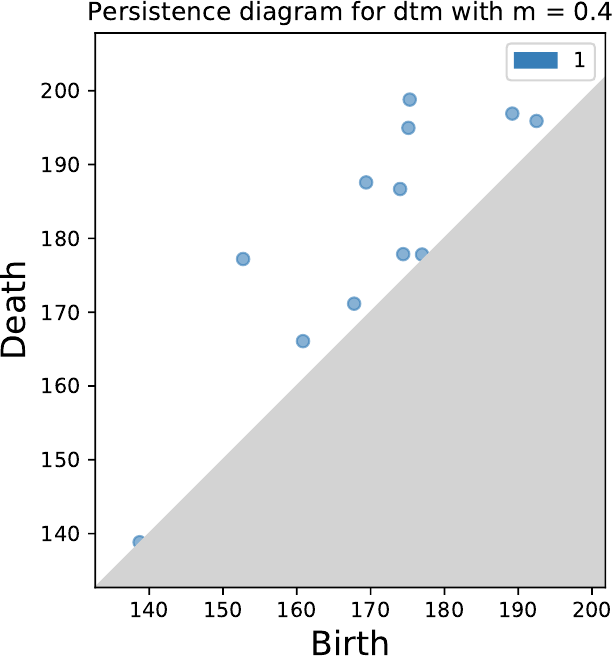}&
  \includegraphics[width=1.8in]{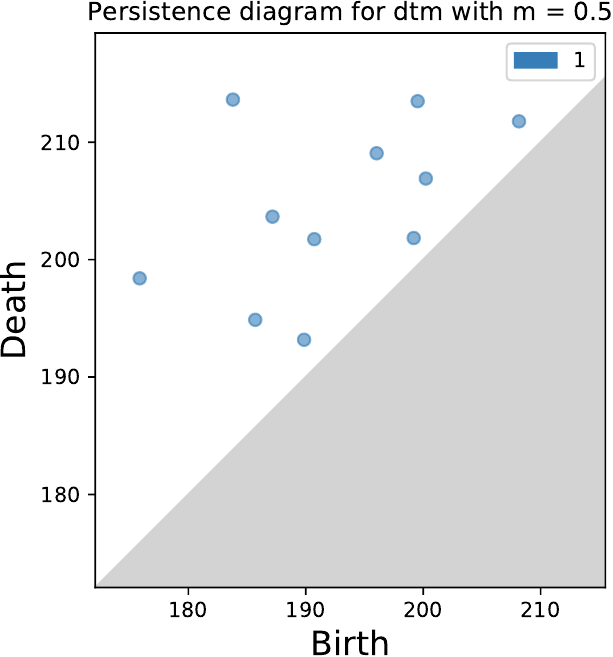}&
  \includegraphics[width=1.8in]{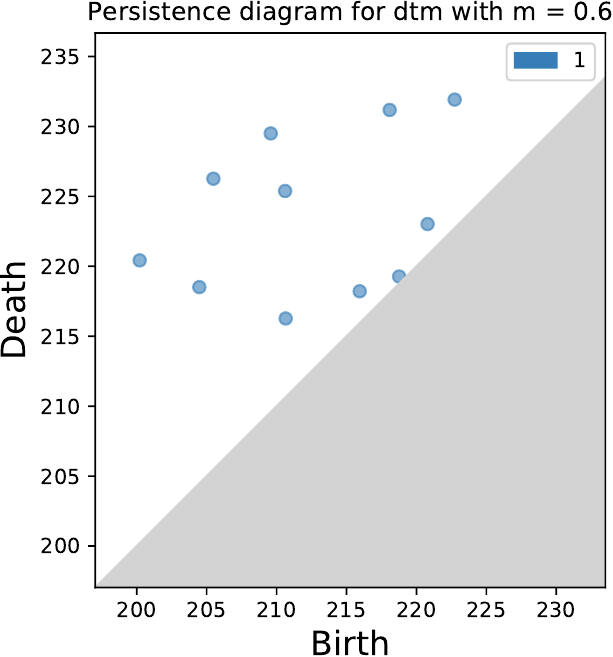}\\
  \includegraphics[width=1.8in]{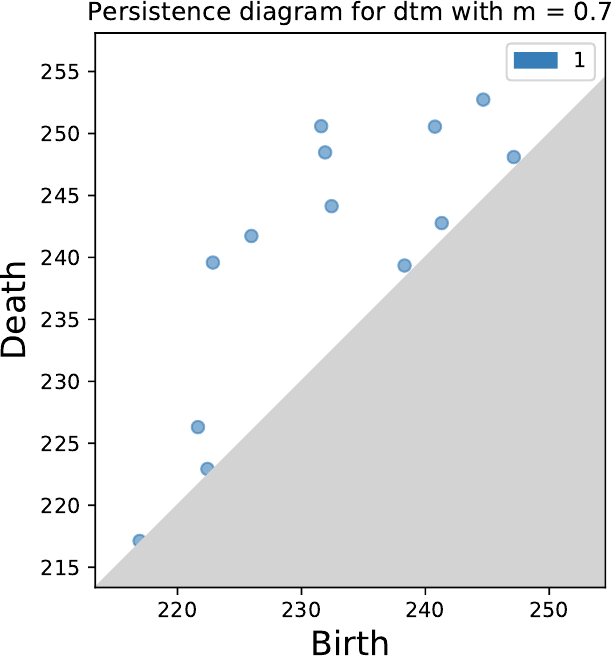}&
  \includegraphics[width=1.8in]{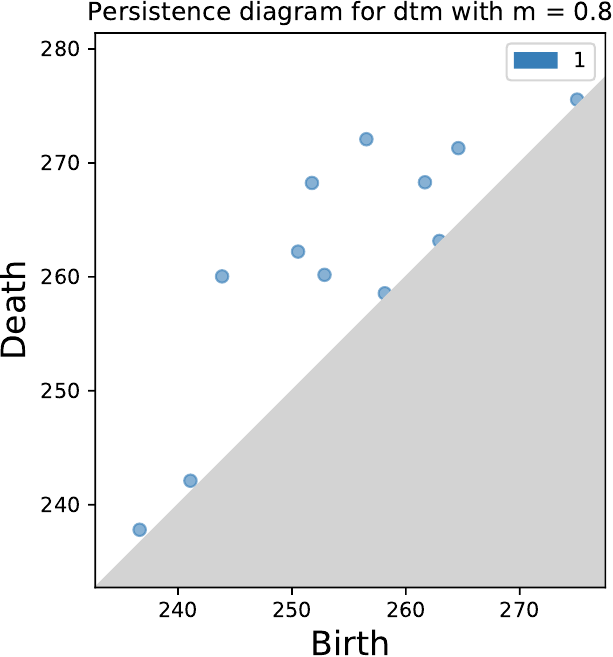}&
  \includegraphics[width=1.8in]{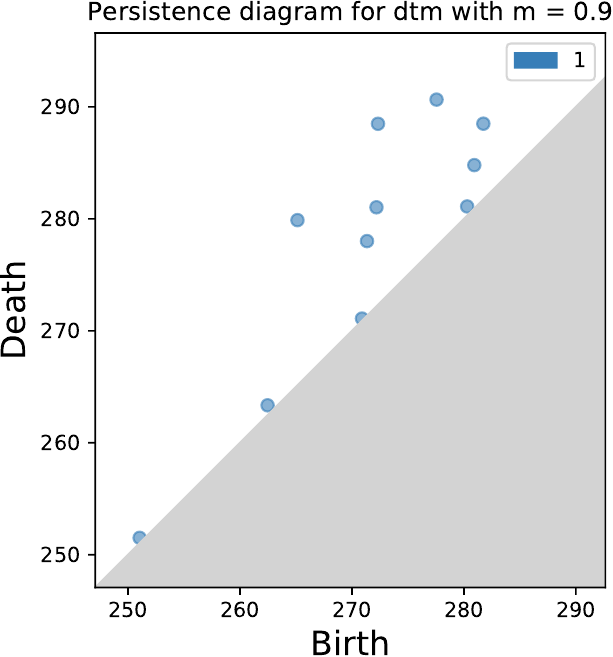}
  \end{tabular}
  \caption{\label{tab:noisycircleDtm} PDs of DTM applied to the noisy circle PCD for $m=0.1$--$0.9$.}
  \end{center}
\end{figure}

\begin{figure}[htp!]
  \begin{center}
  \begin{tabular}{m{4.5cm} m{4.5cm} m{4.5cm} m{4.5cm}}
  \includegraphics[width=1.8in]{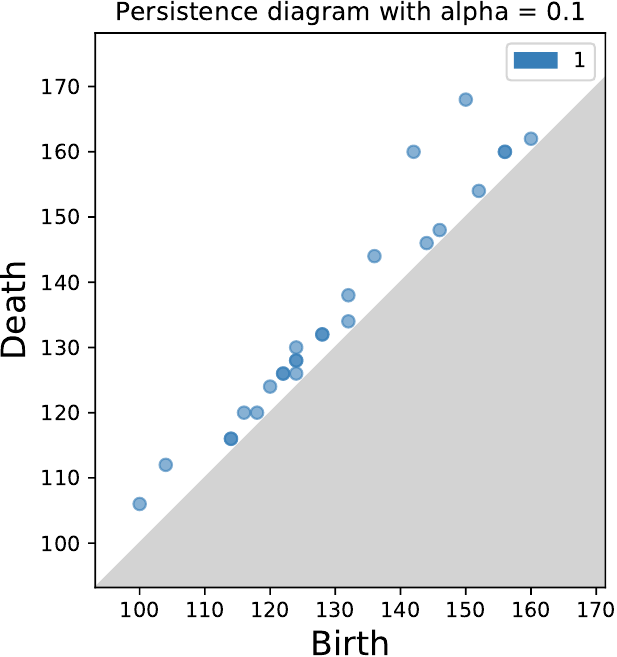}&
  \includegraphics[width=1.8in]{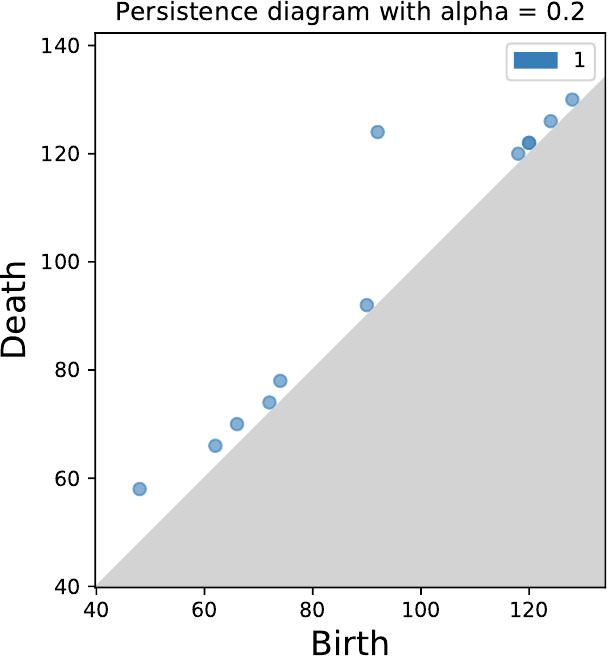}&
  \includegraphics[width=1.8in]{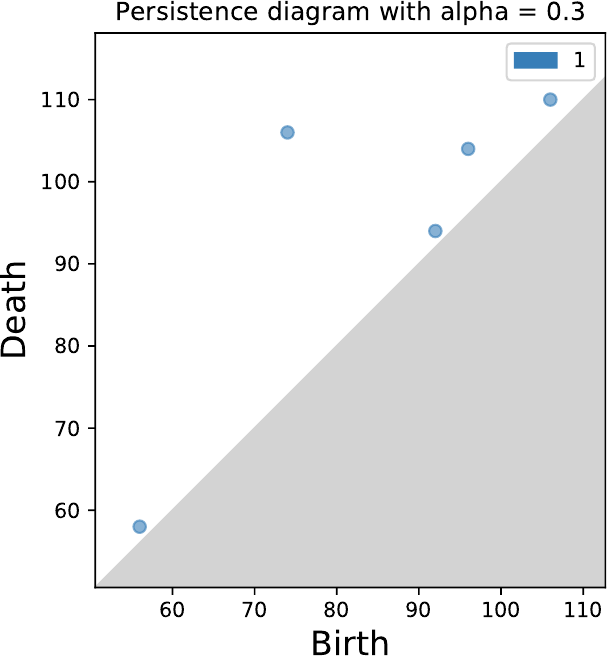}\\
  \includegraphics[width=1.8in]{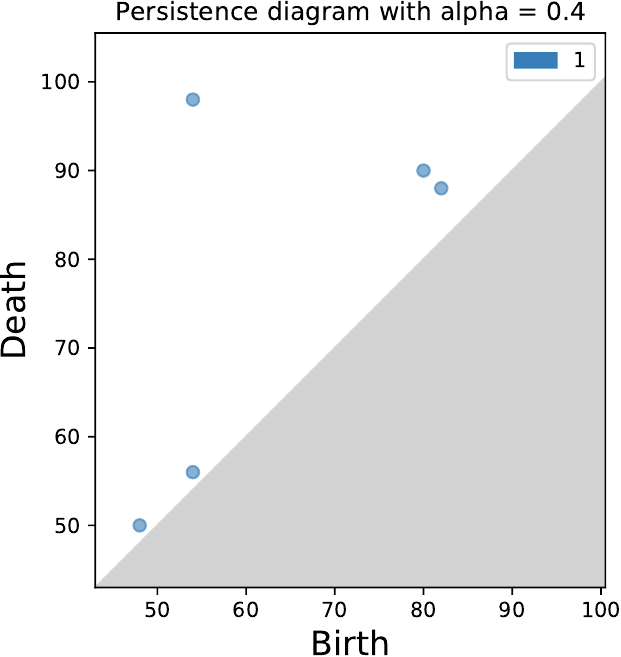}&
  \includegraphics[width=1.8in]{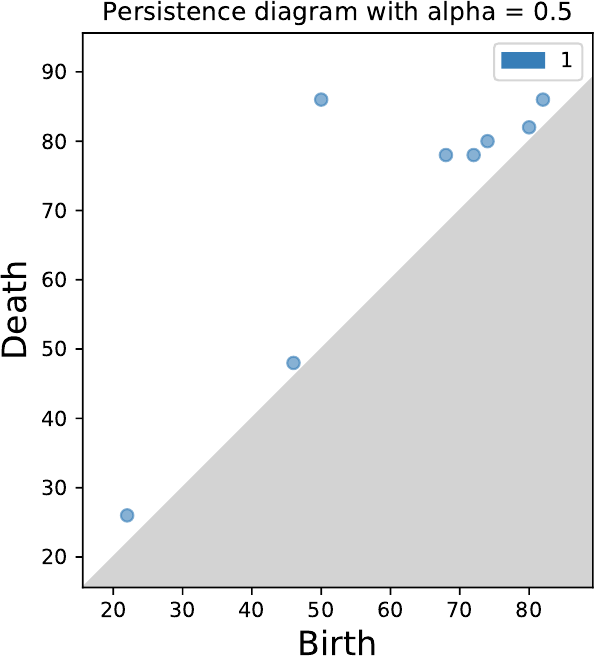}&
  \includegraphics[width=1.8in]{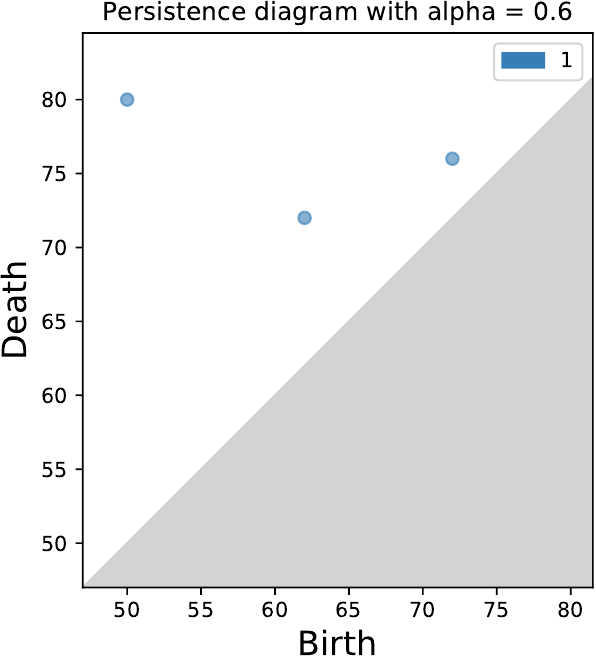}\\
  \includegraphics[width=1.8in]{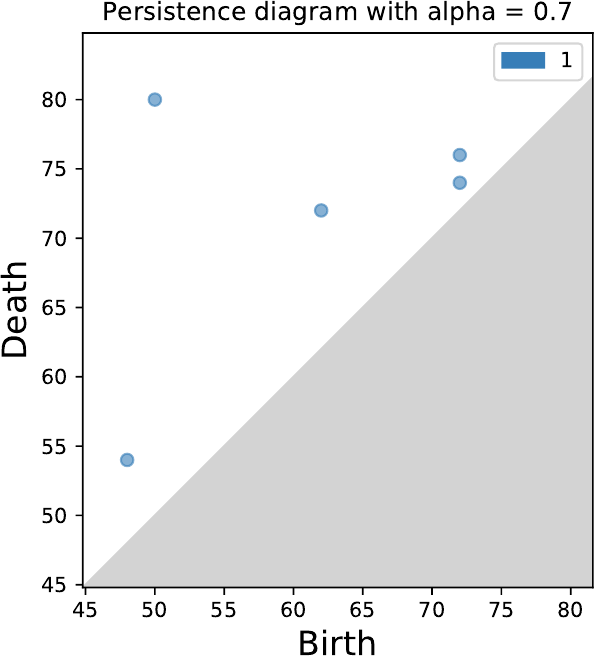}&
  \includegraphics[width=1.8in]{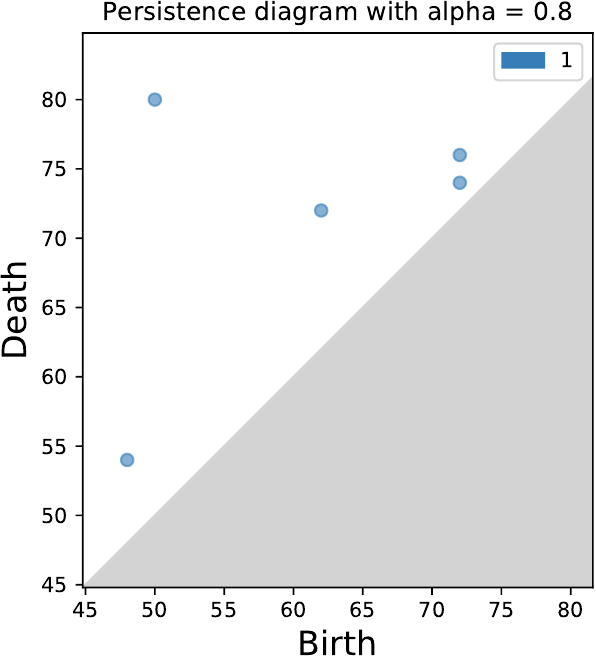}&
  \includegraphics[width=1.8in]{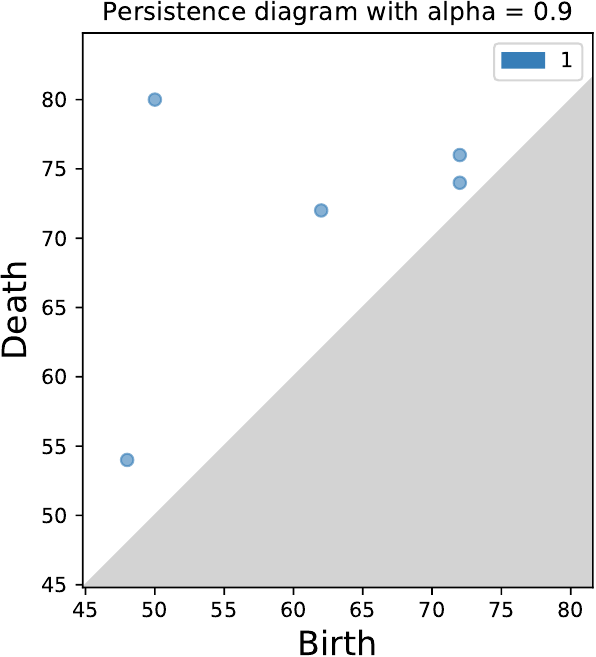}
  \end{tabular}
  \caption{\label{tab:noisycircleBf} PDs of BF applied to noisy circle PCD for $\alpha=0.1$--$0.9$.}
  \end{center}
\end{figure}

\begin{figure}[htp!]
  \begin{center}
  \begin{tabular}{m{4.5cm} m{4.5cm} m{4.5cm} m{4.5cm}}
  \includegraphics[width=1.8in]{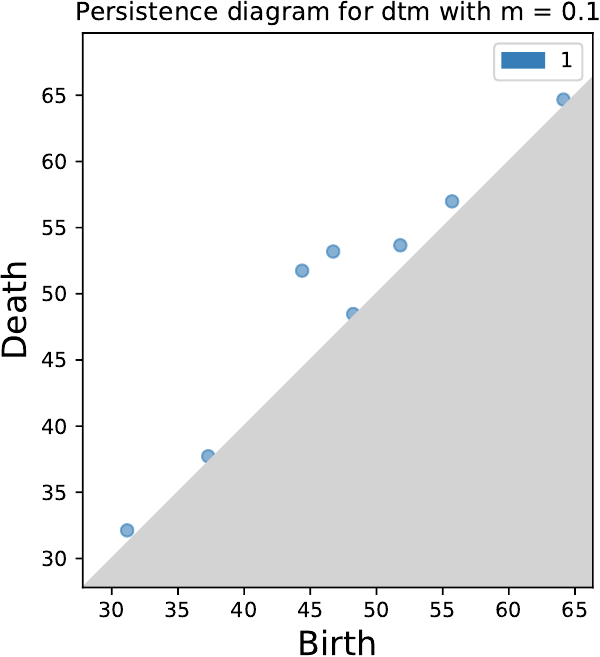}&
  \includegraphics[width=1.8in]{{noisyThinCircleResults/dtmPersistence_with_m_0.2_crop}.pdf}&
  \includegraphics[width=1.8in]{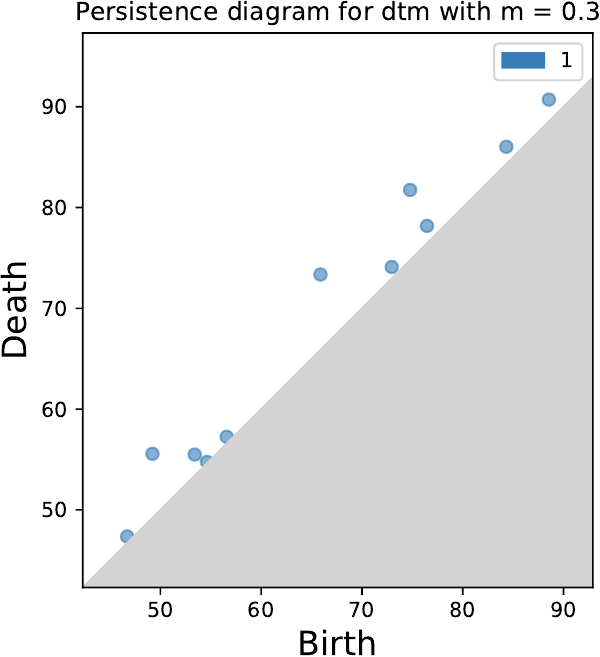}\\
  \includegraphics[width=1.8in]{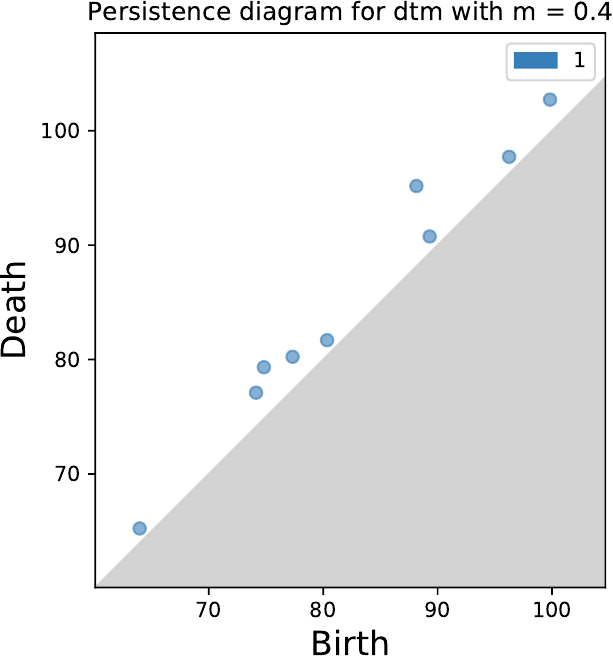}&
  \includegraphics[width=1.8in]{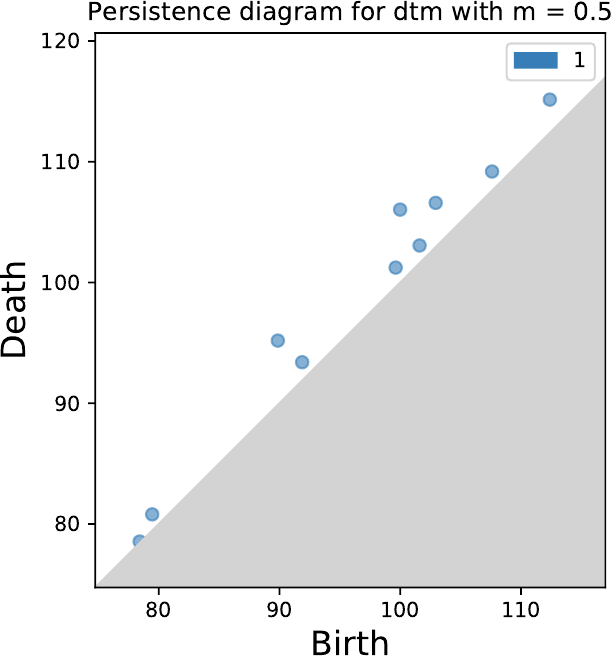}&
  \includegraphics[width=1.8in]{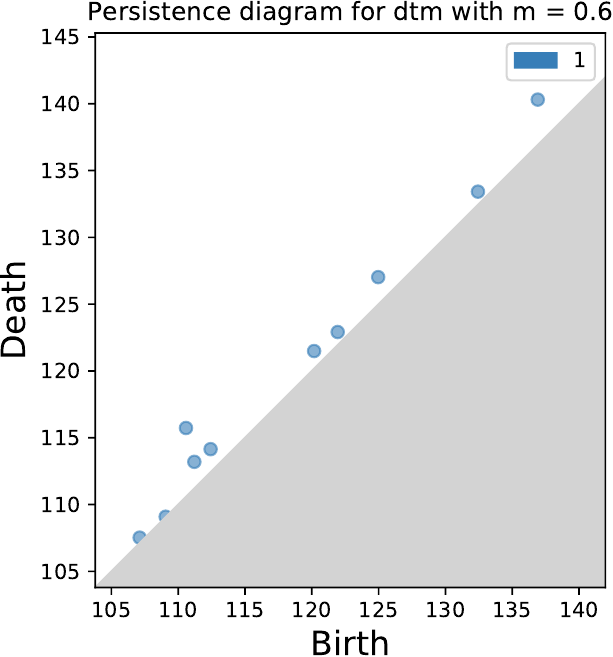}\\
  \includegraphics[width=1.8in]{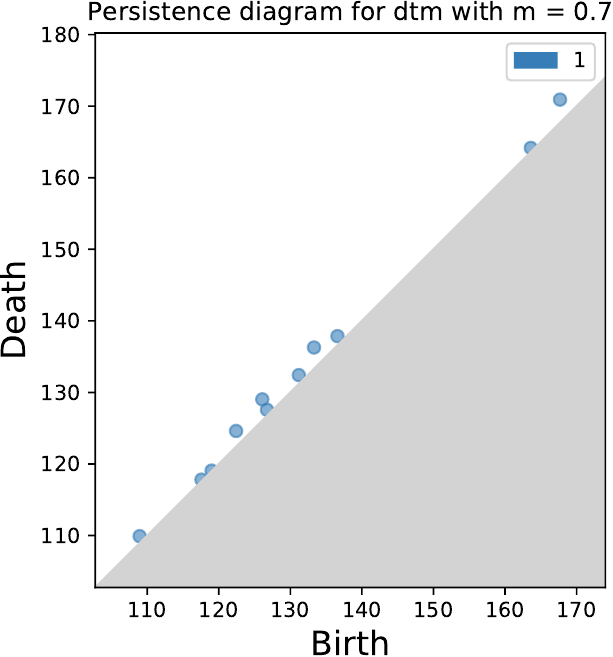}&
  \includegraphics[width=1.8in]{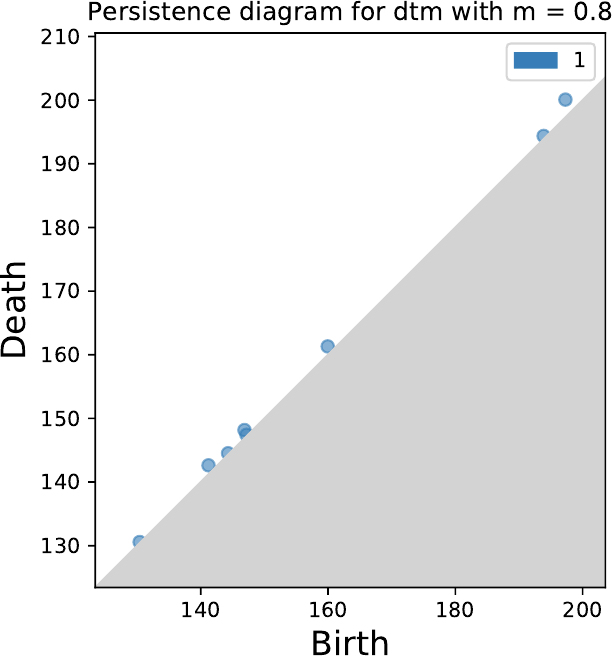}&
  \includegraphics[width=1.8in]{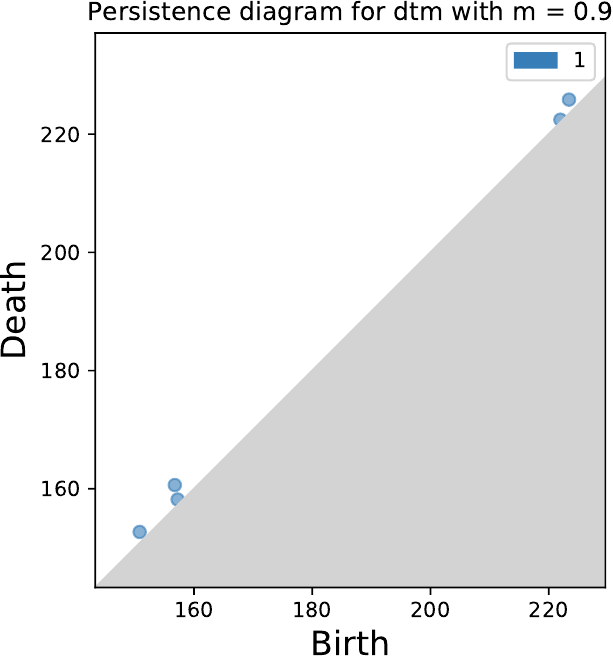}
  \end{tabular}
  \caption{\label{tab:noisethincircleDtm} PDs of DTM applied to the noisy ellipse PCD for $m=0.1$--$0.9$.
  }
  \end{center}
\end{figure}

\begin{figure}[htp!]
  \begin{center}
  \begin{tabular}{m{4.5cm} m{4.5cm} m{4.5cm} m{4.5cm}}
  \includegraphics[width=1.8in]{{noisyThinCircleResults/bfPersistencewith_alpha_0.1_crop}.pdf}&
  \includegraphics[width=1.8in]{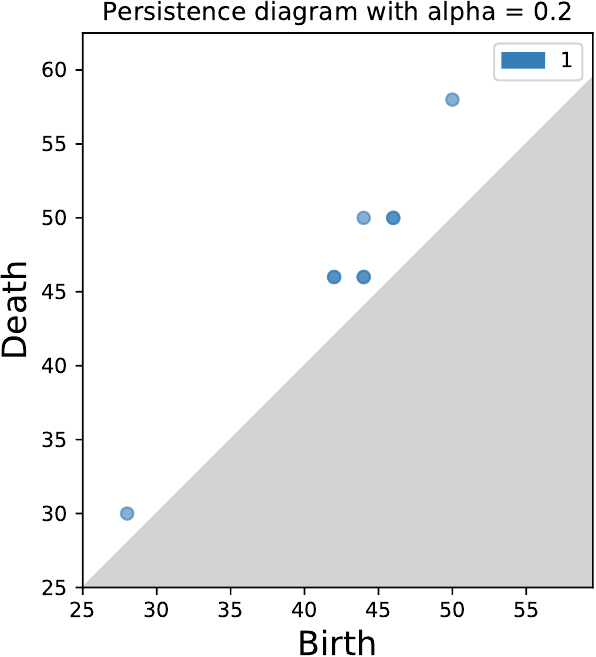}&
  \includegraphics[width=1.8in]{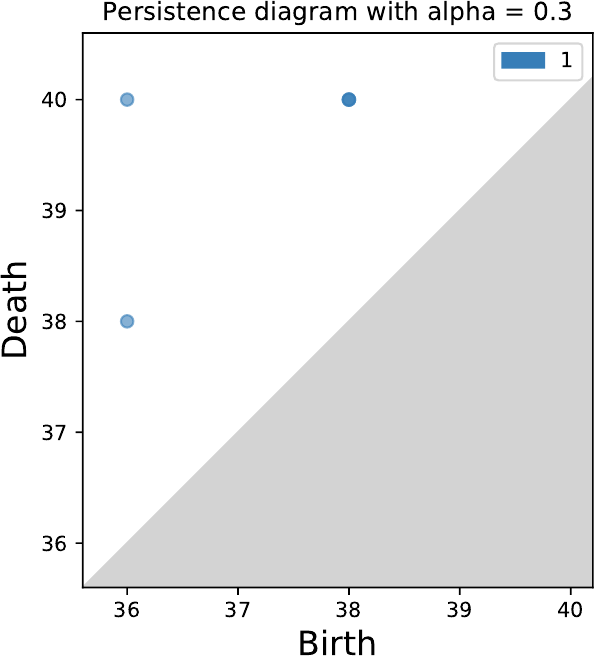}\\
  \includegraphics[width=1.8in]{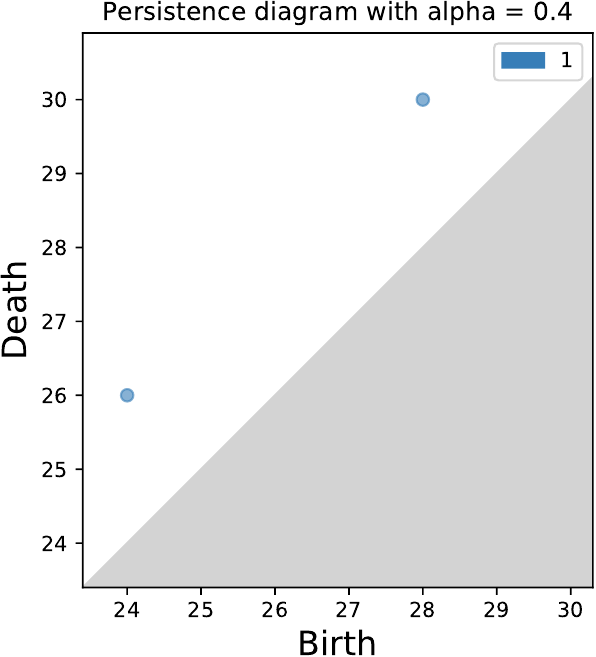}&
  \includegraphics[width=1.8in]{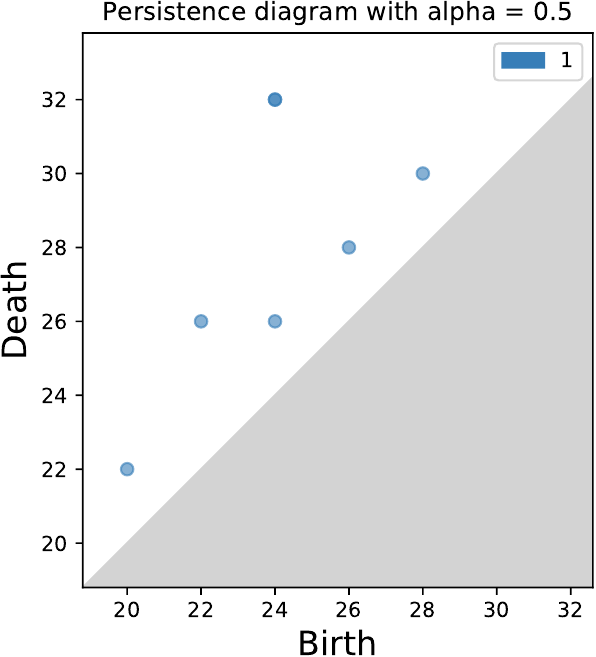}&
  \includegraphics[width=1.8in]{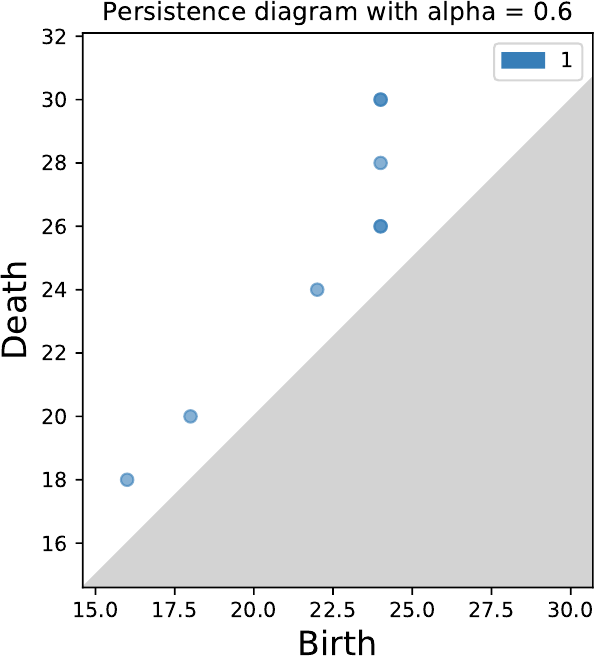}\\
  \includegraphics[width=1.8in]{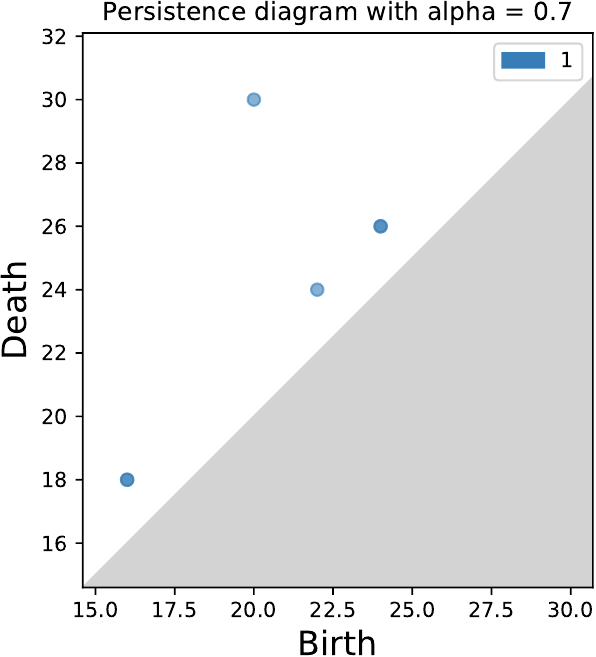}&
  \includegraphics[width=1.8in]{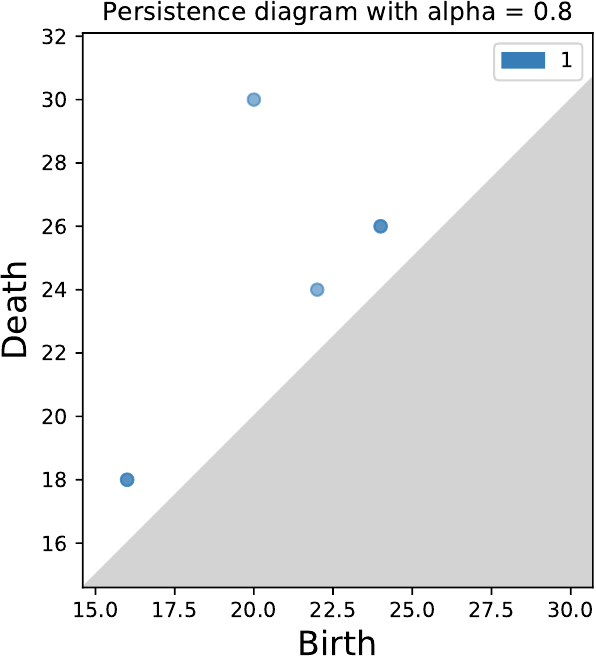}&
  \includegraphics[width=1.8in]{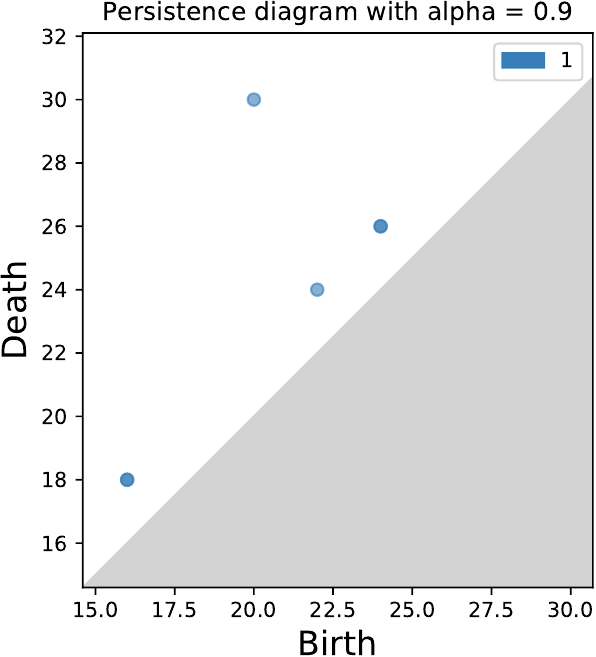}
  \end{tabular}
  \caption{\label{tab:noisethincircleBf} PDs of BF applied to noisy ellipse PCD for $\alpha=0.1$--$0.9$.}
  \end{center}
\end{figure}

\begin{figure}[htp!]
  \begin{center}
  \begin{tabular}{m{4.5cm} m{4.5cm} m{4.5cm} m{4.5cm}}
  \includegraphics[width=1.8in]{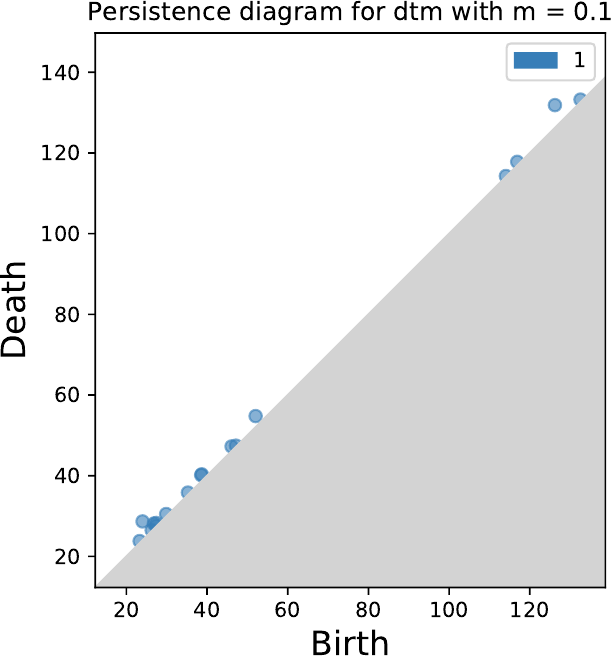}&
  \includegraphics[width=1.8in]{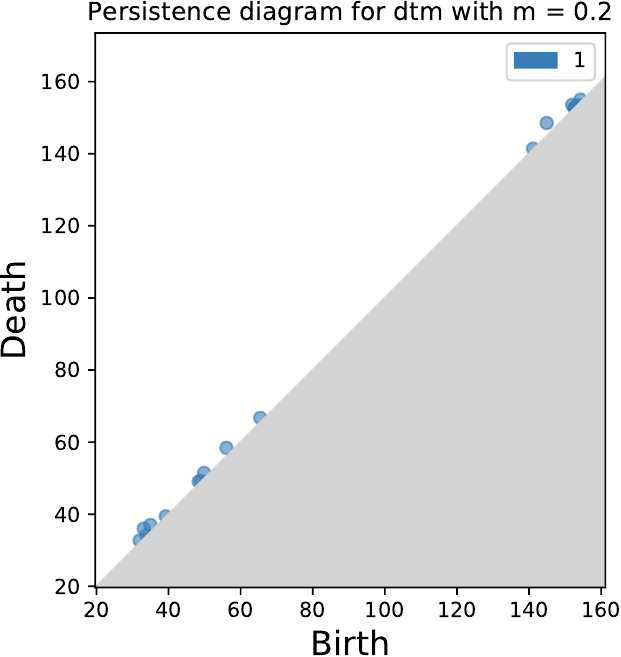}&
  \includegraphics[width=1.8in]{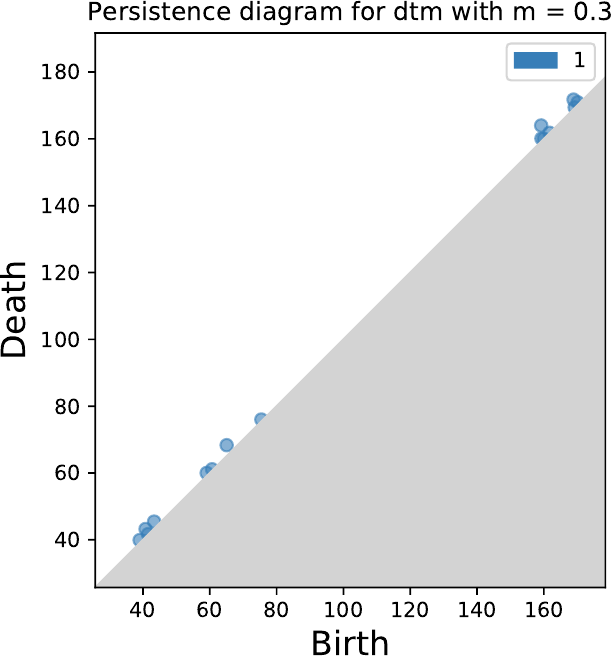}\\
  \includegraphics[width=1.8in]{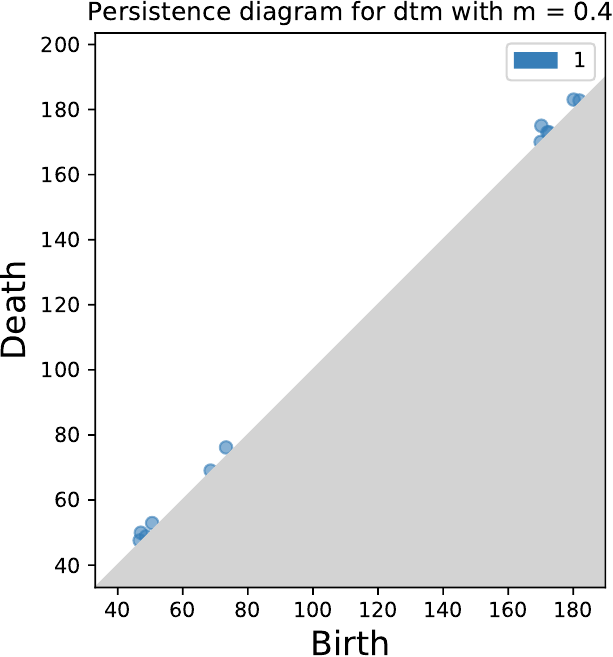}&
  \includegraphics[width=1.8in]{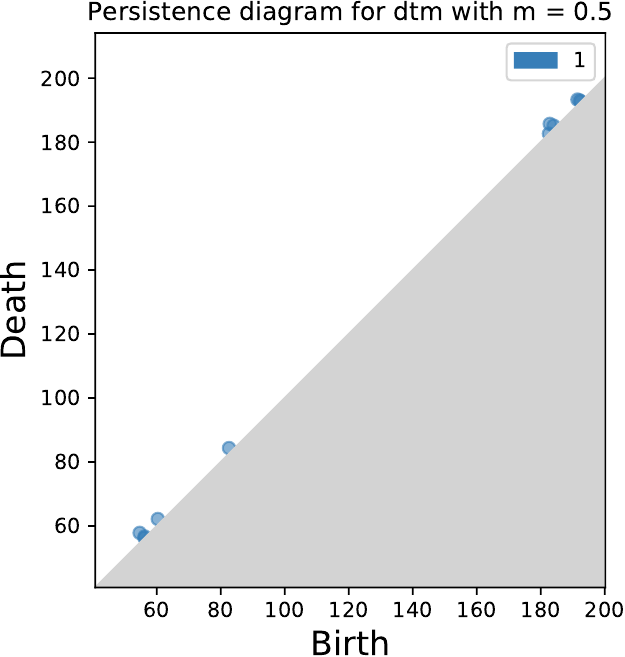}&
  \includegraphics[width=1.8in]{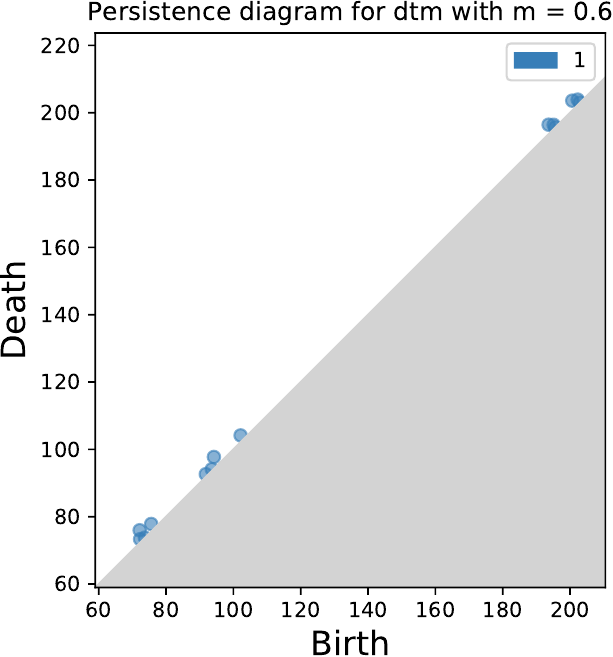}\\
  \includegraphics[width=1.8in]{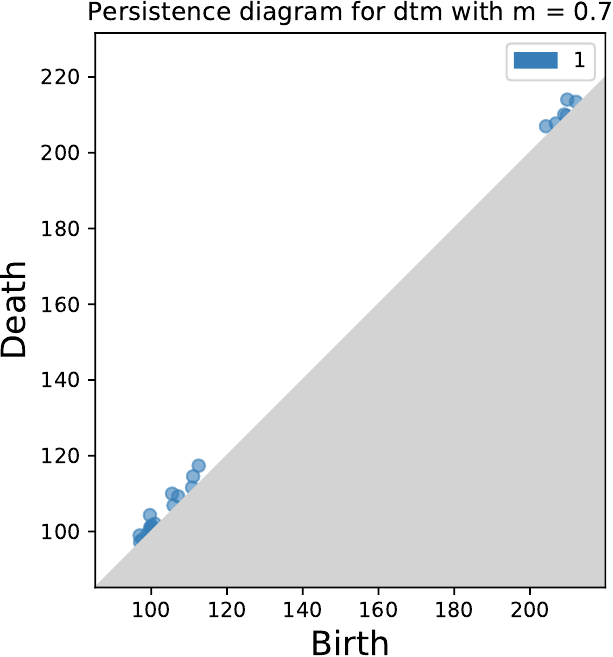}&
  \includegraphics[width=1.8in]{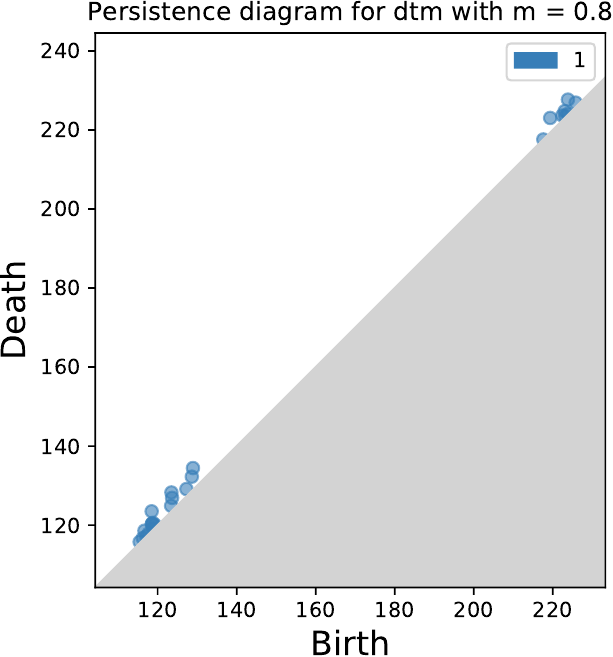}&
  \includegraphics[width=1.8in]{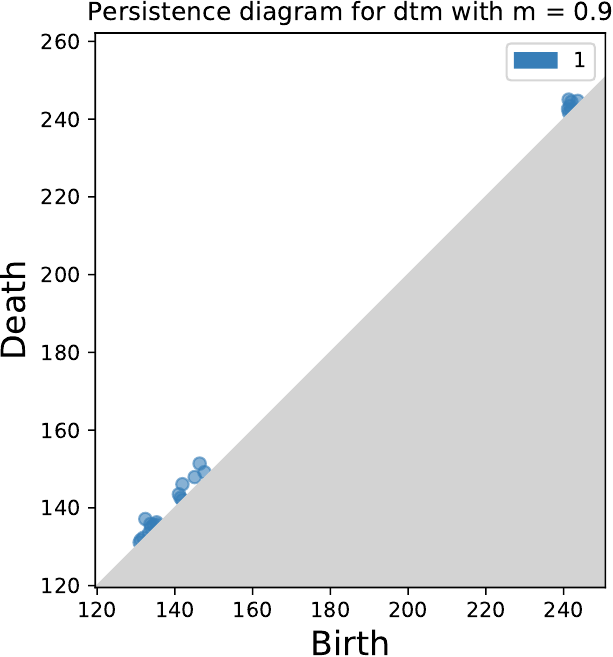}
  \end{tabular}
  \caption{\label{tab:noisecirclewithcentralclusterDtm} PDs of DTM applied to the circle with central cluster PCD for $m=0.1$--$0.9$.}
  \end{center}
\end{figure}

\begin{figure}[htp!]
  \begin{center}
  \begin{tabular}{m{4.5cm} m{4.5cm} m{4.5cm} m{4.5cm}}
  \includegraphics[width=1.8in]{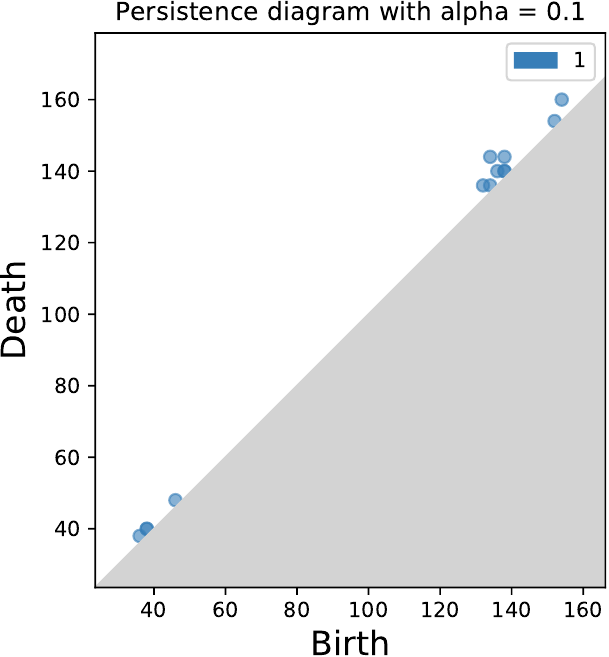}&
  \includegraphics[width=1.8in]{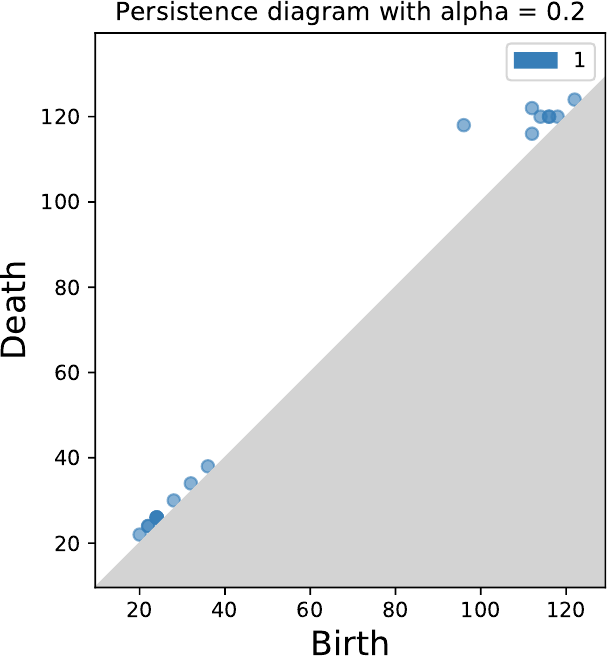}&
  \includegraphics[width=1.8in]{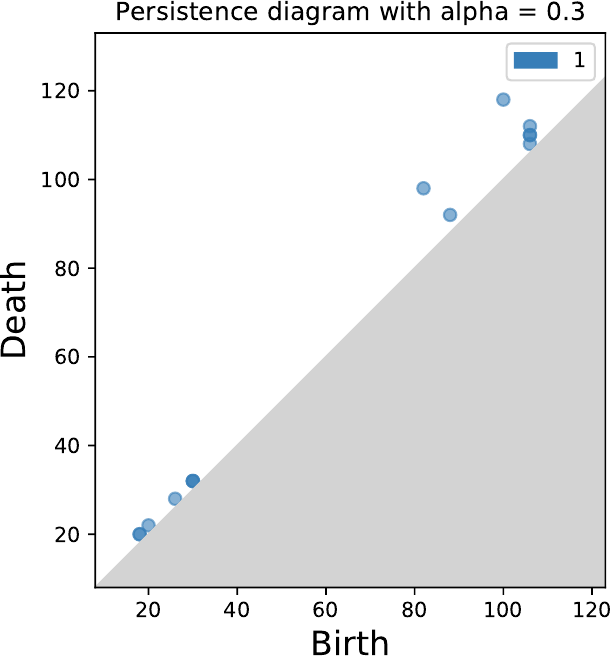}\\
  \includegraphics[width=1.8in]{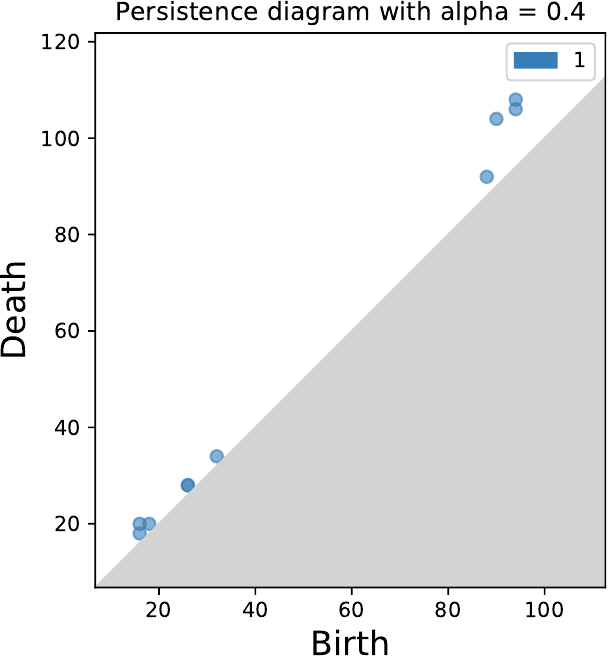}&
  \includegraphics[width=1.8in]{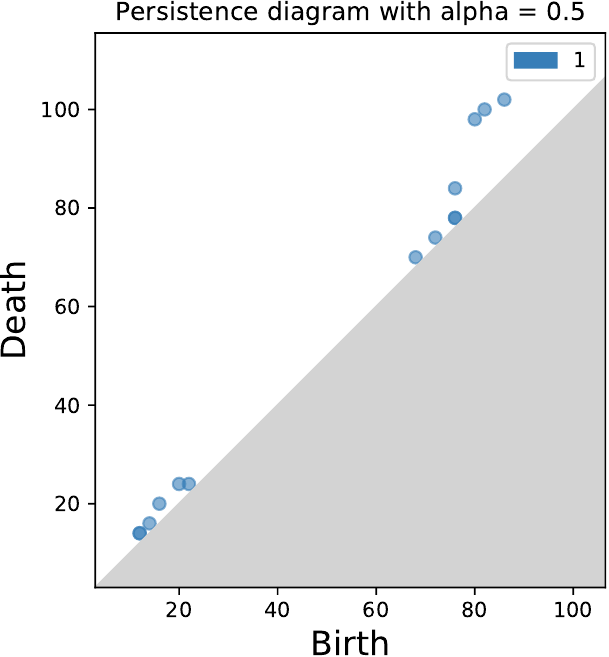}&
  \includegraphics[width=1.8in]{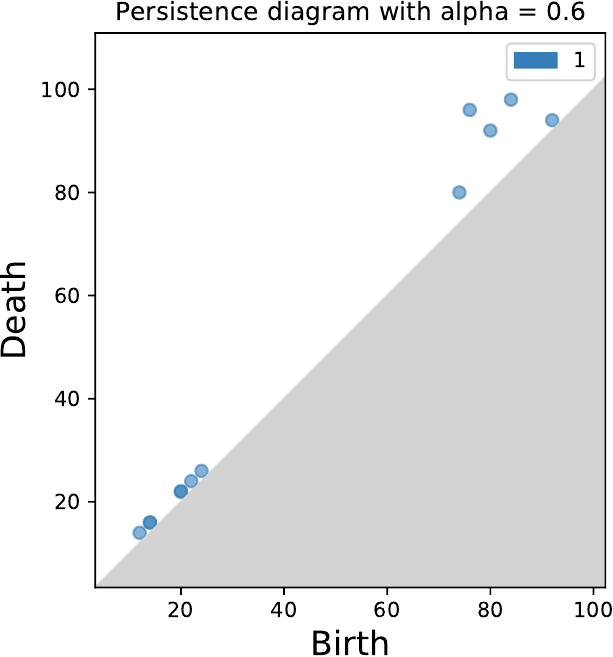}\\
  \includegraphics[width=1.8in]{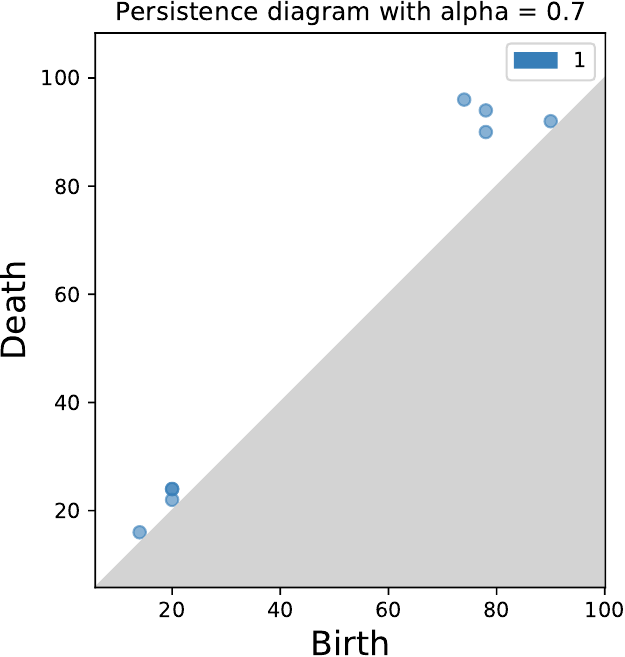}&
  \includegraphics[width=1.8in]{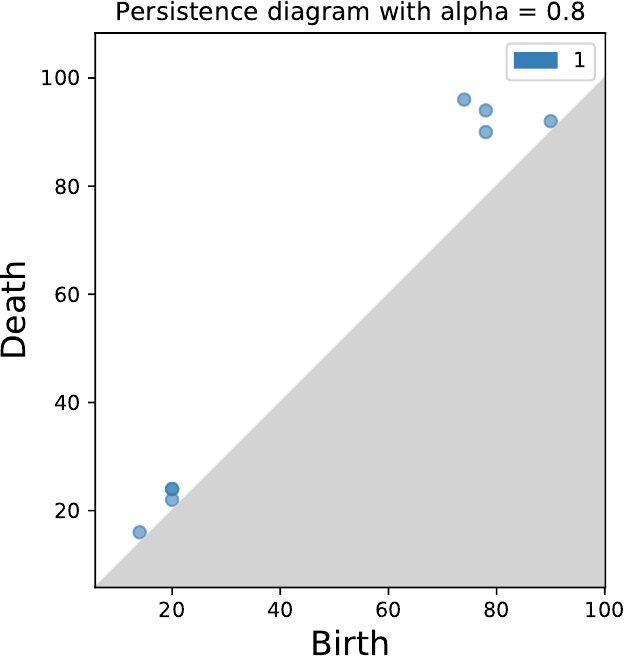}&
  \includegraphics[width=1.8in]{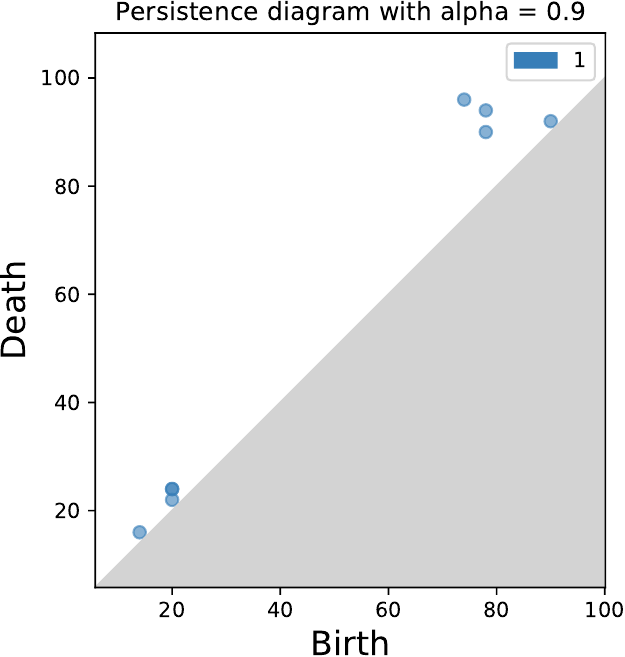}
  \end{tabular}
  \caption{\label{tab:noisecirclewithcentralclusterBf} PDs of BF applied to the circle with central cluster PCD for $\alpha=0.1$--$0.9$.}
  \end{center}
\end{figure}

\begin{figure}[htp!]
  \begin{center}
  \begin{tabular}{m{4.5cm} m{4.5cm} m{4.5cm} m{4.5cm}}
  \includegraphics[width=1.8in]{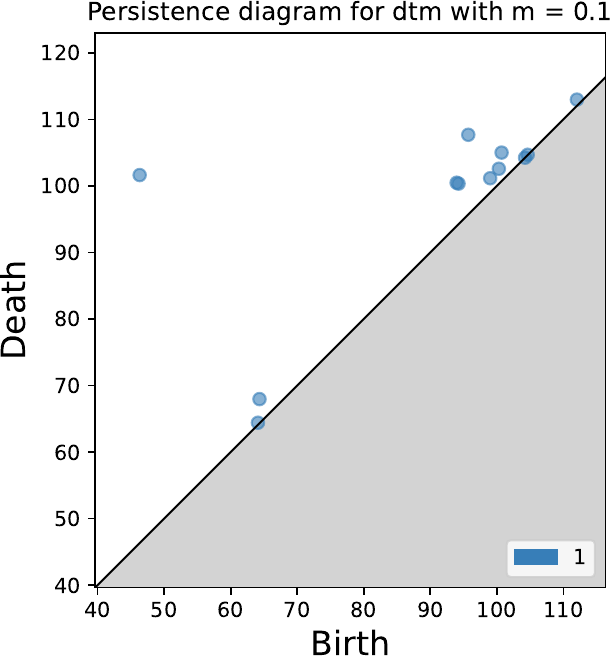}&
  \includegraphics[width=1.8in]{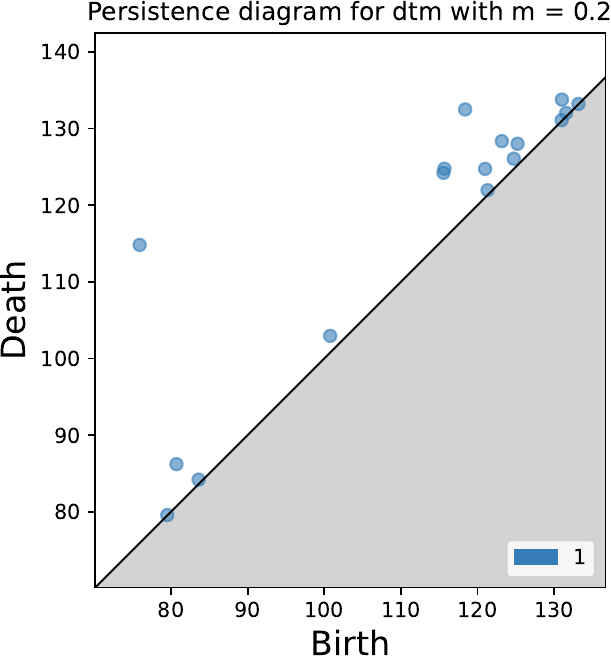}&
  \includegraphics[width=1.8in]{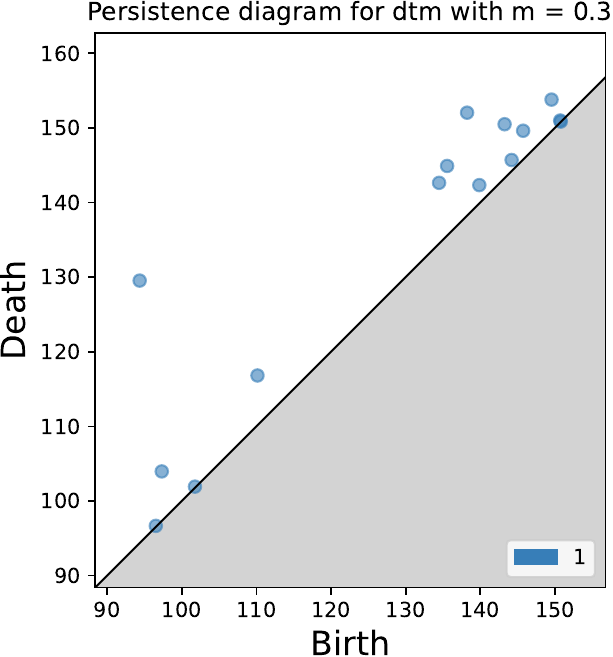}\\
  \includegraphics[width=1.8in]{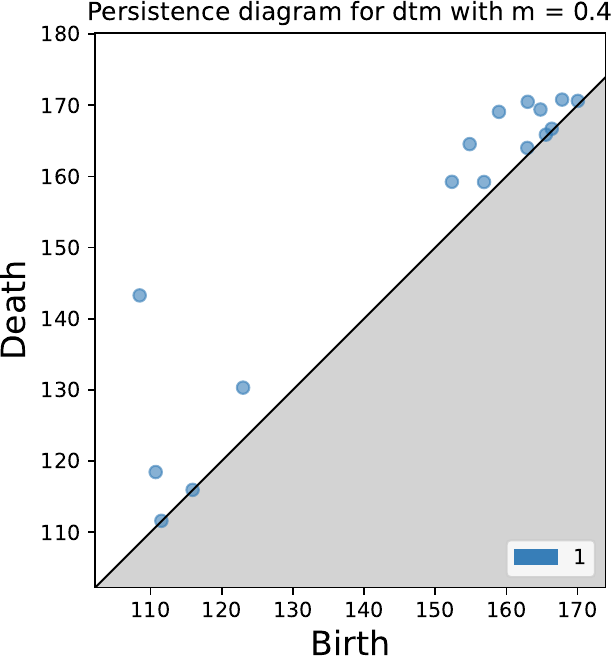}&
  \includegraphics[width=1.8in]{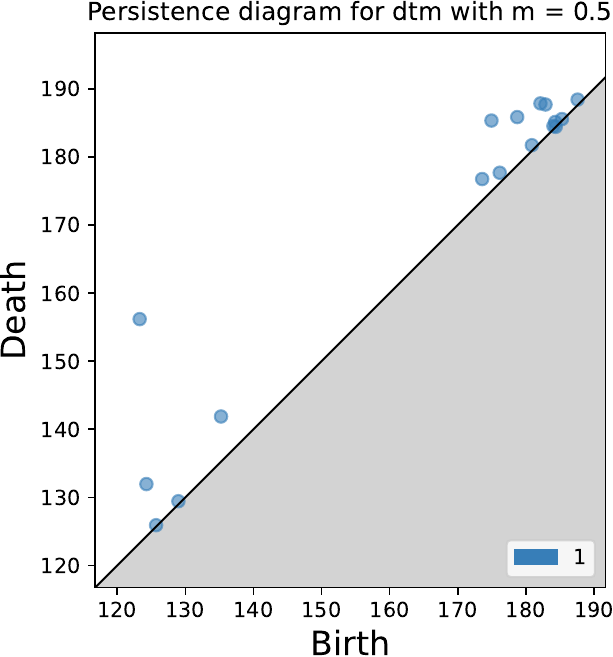}&
  \includegraphics[width=1.8in]{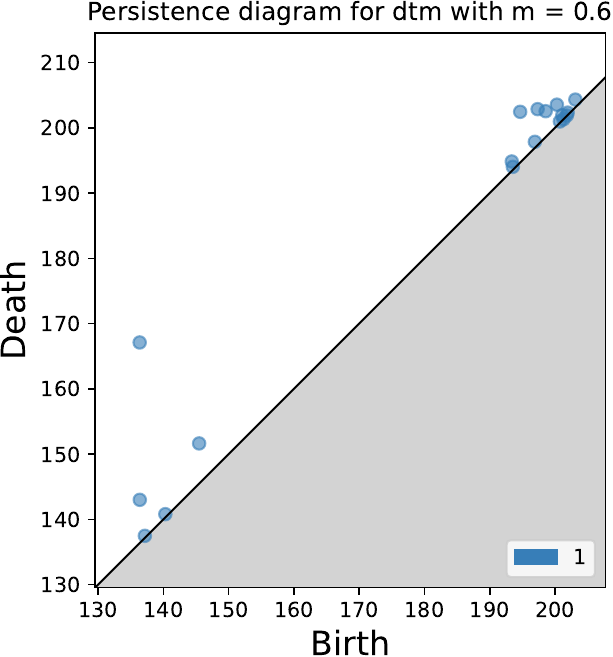}\\
  \includegraphics[width=1.8in]{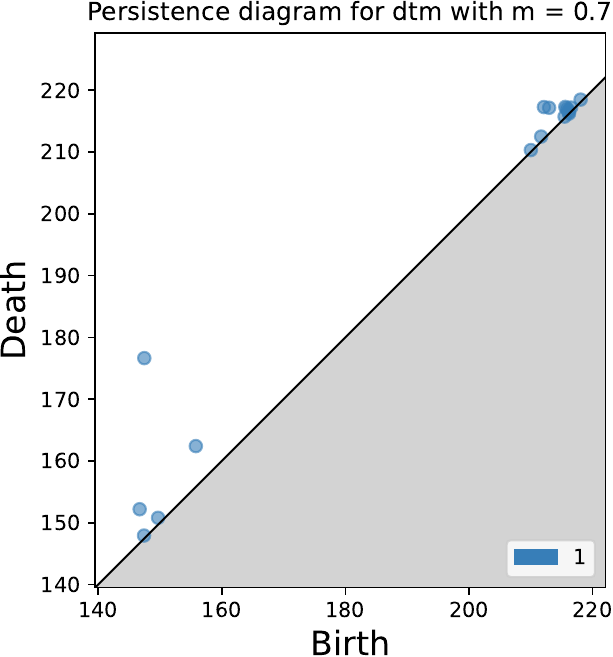}&
  \includegraphics[width=1.8in]{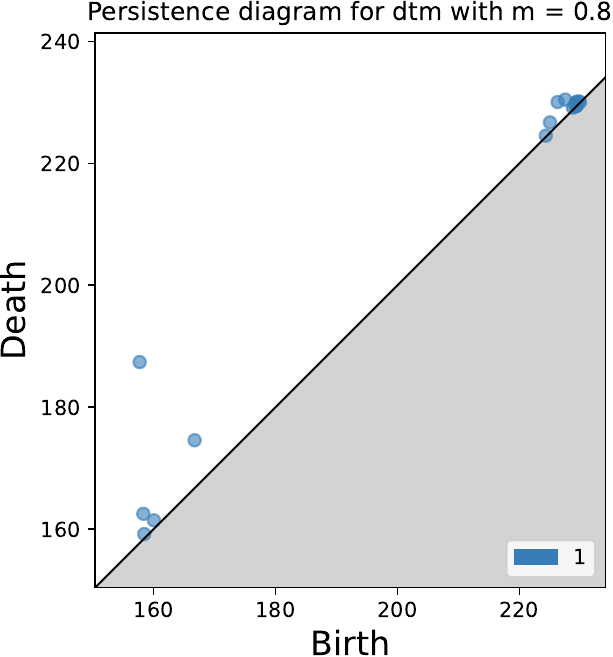}&
  \includegraphics[width=1.8in]{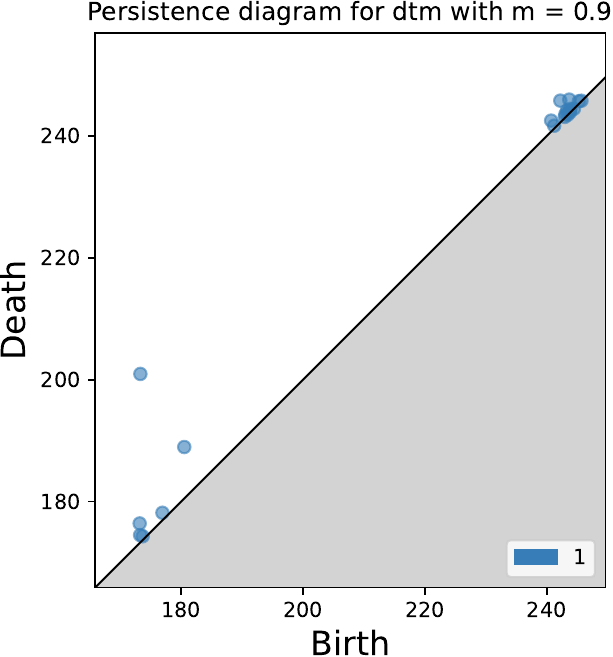}
  \end{tabular}
  \caption{\label{tab:noisecirclewithcentralclusterwithsquareDtm} PDs of DTM applied to the concentric circles with noise PCD for $m=0.1$--$0.9$.}
  \end{center}
\end{figure}

\begin{figure}[htp!]
  \begin{center}
  \begin{tabular}{m{4.5cm} m{4.5cm} m{4.5cm} m{4.5cm}}
  \includegraphics[width=1.8in]{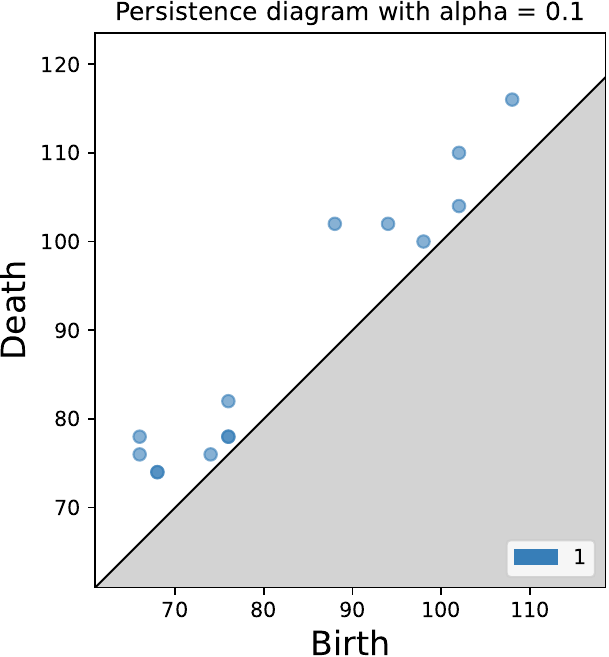}&
  \includegraphics[width=1.8in]{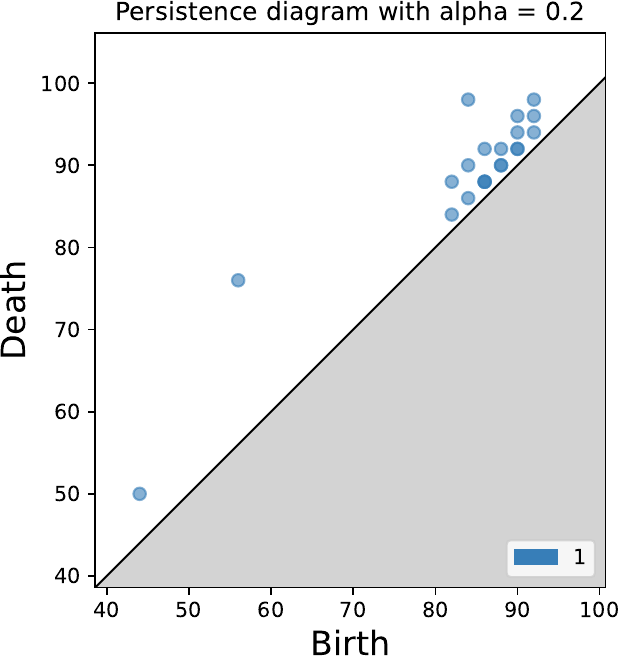}&
  \includegraphics[width=1.8in]{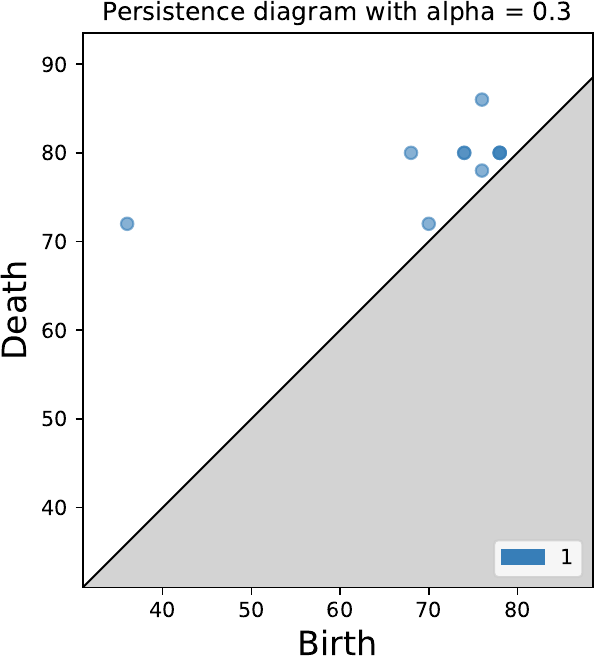}\\
  \includegraphics[width=1.8in]{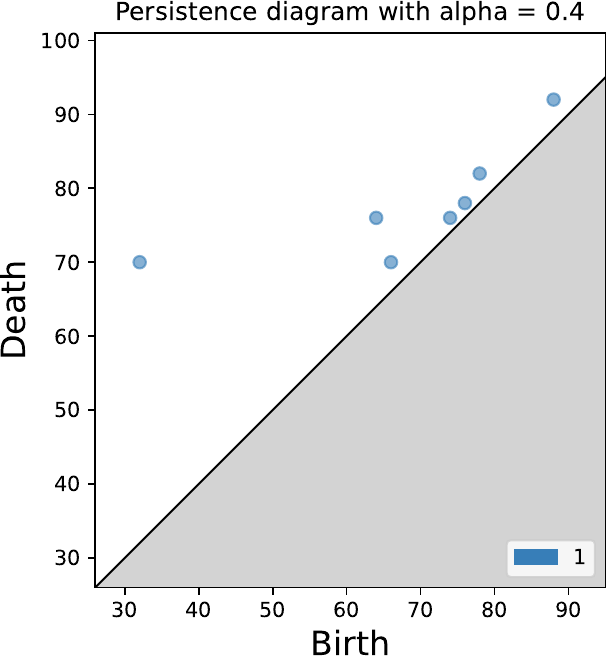}&
  \includegraphics[width=1.8in]{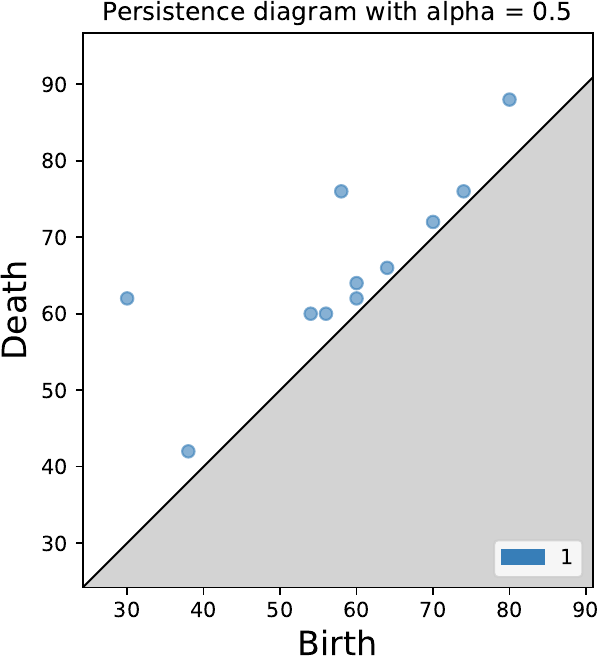}&
  \includegraphics[width=1.8in]{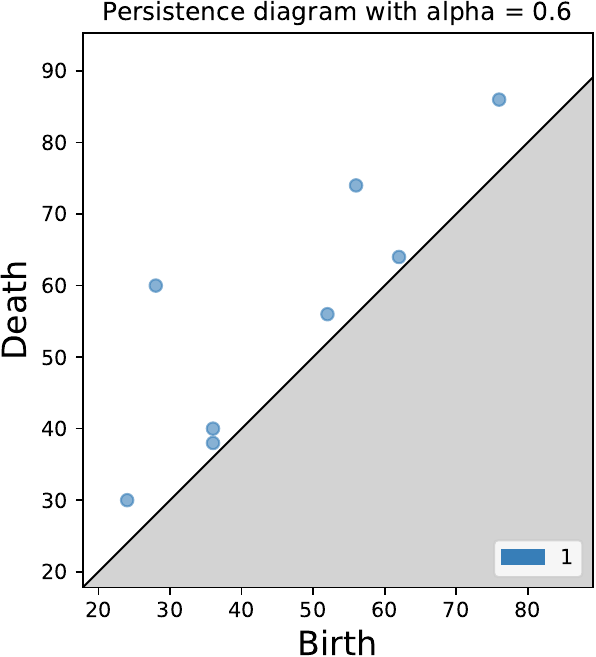}\\
  \includegraphics[width=1.8in]{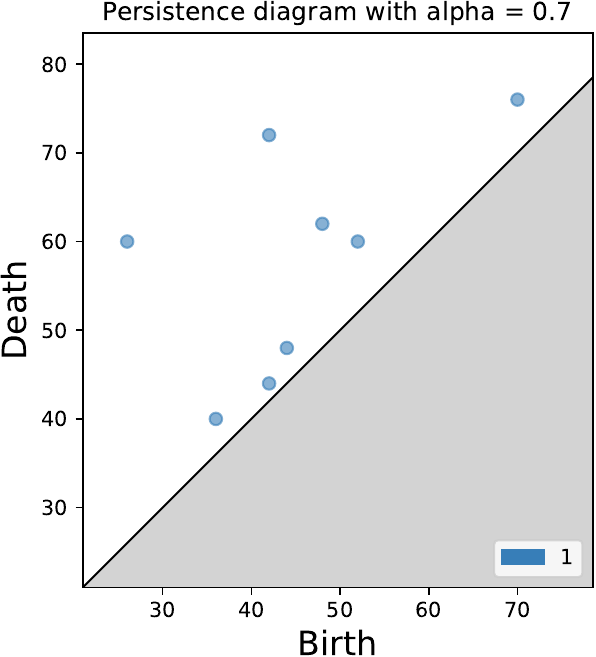}&
  \includegraphics[width=1.8in]{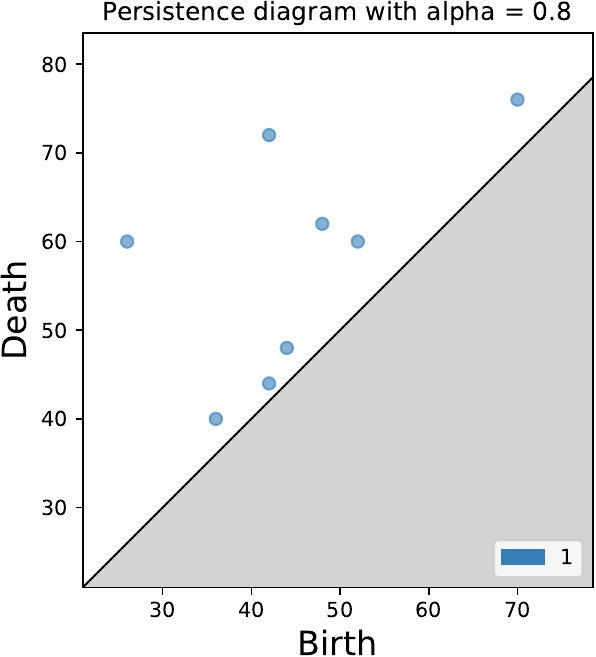}&
  \includegraphics[width=1.8in]{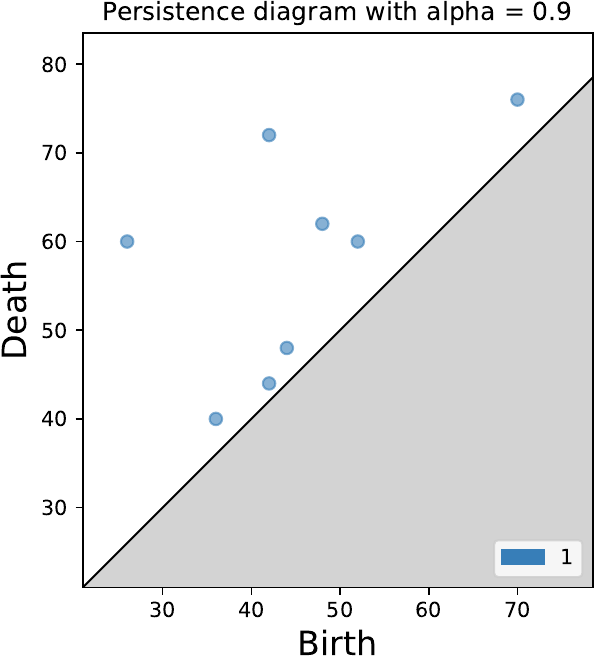}
  \end{tabular}
  \caption{\label{tab:noisecirclewithcentralclusterwithsquareBf} PDs of BF applied to concentric circles with noise PCD for $\alpha=0.1$--$0.9$.}
  \end{center}
\end{figure}

\end{appendices}